%% file: karri.tex
\documentclass[twoside,leqno,twocolumn]{article}

\usepackage{ltexpprt}
\usepackage{hyperref}

\usepackage{amsmath}
\usepackage{amssymb}
\usepackage{gensymb}
\usepackage{cleveref}
\usepackage[noend]{algpseudocode}
\usepackage{algorithm}
\usepackage{xspace}
\usepackage{siunitx}
\usepackage{makecell}
\usepackage{diagbox}
\usepackage{environ}
\usepackage[numbers]{natbib}
\usepackage{multirow}
\usepackage{microtype}
\usepackage{doi}
\usepackage{booktabs}
\usepackage{savefnmark}

\algnewcommand\algorithmicforeach{\textbf{for each}}
\algdef{S}[FOR]{ForEach}[1]{\algorithmicforeach\ #1\ \algorithmicdo}

\usepackage{tikz}
\usetikzlibrary{quotes, babel, arrows.meta, calc, matrix,arrows,automata,positioning,math,plotmarks,decorations.markings,shapes.misc,patterns}

\input{macros.tex}

\input{shared_macros.tex}

\newcolumntype{R}{>{$}r<{$}} 

\crefname{@theorem}{theorem}{theorems}
\Crefname{@theorem}{Theorem}{Theorems}

\begin{document}
\pagestyle{plain}
\newcommand\relatedversion{}
\renewcommand\relatedversion{\thanks{This work is based on an earlier technical report~\cite{Laupichler2023}}} 

\title{\Large Fast Many-to-Many Routing for Dynamic Taxi Sharing with Meeting Points\relatedversion}
\author{
  Moritz Laupichler\thanks{Karlsruhe Institute of Technology, Institute for Theoretical Informatics, Algorithm Engineering.}\saveFN{\affiliation} \and
  Peter Sanders\useFN{\affiliation}}

\date{}

\maketitle

\begin{abstract} 
\small\baselineskip=9pt 
We introduce an improved algorithm for the dynamic taxi sharing problem, i.e. a dispatcher that schedules a fleet of shared taxis as it is used by services like UberXShare and Lyft Shared.
We speed up the basic online algorithm that looks for all possible insertions of a new customer into a set of existing routes, we
generalize the objective function, and we efficiently support a large number of possible pick-up and drop-off locations.
This lays an algorithmic foundation for taxi sharing systems with higher vehicle occupancy -- enabling greatly reduced cost and ecological impact at comparable service quality.
We find that our algorithm computes assignments between vehicles and riders several times faster than a previous state-of-the-art approach.
Further, we observe that allowing meeting points for vehicles and riders can reduce the operating cost of vehicle fleets by up to $15\%$ while also reducing rider wait and trip times. 
\end{abstract}


\section{Introduction}
\label{sec:Introduction}
Current transportation systems are largely based on a combination of individual transport (often with heavy, polluting cars that consume a lot of energy and space) and public transportation that is often slow, inconvenient, and underdeveloped.
Recently, \emph{taxi sharing} systems that intelligently control fleets of shared taxi-like vehicles have garnered a lot of attention as a promising means of interpolating between the economical and ecological benefits of public transportation and the convenience and flexibility of individually used cars.
The traffic engineering community has extensively studied the possible advantages of such systems in a large number of simulation studies~\cite{Bischoff2017,Kuehnel2023,Agatz2011,Fagnant2014,Wilkes2021,Zwick2022} and real-world field tests~\cite{Gilibert2020,Kostorz2021,Yu2017,Weckstroem2018,Jokinen2019,Chichung2008,Gargiulo2015,Zhu2021}.
A widespread adaptation of taxi sharing is expected to coincide with an increased demand for sustainable personal transportation~\cite{Song2021,Yu2017} and the availability of autonomously piloted vehicles~\cite{Fagnant2014,Fagnant2018,Milakis2017,Duarte2018,Badue2021}.

A main issue of current such systems is that the potential for shared rides is usually limited as each additional stop made to pick up or drop off a rider causes delays for other riders.
This makes taxi sharing less attractive and makes larger capacity vehicles infeasible.

We focus on the question of how riders can use local transportation (e.g., walking, bicycles or scooters) to reach a pickup or dropoff location (\emph{meeting point}) that causes less delay for a vehicle~\cite{Stiglic2015,Aissat2014}, may be shared with other customers~\cite{Stiglic2015,Kaan2013}, and may alleviate concerns of privacy for riders~\cite{Goel2016PuPs}. 
This acts as a first step towards a hierarchy of personal transportation consisting of local transportation, taxi sharing, and public transit, promising economical and ecological benefits compared to current transportation systems.

Our starting point is the dynamic taxi sharing dispatcher by \citet{Buchhold2021}.
It uses one-to-many routing based on bucket contraction hierarchies (BCHs) \cite{Knopp2007,Geisberger2012} to efficiently compute the best feasible assignments of riders to vehicles. 
This is a crucial step for handling large fleets in real time and computing realistic simulations of such systems in transportation research.

We introduce the \myalg (\underline{Ka}rlsruhe \underline{R}apid \underline{Ri}desharing) algorithm that extends the dispatcher with the possibility of performing the pickup and dropoff of a rider not at fixed locations but at meeting points which can be any location close to the rider's origin and destination.
The algorithm computes assignments of riders to vehicles (including locations for the pickup and dropoff) that minimize the cost of the assignment w.r.t. a global objective function and the current vehicle routes.
We adapt this objective function to the new scenario with meeting points by incorporating rider trip times and overheads for local travel to and from meeting points.

Finding not only the best vehicle for a request but an optimal combination of a vehicle, a pickup location, and a dropoff location leads to a much larger number of possible assignments. 
To determine the best assignment, we need to solve a number of many-to-many routing problems between vehicle locations and \emph{all} possible meeting points.
We use BCH queries to address this issue and propose novel speedup techniques both for general purpose bucket based queries and for the specific case of localized sources or targets.
We find that these techniques are also applicable for faster routing in a scenario without meeting points. 

Our experimental evaluation uses realistic data sets to evaluate the efficiency of these measures.
In a scenario without meeting points, our implementation is several times faster than the state-of-the-art dispatcher~\cite{Buchhold2021}.
For multiple meeting points, our routing techniques are up to three orders of magnitude faster than a na\"ive extension of previous techniques.
We also give first indications that meeting points can reduce the operating costs of a taxi fleet by up to $15\%$ without increasing rider wait times or trip times.  
A closer investigation of possible effects on the transport system is left to future work likely in cooperation with application experts.

\mysubsection{Related Work}
\label{subsec:related_work}
Taxi sharing and related problems are well studied in transportation research.
We summarize existing solution approaches and research into the effect of meeting points on such systems.

\parheader{Taxi Sharing}
\emph{Taxi sharing} (also called \emph{ride pooling}) is the problem of dispatching rider requests asking to go from an origin to a destination location to a fleet of taxi-like vehicles while adhering to rider time constraints like a latest possible arrival time.
The dispatcher tries to find assignments of riders to vehicles that optimize an objective function such as the total operation time of all vehicles.

Taxi sharing can be seen as a special case of the well studied \emph{Dial-a-Ride problem (DARP)}~\cite{Cordeau2007,Ho2018}.
Most research on taxi sharing deals with the static variant of the problem where all rider requests are known in advance, including their individual time constraints.
The static problem is known to be NP-complete~\cite{Shuo2013,Savelsbergh1985}.
Small problem instances can be solved optimally using integer programming~\cite{Cordeau2006,AlonsoMora2017,Hosni2014,Cordeau2007}.
Other solutions sacrifice optimality for better performance using meta-heuristics like simulated annealing~\cite{Lin2012,Jung2016}, GRASP~\cite{Santos2015}, or the artificial bee colony algorithm~\cite{Zhan2022}.

We study the dynamic taxi sharing problem.
Here, the dispatcher is informed about requests as they come in and has to assign riders to vehicles in that order without knowing future requests.
Though there is increasing interest in dynamic ridesharing with stochastic information about future rider demand~\cite{Manna2014,Schilde2011,Ma2019}, we stick to the traditional agnostic view~\cite{Psaraftis1980}.
Thus, we are concerned with local decision heuristics that try to find a best assignment for each request, attempting to minimize the negative impact on the global objective function or \emph{cost} of the chosen assignment.
Note that the routing techniques discussed in this paper are also applicable to static and stochastic dispatchers as they, too, need to compute many-to-many shortest path queries. 

A lot of work on dynamic taxi sharing focuses on enumerating assignments and assessing their feasibility w.r.t. the riders' time constraints~\cite{Jaw1986,Madsen1995,Hunsaker2002}.
For this, the dispatcher needs to know the extent of the vehicle detours made to service the new rider.
Oftentimes, these detours are simply assumed to be known~\cite{Psaraftis1980,Jaw1986,Madsen1995,Horn2002,Hunsaker2002,Haell2012,Ota2017,Lotze2022}.
However, finding the shortest paths that comprise the detours in the road network poses a major time overhead and can become a bottleneck for the performance of a taxi sharing dispatcher.

Some recent works acknowledge this overhead by first employing filtering heuristics (e.g. based on geodesic distances~\cite{Bischoff2017,Horni2016} or spatial indices~\cite{Shuo2013,Huang2014,Ma2015}) to find a small set of candidate assignments s.t. shortest path queries only have to be executed for these candidates.
These heuristic dispatchers use varying shortest path algorithms as a black box, ranging in efficiency from Dijkstra's algorithm~\cite{Dijkstra1959} to hub labeling~\cite{Abraham2011}.

\Citet{Buchhold2021} employ a more involved approach by using the time constraints of already assigned riders to prune bucket contraction hierarchy searches~\cite{Knopp2007,Geisberger2012}, a state-of-the-art one-to-many shortest path algorithm.
This allows the shortest path algorithm itself to act as a filter of feasible assignments, efficiently computing both a set of candidate vehicles that is guaranteed to contain the best assignment and the required shortest paths. 
The algorithm is also equipped to work with bucket based searches in customizable contraction hierarchies~\cite{Dibbelt2016} which allow for fast readjustment of travel times in the road network caused by changing traffic conditions.

\parheader{Ride Matching}
In taxi sharing, the vehicles' only purpose is to service riders.
In opposition to this, the closely related \emph{ride matching} or \emph{ride sharing} problem assumes that each driver is a private entity with their own origin, destination and time constraints~\cite{Furuhata2013,Agatz2012,Herbawi2012}.

Ride matching is largely faced with the same challenges as taxi sharing.
Static solutions can be found with integer programming~\cite{Agatz2011, Stiglic2015} and branch-and-bound algorithms~\cite{Bistaffa2015}, or approximated with evolutionary algorithms~\cite{Herbawi2012}.
Approaches for the dynamic variant include locality-constrained greedy matching algorithms~\cite{Goel2016Privacy,Pelzer2015} and the application of static solutions for buffered sets of requests~\cite{Herbawi2012,Agatz2011}.
As with taxi sharing, BCHs may be suited to compute vehicle detours~\cite{Geisberger2010}.
For overviews on ride matching, we refer to~\cite{Furuhata2013} and~\cite{Agatz2012}.

\parheader{Meeting Points in Taxi Sharing}
\emph{Meeting points} allow riders to be picked up and dropped off at locations close to their origin and destination, respectively.
This requires the rider to walk a small distance but potentially reduces the cost of an assignment.

Taxi sharing with meeting points has started garnering attention only recently with publications on this problem variant first appearing in 2021.
Most works that we are aware of focus on the positive effects of meeting points on the operation costs and service quality of taxi sharing systems while largely not addressing the added complexity of the problem.

\citet{Fielbaum2021} and~\citet{Mounesan2021} independently extend a previous ILP formulation of the static taxi sharing problem~\cite{AlonsoMora2017} with meeting points.
The authors of both works come to the conclusion that it is a hard problem, akin to the Generalized Traveling Salesman Problem, to find the best route along with the best meeting points even for a fixed vehicle and set of requests.
Both works evaluate the impact of meeting points on static taxi sharing in an experimental study using the road network of Manhattan.
For this small road network (about 10000 edges), both works pre-compute and store all-pairs shortest path distances.
Meeting points are found to substantially increase the rate of requests that can be serviced within certain rider wait time and trip time limits, and to simultaneously decrease the total vehicle operation time.
\citet{Fielbaum2021} and~\citet{Mounesan2021} report that meeting points increase the time needed to find an optimal solution to the taxi sharing ILP by factors of about ten and five, respectively.
Due to these large running times, static solutions are not suited for real-time production systems.

\citet{Lotze2022} explore \emph{stop pooling}, a restricted form of meeting points, for dynamic taxi sharing.
Though the authors' experiments use a simplistic model with uniformly distributed requests and euclidian distances, they find compelling evidence that stop pooling can help to break the traditional trade-off between vehicle operation times and rider trip times. 

We are aware of only one work that considers the scalability of taxi sharing with meeting points on realistic metropolitan scale road networks.
For this purpose,~\citet{Mounesan2021} develop the dynamic taxi sharing simulation \starsplus next to their aforementioned ILP formulation for the static problem variant.
The dynamic dispatcher processes each request as soon as it is issued and greedily chooses the best vehicle and meeting points according to the current vehicle routes. 
A distance cache based on a partition of the road network is designed to facilitate pre-computing and storing all-pairs shortest path distances within reasonable memory limits for larger road networks.
Using the distance cache, \starsplus is shown to be able to answer requests on the road network of all five boroughs of New York City in about 10\ms{} with a fleet of 10000 vehicles, which represents an increase by a factor of about six compared to using no meeting points.
The authors find a reduction in the trip time and total vehicle miles traveled compared to a scenario without meeting points.
The experiments do not evaluate the query times for the new distance cache, inhibiting a comparison to similar techniques like transit node routing~\cite{Bast2007} or customizable route planning~\cite{Delling2017}.

Though \starsplus addresses the same problem as our work, there are some important differences in both the model details and the solution approach.

Firstly, \starsplus assumes a fixed hard limit to each rider's wait time and trip time.
Since many vehicles can then be excluded from consideration for most requests, this speeds up the dispatching process.
However, the hard time limits may cause requests to be rejected entirely.
Meanwhile, we only apply penalties for trips that break such limits which allows us to service every request but necessitates a more careful consideration of a larger set of candidate vehicles.

Secondly, \starsplus chooses the best meeting points for a given request heuristically while we explicitly develop methods that consider all combinations of potential pickup and dropoff locations and find the best one.

Thirdly, \starsplus uses a pre-computed static distance cache, whereas a focus of this work is the computation of shortest paths on-the-fly during the dispatching process.
This has the advantage of being more dynamic:
As travel times in road networks frequently change, e.g. due to congestion, the underlying data structures for on-the-fly shortest-path queries (e.g. a contraction hierarchy) can be updated in seconds while a distance cache has to be reconstructed from scratch.
The authors of \starsplus report a running time in the order of hours for this which would be too slow for updates in a production system.

\parheader{Meeting Points in Ride Matching}
Beyond this limited amount of work on meeting points in taxi sharing, several publications have studied meeting points on the closely related ride matching problem and have found a positive impact on the quality of matches.

\Citet{Li2016} show that it is NP-hard to find optimal meeting points for a set of ride matching requests even when considering only a single vehicle.
The authors present multiple dynamic programming based solution algorithms for a slightly relaxed problem variant.

\Citet{Goel2016PuPs,Goel2016Privacy} find a fixed set of possible meeting points that allow privacy-aware ride matching.
The meeting points cover the road network in such a way that any rider can communicate a small subset of meeting points close to their origin location to the driver without allowing them to identify the rider's true origin location.

\Citet{Aissat2014} consider optimizing existing ride matches with meeting points.
Assuming an existing match between a driver and a single rider both traveling from individual origin to destination locations, they find the optimal intermediate pickup and dropoff locations, reducing the detour and total travel cost. 

\Citet{Stiglic2015} evaluate meeting points as a way to improve the efficacy of a static ride matching system.
The authors suggest data reduction rules for feasible matches between drivers and groups of riders that share potential meeting points.
To make the problem more tractable, the authors limit every driver to two extra stops (one for pickups and one for dropoffs) and only allow a small number of select vertices to be used as meeting points. 
Nonetheless, in a simulation study (with Euclidian distances), the authors observe that meeting points improve the chance of finding feasible matches and reduce the total distance driven.

\mysubsection{Paper Overview}
After a more detailed problem statement in~\cref{sec:Problem_Statement} we introduce basic notation and techniques
in~\cref{sec:Preliminaries}. In~\cref{sec:conceptual_changes}, we examine necessary changes to the taxi sharing model caused by the introduction of meeting points.~\Cref{sec:the_algorithm,} gives an overview on the \myalg algorithm before~\cref{sec:ordinary_op_and_pbns_insertions,sec:pickup_after_last_stop_insertions,sec:dropoff_after_last_stop_insertions} describe our many-to-many routing techniques and their application in taxi sharing.~\Cref{sec:experimental_evaluation} contains our experimental evaluation on large scale realistic input instances. Finally,~\cref{sec:conclusion} summarizes ideas for future extensions of dynamic taxi sharing.


\section{Problem Statement}
\label{sec:Problem_Statement}
This section describes the formal foundations for the dynamic taxi sharing problem considered by our approach.

\parheader{Road Network}
We consider a \emph{road network} to be a directed graph $G=(V,E)$ where edges represent road segments and vertices represent intersections.
Every edge $e = (v,w) \in E$ has a travel time $\ell(e)=\ell(v,w)$.
We denote the \emph{shortest path distance} (i.e. travel time) from a vertex $v$ to a vertex $w$ by $\dist(v,w)$.

\parheader{Vehicle, Stop}
Our algorithm has access to a fleet $F$ of \emph{vehicles}.
The current \emph{route} $R(\nu) = \langle s_0(\nu), \dots, s_{\numstops{\nu}}(\nu) \rangle$ of a vehicle $\nu$ is a sequence of \emph{stops} scheduled for the vehicle.
The vehicle's current location is always somewhere between its previous (or current) stop $s_0(\nu)$ and its next stop $s_1(\nu)$.
We update the routes accordingly as vehicles reach stops or are assigned new stops.
Thus, $\numstops{\nu} = |R(\nu)| - 1$ is the number of stops that the vehicle yet has to visit.
Each stop $s$ is mapped to a vertex $\stoploc(s) \in V$ in the graph.
Given a sufficiently clear context, we may write $s_i$ instead of $s_i(\nu)$ and only $s_i$ instead of $\stoploc(s_i)$.

\parheader{Request}
In our scenario, the dispatcher receives ride requests and immediately assigns them to vehicles.
A request $r=(\orig, \dest, \treq)$ has an origin location $\orig \in V$, a destination location $\dest \in V$ and a time $\treq$ at which the request is issued. 
We do not allow pre-booking, i.e. the request time is also the earliest possible departure time. 

\parheader{Meeting Points}
We assume that riders can reach meeting points in their vicinity using local transportation such as walking or cycling.
We represent the paths accessible to this mode of transportation in a road network $\Gpsg = (\Vpsg, \Epsg)$.
For any request $r$, any two subsets of $\Vpsg \cap \Vveh$ can be chosen as the sets of potential pickup and dropoff locations for $r$.

Let $\distpsg(u,v)$ denote the \emph{rider shortest path distance} between vertices $u$ and $v$ in $\Gpsg$. 
We use a default set of pickup locations (or \emph{pickups}) $\Prho(r)$ that contains all eligible vertices that the rider can reach in $\Gpsg$ from $\orig(r)$ within a time radius $\rho$, i.e. $\Prho(r) \gets \{ p \in \Vpsg \cap \Vveh \mid \distpsg(\orig(r), p) \le \rho \}$.
Similarly, our default set of dropoff locations (or \emph{dropoffs}) is defined as $\Drho(r) \gets \{ p \in \Vpsg \cap \Vveh \mid \distpsg(d, \dest(r)) \le \rho \}$.
We collectively refer to the pickups and dropoffs of $r$ as the \emph{meeting points} of $r$.
Let $\numpickups(r) = |\Prho(r)|$ and $\numdropoffs(r) = |\Drho(r)|$.
We call a pair of pickup and dropoff a \emph{PD-pair} and the distance between a pickup and a dropoff a \emph{PD-distance}.
The radius $\rho$ is a model parameter.
For the sake of simplicity, we use the same $\rho$ for every request but the model also permits varying $\rho$ with each request.

If the context allows it, we omit $r$ in the notation of the terms defined above.
Further, for $p \in \Prho$ and $d \in \Drho$, we use $\distpsg(p)$ and $\distpsg(d)$ as shorthand for $\distpsg(\orig, p)$ and $\distpsg(d, \dest)$.
For ease of notation in the rest of this work, we use the term ``walking'' to mean any mode of local transportation for riders.

\parheader{Insertion}
For each request $r$, our dispatcher finds an \emph{insertion} of a pickup and dropoff of $r$ into any vehicle's route s.t. the cost of that insertion according to a cost function is minimized. 
We formalize an insertion as a tuple $(r, p, d, \nu, i, j)$ indicating that vehicle $\nu$ picks up request $r$ at pickup location $p \in \Prho$ immediately after stop $s_i$ and drops off $r$ at dropoff location $d \in \Drho$ immediately after stop $s_j$ with $0 \le i \le j \le \numstops{\nu}$ .

\mysubsection{Cost Function and Constraints}
\label{subsec:cost_function}
The cost $c(\ins)$ of an insertion $\ins = (r, p, d, \nu, i, j)$ represents the associated vehicle operation cost and the rider service quality in a linear combination of the form
\begin{equation}\label{eq:cost_function}
\begin{split}
  c(\ins) = \; & \tdetour(\ins) + \tripweight \cdot (\ttrip(\ins) + \ttripplus(\ins)) + \\
              & \walkweight \cdot \twalk(\ins) + \cwaitvio(\ins) + \ctripvio(\ins) \text{.}
\end{split}
\end{equation}
Here, the \emph{added vehicle operation time} $\tdetour(\ins)$ describes the time that vehicle $\nu$ needs for the detour it makes to accommodate the pickup at $p$ and dropoff at $d$ in its route.
The \emph{trip time} $\ttrip(\ins)$ denotes the time that passes between the issuing of request $r$ ($\treq(r)$) and the arrival of the rider at their destination $\dest(r)$, including waiting and walking times.
The detours made by $\nu$ may increase this trip time for existing riders of $\nu$.
The \emph{added trip time} $\ttripplus(\ins)$ is the sum of these increases for all affected riders. 
The \emph{walking time} $\twalk(\ins)$ represents how long the rider needs to move from their origin to the pickup and from the dropoff to their destination using local transportation. 
We weight the importance of rider trip times and walking times relative to the vehicle operation time using the model parameters $\tau$ and $\omega$.
Note that if $\tripweight = \walkweight = 0$, our cost function is equivalent to the one used in the baseline dispatcher by~\citet{Buchhold2021}.
We defer the formal definitions of the terms mentioned so far to~\cref{sec:conceptual_changes}.

The remaining terms $\cwaitvio(\ins)$ and $\ctripvio(\ins)$ describe penalties for violating constraints on the service quality of the new rider. 
We consider a total of four constraints originally put forth in~\cite{Buchhold2021}.
After the insertion, the following must hold:
First, the \emph{occupancy} of $\nu$ must never exceed a fixed capacity.
Second, the vehicle must still reach its last stop before a fixed \emph{end of its service time}.
Third, every rider already assigned to $\nu$ must still be picked up at their pickup stop within a \emph{maximum wait time} $\twaitmax$.
Fourth, every rider $\hat{r}$ already assigned to $\nu$ must still arrive at their destination within a \emph{maximum trip time} $\ttripmax(\hat{r}) = \alpha \cdot \distveh(\orig(\hat{r}), \dest(\hat{r})) + \beta$ where $\distveh(\orig(\hat{r}), \dest(\hat{r}))$ is the direct vehicle travel time from the origin to the destination of $\hat{r}$.
The values $\twaitmax$, $\alpha$ and $\beta$ are model parameters.

All four constraints are hard constraints w.r.t. requests already assigned to $\nu$.
If $\ins$ breaks a hard constraint, we set the cost to $\infty$.
For the request $r$ to be inserted, we treat the wait time and trip time constraints as soft constraints, i.e. violating them leads to the cost penalties $\cwaitvio(\ins)$ and $\ctripvio(\ins)$.
Assume, the rider is picked up at $p$ at time $\tdep$.
We define
\begin{align*}
  \cwaitvio(\ins) &= \gammawait \cdot \max\{\tdep - \treq(r) - \twaitmax, 0 \} \text{,} \\
  \ctripvio(\ins) &= \gammatrip \cdot \max\{\ttrip(\ins) - \ttripmax(r), 0 \}
\end{align*}
with model parameters $\gammawait$ and $\gammatrip$ that scale the severity of the penalties.

\section{Preliminaries}
\label{sec:Preliminaries}
In this section, we describe several shortest path algorithms used in this work.
Furthermore, we summarize the dynamic taxi sharing dispatcher introduced by~\citet{Buchhold2021} that serves as the basis of our work.

\mysubsection{Shortest Path Algorithms}
\label{subsec:shortest_path_algorithms}
In the following, we explain a number of algorithms that compute different variants of shortest path queries on road networks.

\parheader{Dijkstra's Shortest Path Algorithm}
\label{par:dijkstras_shortest_path_algorithm}
\emph{Dijkstra's shortest path algorithm}~\cite{Dijkstra1959} computes the shortest path from a source $s \in V$ to all other vertices in a weighted graph $G=(V,E,\ell)$.

The algorithm stores a distance label $\tdist(s,v)$ for every $v \in V$.
An addressable priority queue $Q$ with $\text{key}(v) = \tdist(s,v)$ contains active vertices. 
Initially, $Q \gets \{ s \}$, $\tdist(s,s) \gets 0$ and $\tdist(s,v) \gets \infty$ for $v \ne s$.
The algorithm repeatedly extracts the vertex with the smallest distance label from $Q$ and \emph{settles} it. 
To settle $u \in V$, each outgoing edge $(u,v) \in E$ is \emph{relaxed} by trying to improve the distance label $\tdist(s,v)$ with $\tdist(s,u) + \ell(e)$.
If the distance is improved, $v$ is inserted into $Q$.
The algorithm stops when $Q$ becomes empty. 

\parheader{Contraction Hierarchies}
\label{par:contraction_hierarchies}
\begin{figure}[t]
  \centering
  \includegraphics[width=0.35\textwidth]{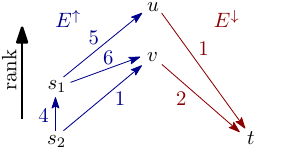}
  \caption{
    Example CH. 
    Edges are annotated with weights.
    Vertical order of vertices indicates rank. 
    Upward edges are blue and downward edges are red.
    \vspace*{-0.6cm}
  }
  \label{fig:ch_example}
\end{figure}
\emph{Contraction Hierarchies (CHs)}~\cite{Geisberger2012} are a speed-up technique for shortest path computations that exploits the hierarchical nature of road networks.
A CH is constructed in a pre-processing phase. 
Then, shortest path queries can be computed on the CH using restricted Dijkstra searches.

To construct a CH, all vertices in a road network $G=(V,E)$ are ordered heuristically by their importance or \emph{rank}~\cite{Geisberger2012}.
Vertices are contracted in the order of increasing rank.
The contraction of $v \in V$ temporarily removes $v$ from the graph.
To preserve shortest paths, a \emph{shortcut edge} $(u,w)$ is created if $(u,v,w) \in E^2$ is the only shortest path between $u$ and $w$.

Let $E^+$ contain all original edges $E$ as well as all shortcut edges.
The graph $G^+=(V, E^+)$ constitutes the CH.
The length $\ell^+(e)$ of a shortcut edge $e$ is the sum of the lengths of replaced original edges while $\dist^+$ is the according distance function.
For the query phase, we partition $E^+$ into \emph{up-edges} $\Eup = \setc{(u,v) \in E^+}{\rank(u) < \rank(v)}$ and \emph{down-edges} $\Edown = \setc{(u,v) \in E^+}{\rank(u) > \rank(v)}$.
We define an \emph{upwards search graph} $\Gup \gets (V, \Eup)$ and a \emph{downwards search graph} $\Gdown \gets (V, \Edown)$.
The distances $\dup$ and $\ddown$ represent $\dist^+$ constrained to $\Gup$ and $\Gdown$.

For any two vertices $s,t \in V$, it can be shown that there is a shortest path from $s$ to $t$ that is an \emph{up-down path} in the CH, i.e. consists of only up-edges followed by only down-edges~\cite{Geisberger2012}.
A \emph{CH query} from a source $s \in V$ to a target $t \in V$ runs a forward Dijkstra search from $s$ in $\Gup$ and a reverse Dijkstra search from $t$ in $\Gdown$.
Whenever the searches meet, they find an up-down-path from $s$ to $t$, eventually finding a shortest path.
The query can stop once the radius of either Dijkstra search exceeds the best previously found distance from $s$ to $t$.
For any vertex $v \in V$, let $\Gup_v$ denote the sub-graph of $\Gup$ that contains all vertices $\Vup_v$ and edges $\Eup_v$ that are reachable from $v$.
Similarly, let $\Gdown_v$ denote the sub-graph of $\Gdown$ that contains all vertices $\Vdown_v$ and edges $\Edown_v$ from which $v$ can be reached.
We call $\Gup_v$ and $\Gdown_v$ the \emph{forward} and \emph{reverse CH search space} of $v$, respectively.

\parheader{Bucket Contraction Hierarchy Searches}
\label{par:bucket_contraction_hierarchy_searches}
\emph{Bucket Contraction Hierarchy (BCH)} searches~\cite{Knopp2007,Geisberger2012} find all shortest path distances from a set of sources $S \subseteq V$ to a target $t \in V$ in a road network $G=(V,E)$.
A CH $G^+$ of $G$ is used as the basis of the algorithm.

The idea is to construct a \emph{(source) bucket} $\Bs(v)$ at each vertex $v \in V$.
Conceptually, $\Bs(v)$ is a list of entries, each one of which stores the upwards distance from one of the sources to $v$.  
For each source $s \in S$, a forward search in $\Gup$ is run that adds an entry $(s, \dup(s,v))$ to $\Bs(v)$ for every settled $v \in V$.
Then, a reverse search from $t$ in $\Gdown$ can compute tentative shortest path distances as $\dup(s,v) + \ddown(v,t)$ for every bucket entry $(s, \dup(s,v)) \in \Bs(v)$ at every settled vertex $v$.

Consider the example CH depicted in~\cref{fig:ch_example}.
 A BCH for $S=\{ s_1, s_2 \}$ would have the buckets $\Bs(u) = \langle (s_1,5), (s_2, 9) \rangle$ and $\Bs(v) = \langle (s_1, 6), (s_2, 1)) \rangle$.
A reverse search from $t$ traverses $\Gdown$ and finds shortest up-down paths between $s_1,s_2$ and $t$ by scanning $\Bs(u)$ and $\Bs(v)$.

BCH searches can analogously compute the distances from a single source to a set of targets.
In that case, we speak of \emph{target buckets} $\Bt(v)$ for every $v \in V$.

The advantage of BCH searches over point-to-point CH queries is that the search space of each source and each target is only traversed once, either to compute bucket entries or to scan bucket entries.
However, BCH searches require more memory to store the bucket entries.

\parheader{Bundled Searches}
\label{par:bundled_searches}
Dijkstra-based shortest path algorithms for multiple sources can be \emph{bundled} s.t. the searches for $k$ sources are advanced simultaneously.
A bundled search maintains $k$ tentative distance labels at each vertex.
When the search relaxes an edge $(u,v) \in E$, it tries to update all $k$ distance labels at $v$.

A bundled relaxation can be more cache efficient than $k$ individual relaxations as all $k$ distances are stored in consecutive memory.
However, the relaxation of $(u,v) \in E$ may perform unproductive work if not all $k$ searches have reached $u$ yet.
Thus, bundling is effective if all $k$ searches relax largely the same edges.
The value of $k$ is a tuning parameter.

Bundled searches were first introduced for Dijkstra searches used for the computation of arc-flags under the name \emph{centralized searches}~\cite{Hilger2009}.
Since then, bundled searches have been used in a number of Dijkstra-based many-to-many shortest path algorithms~\cite{Bauer2010,Yanagisawa2010,Delling2011,Delling2013,Delling2017}, including point-to-point queries in CHs~\cite{Buchhold2019}.

\parheader{SIMD Parallelism in Bundled Searches}
Bundled searches can be sped up using single-instruction multiple-data (SIMD) parallelism~\cite{Buchhold2019}.
Modern CPUs provide special vector registers and instructions that can store and manipulate multiple data items simultaneously.
We can vectorize the computations needed during edge relaxations s.t. $k$ computations are performed at the same time using a single vector instruction.
SIMD instructions can substantially speed up bundled searches~\cite{Buchhold2019}.

\mysubsection{LOUD}
\label{subsec:LOUD}
Our algorithm is based on the dynamic taxi sharing dispatching algorithm \emph{LOUD}~\cite{Buchhold2021}. 

Given a fleet of vehicles and routes, the online algorithm matches incoming taxi sharing requests to vehicles.
For each request, a feasible insertion of the request's origin $o$ and destination $d$ into a vehicle's route is found s.t. the detour of the vehicle is minimized.

\parheader{Elliptic Pruning}
To compute the costs of possible insertions, the algorithm requires the distances between existing vehicle stops and $o$ and $d$.
LOUD computes these distances using BCHs with bucket entries for each vehicle stop and queries run from $o$ and $d$.

We refer to these BCH searches as \emph{elliptic BCH searches} as they utilize a pruning technique for these buckets called \emph{elliptic pruning}:
Each insertion is subject to the same soft and hard constraints that we describe in~\cref{subsec:cost_function}.
The wait time and trip time hard constraints of riders already assigned to a vehicle $\nu \in F$ define a \emph{leeway} $\lambda(s_i,s_{i+1})$, i.e. a maximum permissible detour, between each pair of consecutive stops $(s_i,s_{i+1}) \in R(\nu)$. 
Any detour that exceeds $\lambda(s_i,s_{i+1})$ breaks some hard constraint and is infeasible. 
The leeway $\lambda(s_i,s_{i+1})$ defines a detour ellipse that contains all vertices at which a pickup or dropoff may be made between $s_i$ and $s_{i+1}$ without breaking a hard constraint.
Thus, bucket entries for $s_i$ and $s_{i+1}$ only need to be generated at vertices within the ellipse.
Elliptic pruning vastly reduces the number of bucket entries that need to be scanned by the BCH searches and limits the number of candidate vehicles for insertions~\cite{Buchhold2021}.

\parheader{Last Stop Distances}
LOUD also allows the insertion of the origin and/or destination after the last stop of a vehicle's route.
Here, elliptic pruning is not applicable since the leeway of any vehicle is unbounded after the last stop.
Instead, LOUD uses reverse Dijkstra queries in the road network rooted at $o$ or $d$ to find the distances from last stops to $o$ or $d$.
These Dijkstra queries, particularly for the destination of a request, constitute a significant part (at least $60\%$ and up to more than $90\%$) of the total runtime of LOUD.


\section{Conceptual Changes for Multiple Pickup and Dropoff Locations}
\label{sec:conceptual_changes}
We observe that introducing meeting points requires a careful consideration of their effects on vehicle detours and rider trips.
Here, we describe these effects in detail and establish the formal foundation of our cost function.

Remark that we make two simplifications in this section for the sake of brevity:
First, we only consider ordinary insertions that insert the pickup and dropoff after the next stop, before the last stop, and between two different pairs of stops of the vehicle's route (see~\cref{fig:insertion_types}). 
Formally, we only regard insertions $\ins = (r, p, d, \nu, i, j)$ with $0 < i < j < \numstopsnu$. 
Second, we do not consider the possibility of $p$ or $d$ coinciding with existing stops, i.e. we assume $p \ne l(s_i)$ and $d \ne l(s_j)$.
We ignore these cases as they would lead to bloated definitions.
However, with knowledge of the vehicle's current location and the location $\stoploc(s)$ of each stop $s \in R(\nu)$, the ignored cases could be integrated into the following definitions in a straight forward manner.

\mysubsection{Walking Time and Walking to the Destination}
\label{subsec:walking_time_and_walking_to_the_destination}
The walking time of a regular insertion $\ins=(r,p,d,\nu,i,j)$ is simply $\twalk(\ins) \gets \distpsg(p) + \distpsg(d)$.

We allow a rider $r$ to walk from their origin to their destination without ever boarding a vehicle.
This requires a manner of pseudo-insertion $\ins_{\text{psg}}$ where the rider is matched to no vehicle at all. 
Then, the cost of $\ins_{\text{psg}}$ depends only on the walking distance $\twalk(\ins_{\text{psg}}) = \ttrip(\ins_{\text{psg}}) = \distpsg(\orig, \dest)$.
The pseudo-insertion affects no vehicle operation times or trip times of other riders. 
We ignore the wait time soft constraint since the rider does not wait for a vehicle.
In effect, the total cost is $c(\ins_{\text{psg}}) = (\tripweight + \walkweight) \cdot \distpsg(\orig, \dest) + \ctripvio(\ins_{\text{psg}})$.
We explicitly allow pseudo-insertions for any distance $\distpsg(\orig,\dest)$, i.e. the distance does not have to be found within the radius $\rho$ around $\orig$ or $\dest$.
Instead, we use a CH query in the rider road network to find $\distpsg(\orig,\dest)$.
The cost $c(\ins_{\text{psg}})$ can serve as a first upper bound $\cmaxglobal$ on the cost of any insertion.

\mysubsection{Vehicles Waiting for Riders}
\label{subsec:vehicles_waiting_for_riders}
In dynamic taxi sharing without meeting points, a rider $r=(\orig, \dest, \treq)$ always waits to be picked up by the vehicle at $\orig$.
The request is issued at time $\treq$ and the vehicle $\nu$ matched to the request can start making its way to the pickup location at the earliest at $\treq$. 
This means that only the rider can wait for the vehicle, not the other way around.
Upon arriving at $\orig$, the vehicle can immediately depart (ignoring the time to pick up the rider).
Thus, a vehicle's schedule is fully characterized by the departure time $\tdepmin(s_i)$ at each stop (which is equal to the arrival time $\tarrmin(s_i)$).
The shortest path distances between stops can then be inferred as $\distveh(s_i, s_{i+1}) = \tdepmin(s_{i+1}) - \tdepmin(s_i)$.

In the presence of meeting points, vehicle schedules become more complex.
Consider an insertion $\ins = (r, p, d, \nu, i, j)$ with $p \ne \orig$.
Both the vehicle and the rider have to travel to $p$ starting at $\tdepmin(s_i)$ and $\treq$, respectively.
Consequently, the vehicle can arrive at $p$ earlier than the rider, precisely if $\tdepmin(s_i) + \distveh(s_i, p) < \treq + \distpsg(p)$.
In that case, the vehicle needs to wait for the rider at $p$ for a time $\tvehwait(s_a)$. 
Thus, the actual departure time at a pickup is the maximum of the earliest possible departure time of the vehicle and that of the rider:
\[
  \tdep(\ins) = \max \{ \tdepmin(s_i) + \distveh(s_i, p), \treq + \distpsg(p) \} \text{.}
\]
We later use $\tdep(\ins)$ in the definitions of the vehicle detour, rider wait time and rider trip time needed for the cost function (see~\cref{subsec:cost_function}). 
In particular, the wait times of a vehicle or of a rider contribute to the vehicle operation times and rider trip times.

Whereas in the traditional model $\tarrmin(s) = \tdepmin(s)$ for every stop $s$, we have to now explicitly store $\tarrmin(s)$ and $\tdepmin(s)$.
We can then derive the vehicle wait times as $\tvehwait(s_i) = \tdepmin(s_i) - \tarrmin(s_i)$ and the distance between a pair of consecutive stops as $\distveh(s_i, s_{i+1}) = \tarrmin(s_{i+1}) - \tdepmin(s_i)$.

\mysubsection{Added Vehicle Operation Time and Rider Trip Times}
\label{subsec:times_for_costs_definitions}
In the following, we define the added vehicle operation time $\tdetour(\ins)$, the trip time for the new rider $\ttrip(\ins)$, as well as the sum of added trip times for existing riders $\ttripplus(\ins)$ for an insertion $\ins = (r, p, d, \nu, i, j)$.

\parheader{Detours}
We can express all changes to the vehicle operation time and rider trip times in terms of the detours made by $\nu$ to perform an additional pickup at $p$ and dropoff at $d$.

\begin{Definition} 
\label{def:full_detour}
The \emph{full pickup (dropoff) detour} for an insertion $\ins = (r,p,d,\nu,i,j)$ is the detour that results from the vehicle $\nu$ first driving to $p$ ($d$) after stop $s_i$ ($s_j$) instead of driving to $s_{i+1}$ ($s_{j+1}$) directly.\\ 
Formally, we define the full pickup detour $\initpdetour(\ins)$ and the full dropoff detour $\initddetour(\ins)$ as 
\begin{align*}
  \initpdetour(\ins) & \gets \tdep(\ins) - \tdepmin(s_i) + \distveh(p,s_{i+1}) - \distveh(s_i,s_{i+1}) \\
  \initddetour(\ins) & \gets \distveh(s_j, d) + \distveh(d, s_{j+1}) - \distveh(s_j,s_{j+1}) \\
\end{align*}
\end{Definition}

Intuitively, the vehicle operation time increases by the sum of these full detours.
However, existing vehicle wait times at later stops $\nu$ (see~\cref{subsec:vehicles_waiting_for_riders}) can act as buffers that reduce the added vehicle operation time.
Assume that $\nu$ has to wait for a rider at a stop $s_a$ with $i < a \le j$.
Then, the insertion $\ins$ will cause $\nu$ to arrive at $s_a$ with a delay of $\initpdetour(\ins)$.
However, $\nu$ would have spent a time $\tvehwait(s_a)$ waiting at $s_a$ anyways.
This time is now spent making (part of) the detour $\initpdetour(\ins)$.
Thus, the arrival at $s_{a+1}$ is only delayed by $\initpdetour(\ins) - \tvehwait(s_a)$.
We call this a \emph{residual detour}.

\begin{Definition} 
\label{def:residual_detour}
  The \emph{residual detour} $\resdetour{a}(\ins)$ for an insertion $\ins = (r, p, d, \nu, i, j)$ at stop $s_a \in R(\nu)$ describes how much later the vehicle $\nu$ will arrive at stop $s_a$ after the insertion is performed. 
  We define it inductively as
  \[ \resdetour{a + 1}(\ins) \gets \begin{cases} 
      0 & \text{if } a < i \\
      \initpdetour(\ins) & \text{if } a = i \\
      \max \{ \resdetour{j}(\ins) - \tvehwait(s_j), 0 \} + \initddetour(\ins) & \text{if } a = j \\
      \max \{ \resdetour{a}(\ins) - \tvehwait(s_{a}), 0 \} & \text{otherw.}
  \end{cases}\]
\end{Definition}

\parheader{Added Vehicle Operation Time, Trip Times}
Residual detours allow us to define the added vehicle operation times as well as the trip times of both the new rider and existing riders.

The added operation time of $\nu$ is the delay of the arrival of $\nu$ at its last scheduled stop.
This is simply the residual detour at the last stop.

\begin{Definition} 
  The added vehicle operation time $\tdetour(\ins)$ caused by an insertion $\ins = (r, p, d, \nu, i, j)$ is defined as
  \[ \tdetour(\ins) \gets \resdetour{\numstopsnu}(\ins) \]
\end{Definition}

Further, residual detours define the new arrival times $\tarrminprime$ and departure times $\tdepminprime$ after performing $\ins$ as
\begin{align*}
  \tarrminprime(s_a, \ins) &= \tarrmin(s_a) + \resdetour{a}(\ins) \text{ and}\\
  \tdepminprime(s_a, \ins) &= \max \{ \tarrmin{}'(s_a), \tdepmin(s_a) \} \text{.}
\end{align*}
With this, we can calculate the new rider's total trip time. 
\begin{Definition} 
  \label{def:trip_time}
  The trip time $\ttrip(\ins)$ for an insertion $\ins = (r,p,d,\nu,i,j)$ is defined as 
  \[ \ttrip(\ins) \gets \tdepminprime(s_j,\ins) + \distveh(s_j, d) + \distpsg(d) - \treq(r) \]
\end{Definition}

Each existing rider experiences an added trip time depending on where they are dropped off.
The delay of each riders arrival at their dropoff stop is the residual detour at that stop.
\begin{Definition} 
  \label{def:added_trip_time_for_existing_riders}
  Let $\numdropoffs(s)$ be the number of dropoffs currently scheduled to be performed at stop $s \in R(\nu)$ for a vehicle $\nu \in F$.
  The combined added trip time for existing riders $\ttripplus(\ins)$ caused by an insertion $\ins = (r, p, d, \nu, i, j)$ is defined as
  \[ \ttripplus(\ins) \gets \sum_{a = i + 1}^{\numstops{\nu}} \numdropoffs(s_a) \cdot \resdetour{a}(\ins) \]
\end{Definition}


\section{Algorithm Overview}
\label{sec:the_algorithm}
We introduce the \emph{\karri} (\underline{Ka}rlsruhe \underline{R}apid \underline{Ri}desharing) algorithm that efficiently answers ridesharing requests with multiple meeting points using fast many-to-many routing.  
\begin{figure*}[tb]
  \centering
  \begin{tabular}{ccl}
    \hline
    \makecell[c]{Ordinary} & \makecell[c]{$0 < i < j < \numstopsnu$} & \makecell[l]{\rule{0pt}{1.1cm}\rule[-.15cm]{0pt}{.15cm}\includegraphics[height=0.9cm,page=2]{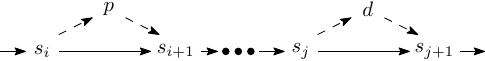}} \\
    \hline
    \makecell[c]{Ordinary\\Paired (OP)} & \makecell[c]{$0 < i = j < \numstopsnu$} & \makecell[l]{\rule{0pt}{1.1cm}\rule[-.15cm]{0pt}{.15cm}\includegraphics[height=0.9cm,page=3]{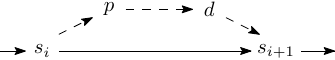}} \\
    \hline
    \makecell[c]{Pickup Before\\Next Stop (PBNS)} & \makecell[c]{$0 = i \le j < \numstopsnu$} & \makecell[l]{\rule{0pt}{1.1cm}\rule[-.15cm]{0pt}{.15cm}\includegraphics[height=0.9cm,page=4]{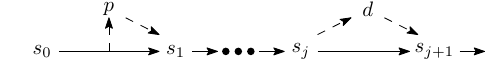}} \\
    \hline
    \makecell[c]{Pickup After\\Last Stop (PALS)} & \makecell[c]{$0 \le i = j = \numstopsnu$} & \makecell[l]{\rule{0pt}{1.1cm}\rule[-.15cm]{0pt}{.15cm}\includegraphics[height=0.9cm,page=5]{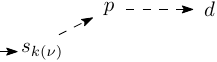}} \\
    \hline
    \makecell[c]{Dropoff After\\Last Stop (DALS)} & \makecell[c]{$0 < i < j = \numstopsnu$} & \makecell[l]{\rule{0pt}{1.1cm}\rule[-.15cm]{0pt}{.15cm}\includegraphics[height=0.9cm,page=6]{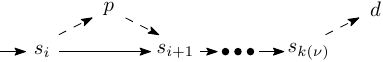}} \\
    \hline
  \end{tabular}
  \caption{Insertion types.
  Shows characterization of each type based on the pickup and dropoff insertion points $i$ and $j$ of an insertion $\iota=(r,p,d,\nu,i,j)$. 
  Illustrations depict the current route of $\nu$ (solid arrows) with stops $s \in R(\nu)$ as well as the detours to and from $p$ and $d$ (dashed lines).}
  \label{fig:insertion_types}
\end{figure*}
The \karri algorithm dynamically accepts requests and finds an insertion for each request that has optimal cost according to the cost function and current system state.

For a request $r$, the algorithm first finds the possible meeting points in a walking radius $\rho$ around the origin and destination using bounded Dijkstra searches.
Then, the algorithm evaluates all insertions in the order of types illustrated in~\cref{fig:insertion_types}.
For each insertion, \karri computes the cost according to the cost function (see~\cref{subsec:cost_function}).
The insertion with the smallest cost $\iota^\ast$ is repeatedly updated and eventually returned.

Since we consider sets of possible meeting points, the number of potential insertions becomes the main challenge of the algorithm.
In particular, we face the issue of computing the shortest paths between existing vehicle stops and each meeting point.
For this, we do not employ inexact filtering heuristics to reduce the number of necessary shortest path computations.
Instead, we find the insertion with optimal cost by applying many-to-many routing techniques to all combinations of vehicle stops and meeting points.

We engineer these routing techniques to act as filters of feasible insertions themselves by discarding sub-optimal candidates as early as possible during our searches using a variety of pruning methods. 
In the following sections, we describe these methods for each insertion type and associated shortest path query.


\section{Ordinary, Ordinary Paired, and Pickup Before Next Stop Insertions}
\label{sec:ordinary_op_and_pbns_insertions}
This section is concerned with \Ord, \OP, and \PBNS insertions (see~\cref{fig:insertion_types}).
In all three insertion types mentioned, both the pickup $p$ and the dropoff $d$ are inserted between two existing stops of the route $R(\nu)$.
To compute the cost of any insertion of one of these types, we need to know the distances between existing vehicle stops and the meeting points. 
BCH searches with elliptic pruning (see~\cref{subsec:LOUD}) have been shown to efficiently compute these distances~\cite{Buchhold2021}.
We call these \textit{elliptic BCH searches}.
Additionally, cost calculations for paired insertions require the PD-distance $\distveh(p,d)$.
In this section, we explain how we extend the required distance queries for multiple meeting points.

\subsection{Elliptic BCH Searches}
\label{subsec:elliptic_bch_searches}
We can extend elliptic BCH queries to multiple meeting points by simply repeating the queries for each meeting point.
Even with elliptic pruning, this can lead to impractical running times for large numbers of meeting points, though.
Therefore, we supplement elliptic BCH searches with two techniques for better scalability to larger numbers of meeting points. 
We describe these techniques only for pickups but they work analogously for dropoffs.

\parheader{Elliptic BCH Searches with Sorted Buckets}
\label{par:elliptic_bch_searches_with_sorted_buckets}
First, we propose using \emph{sorted buckets} to reduce the number of bucket entries scanned by BCH queries.
We explain the principle only for source buckets but it can be analogously applied for target buckets.

Recall that the constraints for existing riders of a vehicle $\nu$ define a leeway $\lambda(s_i, s_{i+1})$ for the detour between any two consecutive vehicle stops $s_i$ and $s_{i+1}$ of $\nu$ (see~\cref{subsec:LOUD}).
When an elliptic BCH search to a pickup $p$ scans a source bucket entry $(s, \dup(s_i,v)) \in \Bs(v)$, the tentative distance $\dup(s_i,v) + \tddown(v,p)$ can only lead to an insertion that holds all hard constraints if $\dup(s_i,v) + \tddown(v,p) \le \lambda(s_i, s_{i+1})$.
This means the entry is only relevant for $p$ if $\tddown(v,p) \le \lambda(s_i, s_{i+1}) - \dup(s_i,v)$.
We call $\remleeway(s_i, v) \gets \lambda(s_i,s_{i+1}) - \dup(s_i,v)$ the \emph{remaining leeway} of a source bucket entry $(s_i, \dup(s_i,v))$.

We sort the entries of each source bucket at each vertex $v$ by their remaining leeway in decreasing order.
Then, an elliptic BCH search to a pickup $p$ can stop scanning the entries at $v$ once an entry $(s, \dup(s,v))$ is scanned for which $\tddown(v,p) > \remleeway(s,v)$.
In this case, we have $\tddown(v,p) > \remleeway(s,v) \ge \remleeway(s',v)$ for any subsequent entry $(s', \dup(s',v))$, so the remaining entries cannot lead to any insertions that adhere to the hard constraints.
Maintaining the order of each bucket comprises a time overhead when inserting bucket entries.
However, since bucket sizes are small, this overhead is limited.
Note that sorted buckets can also be applied in the case of only a single pickup.

\parheader{Bundled Elliptic BCH Searches}
\label{par:bundled_elliptic_bch_searches}
Second, we employ \emph{bundled elliptic BCH searches} that exploit the locality of pickups.

Like any bundled search, a bundled elliptic BCH search is rooted at $k$ pickups and updates $k$ distances with each edge relaxation (see~\cref{par:bundled_searches}).
Additionally, we can bundle bucket entry scans.
Whenever a bucket entry for a stop $s$ is scanned, the bundled search tries to improve upon each of the $k$ tentative distances between $s$ and any of the $k$ pickups.  
We can effectively bundle the edge relaxations and bucket entry scans of elliptic BCH searches because the localized pickups share similar CH search spaces.  
Moreover, we can use vectorized instructions to parallelize both edge relaxations and bucket entry scans.
At the same time, elliptic pruning and sorted buckets can still be applied.
To our knowledge, our algorithm is the first to explicitly use bundled BCH searches. 
The idea follows from the bundled CH searches used in~\cite{Buchhold2019}.

\subsection{PD-Distance Searches}
\label{subsec:pd_distance_searches}
Computing the PD-distances, i.e. the distances between pickups and dropoffs, is a many-to-many shortest path problem where the set of sources and the set of targets are localized.

Our algorithm uses a BCH approach to address this problem.
We generate bucket entries for all dropoffs in their reverse CH search spaces.
Then, we run queries in the upward CH search graph rooted at each pickup to find the PD-distances using the dropoff bucket entries.
We propose two methods to improve these BCH searches.

Firstly, let $\maxPDDist$ be an  an upper bound on all PD-distances.
Then, we only have to generate and scan bucket entries in a radius of $\maxPDDist$. 
We use
\[
  \maxPDDist \gets \max_{p \in \Prho} \distveh(p, \orig) + \distveh(\orig,\dest) + \max_{d\in \Drho} \distveh(\dest, d)\text{.}
\]
We can compute $\distveh(p,\orig)$ for all $p \in \Prho$ and $\distveh(\dest, d)$ for all $d \in \Drho$ using two local Dijkstra searches rooted at $\orig$ and $\dest$, respectively.
We obtain $\distveh(\orig,\dest)$ with a single preliminary CH query. 

Secondly, we can once again use bundled BCH searches.
More specifically, we can generate bucket entries for batches of $k$ dropoffs and then run queries for batches of $k$ pickups where $k$ is a configuration parameter. 
Again, bundled PD-distance searches utilize the locality of pickups and dropoffs and allow us to employ SIMD parallelism.

\subsection{Enumerating Ordinary, Ordinary Paired, and Pickup Before Next Stop Insertions}
\label{subsec:enumerating_ordinary_op_and_pbns_insertions}
After running our elliptic BCH queries and PD-distance searches, we know all distances that are required for ordinary and \OP insertions.
We enumerate the insertions $\iota=(r,p,d,\nu,i,j)$ with $0 < i \le j < \numstopsnu$ for a set of candidate vehicles found by the elliptic BCH queries~\cite{Buchhold2021}.
We compute the cost $c(\iota)$ for each insertion and update $\iota^\ast$ to $\iota$ if $c(\iota) < c(\iota^\ast)$.

In a \emph{pickup before next stop (PBNS)} insertion $\ins=(r,p,d,\nu,0,j)$, the pickup $p$ is inserted between stops $s_0$ and $s_1$ of the vehicle.
This requires the vehicle to be redirected at its current location $\curloc(\nu)$ to drive to $p$ next instead of $s_1$ (cf.~\cref{fig:insertion_types}).
Thus, to compute the cost of a PBNS insertion, we need to know the distance $\distveh(\curloc(\nu), p)$.

In order to avoid finding $\distveh(\curloc(\nu),p)$ for every $\nu \in F$ and $p \in \Prho$, we employ a filtering technique proposed by~\citet{Buchhold2021}.
The technique exploits that $\distveh(s_0,p)$ is a lower bound on $\distveh(s_0,\curloc(\nu)) + \distveh(\curloc(\nu), p)$.
The distance $\distveh(s_0,p)$ can be computed by the elliptic BCH searches which means we can use it to compute a lower bound on the pickup detour as well as the cost of $\ins$.
If the lower bound cost of $\ins$ exceeds the best known cost, we can discard $\ins$.

Most of the time, this filter leaves us with a very small number of pairs of vehicle $\nu$ and pickup $p$ for which we actually need to compute $\distveh(\curloc(\nu), p)$.
\myalg uses a bucket based approach for the remaining necessary queries.
We generate source bucket entries for the current location of every affected vehicle and run bundled queries from the pickups.
The average number of such queries per request is less than $0.5$.


\section{Pickup After Last Stop Insertions}
\label{sec:pickup_after_last_stop_insertions}
In this section, we consider \PALS (PALS) insertions.
The main challenge of PALS insertions is the computation of the distances from last stops to pickups.
The authors of LOUD find that elliptic pruning is not applicable for the computation of these distances~\cite{Buchhold2021}. 
Instead, LOUD uses a reverse Dijkstra search rooted at $\orig$ that is stopped early when the search can no longer find an insertion better than the best known one.
For multiple pickups, we can compute the required distances by analogously running reverse Dijkstra searches for each pickup. 
These Dijkstra searches may also be bundled to exploit the locality of the pickups.

However, even with a single pickup, this Dijkstra search takes up a significant part of the running time of the LOUD algorithm. 
Thus, for a large number of pickups, we expect infeasible running times.
In this section, we introduce two new BCH based approaches for the computation of last stop distances.
For the rest of this section, let $\cmaxglobal$ denote an upper bound on the best known insertion cost (initially $\cmaxglobal \gets c(\iota^\ast)$).

\parheader{Reformulation of Cost Function for PALS Insertions}
Note that the cost of any PALS insertion $\iota=(r,p,d,\nu,\numstopsnu, \numstopsnu)$ is fully characterized by the pickup $p$, the PD-distance $\distveh(p,d)$, the walking distance $\distpsg(d)$, the departure time $\tdepmin(\laststopnu)$ of $\nu$ at $\laststopnu$, and the last stop distance $\distveh(\laststopnu,p)$. 
Thus, we can write the cost of $\iota$ as 
\[
  c(\ins) = c'(r,p,\distveh(p,d), \distpsg(d), \tdepmin(\laststopnu), \distveh(\laststopnu, p))\text{.}
\]
Importantly, the value of $c'$ monotonously increases with its last four arguments.
Furthermore, it will be helpful to define $\delta_{\text{pd}}^{\min} \gets \min_{p \in \Prho, d \in \Drho} \distveh(p,d)$ as a lower bound on any PD-distance.

\subsection{Last Stop BCH Searches for PALS}
\label{subsec:last_stop_bch_searches_for_PALS}
Even though elliptic pruning is not applicable, we can still employ a BCH search approach for distances from last stops to pickups.
For this, we maintain a \emph{last stop bucket} $\Blast(v)$ for every $v \in V$.
For every last stop $\laststopnu$, we generate an entry $(\laststopnu, \dup(\laststopnu,v)) \in \Blast(v)$ at each vertex $v$ in the upward CH search space rooted at $\laststopnu$.
Then, for every pickup $p \in \Prho$, we run an \emph{individual (last stop) BCH query} that explores the reverse CH search space $\Gdown_p$ rooted at $p$ and scans the last stop bucket at each settled vertex to compute the shortest path distances from last stops to $p$.
When the search scans an entry $(\laststopnu, \dup(\laststopnu,v)) \in \Blast(v)$, it tries to improve the tentative distance $\tdist(\laststopnu, p)$ with $\dup(\laststopnu, v) + \ddown(v,p)$.
Eventually, the shortest distance $\distveh(\laststopnu, p)$ is found for every last stop $\laststopnu$.

Whenever the last stop of a vehicle changes, we traverse the forward CH search spaces of the old and new last stop to remove all old bucket entries and insert new entries, respectively.
We stop each individual last stop BCH search as soon as the current distance no longer admits a PALS insertion with cost smaller than the best known cost.
Furthermore, we can bundle the BCH queries and use SIMD parallelism in a similar manner to bundled elliptic BCH searches (see~\cref{par:bundled_elliptic_bch_searches}).

\parheader{Pruning Bucket Scans using Sorted Buckets}
\label{par:cost_pruning_of_bucket_scans_using_sorted_buckets}
A remaining issue of this approach is the size of the last stop buckets. 
Without elliptic pruning, buckets contain many more entries, especially at vertices that have a high rank in the CH.
Therefore, the queries have to scan large numbers of bucket entries, rendering the last stop BCH approach inefficient.

The future work section of~\cite{Buchhold2021} suggests sorting the entries within each last stop bucket by their distance to address this issue.
Suppose an individual BCH query rooted at $p \in \Prho$ scans the bucket $\Blast(v)$.
For every entry $e = (\laststopnu, \dup(\laststopnu, v))$, let 
\[ 
  \cmin(e) \gets c'(r,p,\delta_{\text{pd}}^{\min}, 0, \treq(r), \dup(\laststopnu, v) + \ddown(v,p)) \text{.}
\]

We can stop each bucket scan early based on this lower bound:

\begin{theorem}
  Let the entries of $\Blast(v)$ be sorted by their upward distances in ascending order.
  Further, let $\cmaxglobal$ be an upper bound on the cost of the best insertion of the current request.
  Then, a BCH query rooted at pickup $p \in \Prho$ can stop scanning $\Blast(v)$ once it encounters an entry $e \in \Blast(v)$ with $\cmin(e) > \cmaxglobal$.
\end{theorem}

\begin{proof}
  This can be shown by contradiction.
  Let $e = (\laststopnu, \dup(\laststopnu, v))$ with $\cmin(e) > \cmaxglobal$.
  Assume that the best insertion for request $r$ is the PALS insertion $\ins=(r,p,d,\nu', \numstops{\nu'}, \numstops{\nu'})$ and that $v$ is the highest ranked vertex on the only shortest up-down path from $\laststop{\nu'}$ to $p$.
  Then, $\distveh(\laststop{\nu'},p) = \dup(\laststop{\nu'}, v) + \ddown(v,p)$ and the shortest path will be found by the BCH query using an entry $e'=(\laststop{\nu'}, \dup(\laststop{\nu'}, v)) \in \Blast(v)$.

  If $\dup(\laststop{\nu'}, v) < \dup(\laststopnu, v)$, then $e'$ is scanned before $e$ and the shortest path from $\laststop{\nu'}$ to $p$ will be found.
  Otherwise, 
  \begin{align*}
      c(\ins) & = c'(r, p, \distveh(p,d), \distpsg(d), \tdepmin(\laststop{\nu'}), \distveh(\laststop{\nu'}, p)) \\ 
               & \ge c'(r, p, \delta_{\text{pd}}^{\min}, 0, \treq(r), \dup(\laststop{\nu'}, v) + \ddown(v,p)) \\
               & \ge c'(r, p, \delta_{\text{pd}}^{\min}, 0, \treq(r), \dup(\laststopnu, v) + \ddown(v,p)) \\
               & = \cmin(e) > \cmaxglobal 
  \end{align*} 
  Thus, $c(\ins)$ is larger than the upper bound $\cmaxglobal$ on the best insertion cost and $\ins$ cannot be the best insertion for $r$.
  Consequently, we do not have to scan any entries that come after $e$ in $\Blast(v)$.
\end{proof}

\parheader{Updating the Upper Bound Cost}
\label{par:pals_updating_upper_bound}
It is possible to simply use the cost of the best known insertion $c(\iota^\ast)$ for the upper bound cost $\cmaxglobal$ needed for cost pruning.
However, we can also improve $\cmaxglobal$ during the search.
Each tentative distance $\tdist(\laststopnu, p)$ found acts as an upper bound on the actual shortest distance $\distveh(\laststopnu,p)$. 
Thus, whenever the tentative distance $\tdist(\laststopnu, p)$ is updated, we can compute an upper bound
\[
  \cmax \gets c'(r,p,\distveh(p,\dest),0, \tdepmin(\laststopnu),\tdist(\laststopnu, p))
\]
on the cost of the best PALS insertion with $\nu$ and $p$. 
We update $\cmaxglobal$ to $\cmax$ if $\cmax < \cmaxglobal$.
This technique finds inexact upper bounds on the cost of the best PALS insertion early which is helpful for the stopping criterion of bucket scans.

\subsection{Collective Last Stop Searches for PALS}
\label{subsec:collective_last_stop_searches_for_PALS}
Finally, we propose a search approach based on the idea that we do not actually need to know the distance between every last stop and every pickup.
If we knew the best PALS insertion $\bestPALSins=(r,p,d,\nu,\numstopsnu,\numstopsnu)$ in advance, we would only need to find $\distveh(\laststopnu, p)$.
Obviously, we do not know $\bestPALSins$ in advance but we find that it is possible to prune the distance queries for individual pickups (or actually PD-pairs) by comparing the queries to each other whenever they meet at a vertex.
We introduce a collective BCH query that finds the best PALS insertion $\bestPALSins$ as well as the last stop distance $\distveh(\laststopnu, p)$.
In the following, we explain how labels representing PD-pairs are propagated through the CH search graph and how these labels can be pruned based on label domination.

\parheader{Open and Closed Labels}
\label{par:open_and_closed_labels}
A PD-pair label $(p,d,\ddown(v,p))$ at a vertex $v \in V$ consists of the pickup $p \in \Prho$, dropoff $d \in \Drho$ and downwards distance $\ddown(v,p)$. 
At each vertex $v \in V$, there is a set of \emph{open} labels $\open(v)$ and a set of \emph{closed} labels $\closed(v)$.
An open label is a label that has not been settled yet.
For each open label $l=(p,d,\ddown(v,p))$, we store a lower bound $\cmin(l)$ for the cost of a PALS insertion that can be found for $l$ in $\Gdown_v$
\[
  \cmin(l) \gets c'(r,p,\distveh(p,d),\distpsg(d),\treq(r),\ddown(v,p))
\]

\parheader{Algorithm Outline}
\label{par:collective_pals_algorithm_outline}
\begin{algorithm}[t!]
      \caption{Collective BCH search used to find distances from last stops to pickups.}
      \label{alg:collective_bch}
      \begin{algorithmic}[1]
      \State \textbf{Input:} $\Prho$, $\Drho$, $\Gdown = (\Vdown,\Edown)$, $\Blast(v)$ for $v \in V$
      \State \textbf{Output:} $(\bestp, \bestd)$ and $\distveh(\laststopnu, \bestp)$

      \Statex \phantom{empty line}

      \Procedure{\texttt{CollectiveBCH}}{} \label{alg:collective_bch:main_body}
        
        \State $Q \gets$ PQ of labels with $key_Q(l) = \cmin(l)$
        \State $\forall v \in V: \open(v) \gets \closed(v) \gets \emptyset$
        \State $\cmaxglobal \gets c(\iota^\ast)$

        \smallskip

        \ForEach{$(p,d) \in \Prho \times \Drho$} \label{alg:collective_bch:initial_labels}
          \State \texttt{insertLabelAtVertex}($p$, $(p,d,0)$)
        \EndFor
        \While{$Q \ne \emptyset$} \label{alg:collective_bch:main_loop_start}
          \State $l \gets Q$.\texttt{deleteMin}()
          \If{$\cmin(l) > \cmaxglobal$} \Return \EndIf
          \State \texttt{settleLabel}($l$) \label{alg:collective_bch:main_loop_end}
        \EndWhile
        \EndProcedure

        \Statex \phantom{empty line}

        \Procedure{\texttt{settleLabel}}{$l = (p,d,\ddown(v,p))$} \label{alg:collective_bch:settle_label}
        
          \State $\open(v)$.\texttt{remove}($l$); $\closed(v)$.\texttt{insert}($l$) \label{alg:collective_bch:mark_closed}

          \smallskip

          \ForEach{$e = (\laststop{\nu}, \dup(\laststop{\nu}, v)) \in \Blast(v)$} \label{alg:collective_bch:scan_entries_start}
            \If{$\cmin(l,e) > \cmaxglobal$} \textbf{break} \EndIf
            \If{$\cmax(l,e) < \cmaxglobal$}
              \State $(\bestp, \bestd) \gets (p,d)$; $\cmaxglobal \gets \cmax(l,e)$ \label{alg:collective_bch:scan_entries_end}
            \EndIf
          \EndFor

          \smallskip

          \ForEach{$(u,v) \in \Edown$} \label{alg:collective_bch:propagate_label_start} 
            \State $l' \gets (p,d,\ell^+(u,v) + \ddown(v,p))$
            \State \texttt{insertLabelAtVertex}($u$, $l'$) \label{alg:collective_bch:propagate_label_end}
          \EndFor

        \EndProcedure

        \Statex \phantom{empty line}

        \Procedure{\texttt{insertLabelAtVertex}}{$v$, $l'$} \label{alg:collective_bch:insertLabelAtVertex}
          \If {$\cmin(l') > \cmaxglobal$} \Return \EndIf \label{alg:collective_bch:cost_pruning}
          \ForEach{$l \in \open(v) \cup \closed(v)$} \label{alg:collective_bch:domination_pruning_start}
            \If{$l$ dominates $l'$} \Return \EndIf  \label{alg:collective_bch:domination_pruning_end}
          \EndFor
          \ForEach{$l \in \open(v)$}
            \If{$l'$ dominates $l$} $\open(v)$.\texttt{remove}($l$) \EndIf
          \EndFor
          \State $\open(v)$.\texttt{insert}($l'$); $Q$.\texttt{insert}($l'$)
        \EndProcedure
      \end{algorithmic}
\end{algorithm}
We give pseudocode for a collective BCH search in~\cref{alg:collective_bch}.
Our search maintains a priority queue $Q$ that contains all open labels ordered increasingly by $\cmin$.
Initially, at each pickup $p \in \Prho$, an open label $(p,d,0) \in \open(p)$ is created for each $d \in \Drho$ (line~\ref{alg:collective_bch:initial_labels}).
As long as $Q$ contains a label $l$ with $\cmin(l) \le \cmaxglobal$ for a known upper bound $\cmaxglobal$ on the cost of any insertion, our search proceeds with a next step (lines~\ref{alg:collective_bch:main_loop_start}-\ref{alg:collective_bch:main_loop_end}).
In each step of the search, the label $l \gets \min(Q)$ is removed from $Q$ and settled as described in the following.

\parheader{Settling Open Labels}
\label{par:settling_open_labels_for_collective_PALS}
Settling an open label $l = (p,d,\ddown(v,p))$ consists of three steps:
First, we mark $l$ closed at $v$, i.e. we move $l$ from $\open(v)$ to $\closed(v)$ (line~\ref{alg:collective_bch:mark_closed}).
Second, we search for a new best insertion by traversing all entries in the last stop bucket $\Blast(v)$ (lines~\ref{alg:collective_bch:scan_entries_start}-\ref{alg:collective_bch:scan_entries_end}).
For each entry $e = (\laststop{\nu}, \dup(\laststop{\nu}, v)) \in \Blast(v)$, we compute the tentative distance $\tdist(\laststopnu, p) = \dup(\laststopnu, v) + \ddown(v, p)$ and a cost upper bound
\begin{align*}
  \cmax(l,e) \gets c'( & r,p,\distveh(p,d),\distpsg(d), \\
  & \tdepmin(\laststopnu), \tdist(\laststopnu, p)).
\end{align*}
If $\cmax(l,e) < \cmaxglobal$, we mark $\iota=(r,p,d,\nu,\numstopsnu,\numstopsnu)$ as the best known PALS insertion, store the tentative distance $\tdist(\laststopnu, p)$, and update $\cmaxglobal \gets \cmax(l,e)$. 
Note that $\cmax(l,e)$ is the exact cost of the PALS insertion $\iota=(r,p,d,\nu,\numstopsnu,\numstopsnu)$ if $\tdist(\laststopnu, p)$ is a shortest path distance.
Since the BCH search finds shortest up-down paths, we will thus eventually find the best PALS insertion.
As before, we can stop each bucket scan early.
For this purpose, we compute a vehicle-independent cost lower bound $\cmin(l,e)$ s.t. we can stop the search early if $\cmin(l,e) > \cmaxglobal$ using
\begin{align*}
  \cmin(l,e) \gets  c'( & r,p,\distveh(p,d), \distpsg(d), \treq(r),  \\
  & \dup(\laststopnu, v) + \ddown(v,p)).
\end{align*}
Third, we propagate $l$ to all neighboring vertices of $v$.
For each neighboring vertex $w \in V$ with $(w,v) \in \Gdown$, we create a new open label $l' = (p,d,\ell^+(w,v) + \ddown(v,p))$ at $w$ (lines~\ref{alg:collective_bch:propagate_label_start}-\ref{alg:collective_bch:propagate_label_end}).
Here, we employ cost pruning by discarding $l'$ if the lower bound cost $\cmin(l')$ for this PD-pair and this distance exceeds $\cmaxglobal$ (line~\ref{alg:collective_bch:cost_pruning}).
Furthermore, we may be able to prune $l'$ at $v$ if it is dominated by an existing label at $v$ (lines~\ref{alg:collective_bch:domination_pruning_start}-\ref{alg:collective_bch:domination_pruning_end}) as described in the following.

\parheader{Domination Pruning}
\label{par:domination_pruning_for collective_PALS}
Propagating a label through the entire search space for every PD-pair is too expensive.
However, we find that we can compare labels at the same vertex and prune dominated labels in a technique we call \textit{domination pruning}.
Intuitively, a label $l$ dominates a label $l'$ at a vertex $v$ if we know that any insertion found in the reverse CH search space $\Gdown_v$ rooted at $v$ that uses $l'$ has higher costs than the equivalent insertion using $l$.

To formalize this, we first define an upper bound for the cost of a PALS insertion found in $\Gdown_v$ for a label $l$.
Let $l=(p,d,\ddown(v,p))$ be a PD-pair label at a vertex $v \in V$.
Let $w \in \Vdown_v$ and $e=(\laststopnu, \dup(\laststopnu, w)) \in \Blast(w)$.
Then,
\begin{align*}
  \cmax(l,v,e) \gets c'( & r,p,\distveh(p,d),\distpsg(d),\tdepmin(\laststopnu), \\
  & \dup(\laststopnu, w) + \ddown(w,v) + \ddown(v, p)) \text{.}
\end{align*}

With this, we can formally define the domination relation between labels:

\begin{Definition} \label{def:pd_pair_domination}
  A PD-pair label $l$ \emph{dominates} another label $l'$ at a vertex $v \in V$ exactly if $\cmax(l,v,e) < \cmax(l',v,e)$ for every $w \in \Vdown_v$ and $e \in \Blast(w)$.
\end{Definition}

\begin{theorem}
If a label $l$ dominates another label $l'$ at $v$, we do not need to settle $l'$ at $v$.
\end{theorem}

\begin{proof} \label{proof:pals_domination}
This can be shown by contradiction.
Assume $l_1=(p_1,d_1,\ddown(v,p_1))$ dominates $l_2=(p_2,d_2,\ddown(v,p_2))$ at $v$.
Further, assume that $\ins=(r,p_2,d_2,\nu,\numstopsnu,\numstopsnu)$ is the best PALS insertion.
Let $\pi$ be a shortest path from $\laststop{\nu}$ to $p_2$.
Wlog. $\pi$ is an up-down-path in the CH consisting of an upwards prefix $\pi^\uparrow$ and a downwards suffix $\pi^\downarrow$.
If $\pi^\downarrow$ does not contain $v$, then the collective search will not find $\pi$ in $\Gdown_v$, and we do not have to settle $l_2$ at $v$.

Otherwise, $\pi^\downarrow = (w, \dots, v, \dots, p_2)$ with $w \in \Vdown_v$.
Let $e = (\laststop{\nu}, \dup(\laststop{\nu}, w)) \in \Blast(w)$.
Since $\pi$ is a shortest path, we know that
\[
  \cmax(l_2,v,e) = c((r,p_2,d_2,\nu,\numstopsnu, \numstopsnu)) \text{.}
\]
However, $l_1$ dominates $l_2$ which means that
\begin{align*}
  c((r,p_1,d_1,\nu,\numstopsnu, \numstopsnu)) & \le \cmax(l_1,v,e) \\
  & < \cmax(l_2,v,e) \\
  & = c((r,p_2,d_2,\nu,\numstopsnu, \numstopsnu))
\end{align*}

This contradicts $\ins$ being the best PALS insertion. 
Hence, we do not have to settle label $l_2$ at $v$ to find the best pair for $\nu$.
\end{proof}

\parheader{Efficiently Computing the Domination Relation}
\label{par:efficiently_computing_pals_domination}
\begin{figure}
  \centering \includegraphics[width=0.4\textwidth]{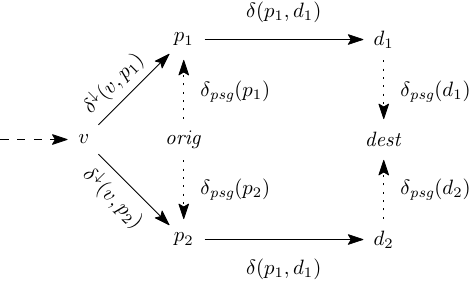}
  \caption{
  Depiction of a situation during a collective PALS search in which domination pruning may be applied for two labels $l_i=(p_i,d_i,\ddown(v,p_i))$ for $i=1,2$ at $v$.
  Solid arrows indicate known parts of potential vehicle routes and dotted arrows indicate walking routes.
  Dashed arrow symbolizes unknown vehicle route to $v$.
  }
  \label{fig:pickup_domination_relation}
\end{figure}

We find that it is not trivial to compute the domination relation efficiently because of the non-linearity of the cost function.

Consider the example depicted in~\cref{fig:pickup_domination_relation}.
Our algorithm needs to compute the domination relation between two labels $l_i=(p_i,d_i,\ddown(v,p_i))$ for $i=1,2$ at $v$.
For any bucket entry for a vehicle $\nu$ that we may scan in $\Gdown_v$, we need to decide whether $\nu$ should make the pickup and dropoff at $p_1$ and $d_1$ or $p_2$ and $d_2$.

In any case, $\nu$ passes the vertex $v$ on its way to the pickup.
Assume the vehicle arrives at $v$ at some time $t$.
We can then characterize the costs $c(\ins_i, t)$ of the tentative insertions $\ins_i=(r,p_i,d_i,\nu,\numstopsnu,\numstopsnu)$ for $i=1,2$ using $t$.
We define the components of our cost function (see~\cref{subsec:cost_function}) for this specific case:
\begin{align*}
  \tdetour(\ins_i, t) & = \tarr(\ins_i,t) - \tdepmin(\laststopnu) \\
  \ttrip(\ins_i, t) & = \tarr(\ins_i,t) + \distpsg(d_i) - \treq \\
  \ttripplus(\ins_i, t) & = 0 \\
  \twalk(\ins_i, t) & = \distpsg(p_i) + \distpsg(d_i)
\end{align*}

Here, the departure time at pickup $p_i$ and the arrival time at dropoff $d_i$ are defined as
\begin{align*}
  \tdep(p_i, t) & \gets \max \{ t + \ddown(v, p_i), \treq + \distpsg(p_i) \} \\
  \tarr(\ins_i,t) & \gets \tdep(p_i,t) + \distveh(p_i, d_i) \text{.}
\end{align*}

The constant term $\tdepmin(\laststopnu)$ in the detour time vanishes when we compute the cost difference $c(\ins_1,t)- c(\ins_2,t)$.
Thus, $t$ fully expresses the contribution of the vehicle and the cost difference is a function $\Delta_c(l_1,l_2,t)$ of only the two labels and $t$.
Then, if $\Delta_c(l_1, l_2, t) < 0$ for all $t \ge \treq$, every insertion found in $\Gdown_v$ will be better with $l_1$ than with $l_2$, i.e. $l_1$ dominates $l_2$.

If the cost function for PALS insertions were to grow linearly with $t$, then $\Delta_c(l_1,l_2,t)$ would be constant w.r.t. $t$.
In that case, we could simply check for domination using $\Delta_c(l_1,l_2,0) < 0$. 
However, the cost function does not increase linearly with $t$:
Firstly, the cost is constant w.r.t. $t$ as long as the vehicle arrives at $p_i$ earlier than the rider (see~\cref{subsec:vehicles_waiting_for_riders}).
If $t + \ddown(v,p_i) \le \treq + \distpsg(p_i)$, then $\nu$ will wait for the rider at $p_i$ for a certain wait time.
The resulting maximum operation in $\tdep(p_i, t)$ introduces a point of non-linearity.
Secondly, due to the wait time and trip time soft constraints, linear penalty terms are added to the cost function starting at a certain threshold for the wait time and trip time, both of which $t$ contributes to.

The rider arrival time at $p_i$ and the thresholds for soft constraint penalties differ between labels since $\distpsg(p_i)$, $\distpsg(d_i)$, and $\distveh(p_i, d_i)$ differ.
Thus, $\Delta_c(l_1,l_2,t)$ varies with $t$.   
Since we do not know which values of $t$ are possible for insertions found in $\Gdown_v$, we cannot trivially determine whether $l_1$ dominates $l_2$.

\parheader{Approximating the Domination Relation}
\label{par:approximating_the_domination_relation}
Instead, we under-approximate the domination relation by computing a sufficient precondition.
For this purpose, we find an upper bound $\deltacmax(l_1, l_2) \ge \max_{t \ge \treq} \Delta_c(l_1,l_2,t)$ on the difference in insertion costs between any insertion that can be found for $l_1$ and $l_2$ in $\Gdown_v$.
Then, $l_1$ dominates $l_2$ if $\deltacmax(l_1,l_2) < 0$. 

The cost upper bound is based on an upper bound on the difference in departure times at the pickup.
For any $t \ge \treq$, we have
\begin{align*}
  & \, \tdep(p_1,t) - \tdep(p_2,t) \\
  \le & \max \{ t + \ddown(v,p_1), t + \distpsg(p_1) \} - ( t + \ddown(v, p_2) ) \\
  = & \max \{ \ddown(v,p_1), \distpsg(p_1) \} - \ddown(v, p_2)\text{.}
\end{align*}
Thus, for any vehicle, the difference in the departure times at $p_1$ and $p_2$ is at most
\[ \deltatdepatpickup(l_1, l_2) \gets \max \{ \ddown(v,p_1), \distpsg(p_1) \} - \ddown(v,p_2)\text{.} \]
Then, for every insertion found in $\Gdown_v$, we have the following upper bounds on the difference in detours and trip times between $l_1$ and $l_2$: 
\begin{align*}
    \deltadetour(l_1,l_2) &\gets \deltatdepatpickup(l_1, l_2) + \distveh(p_1,d_1) - \distveh(p_2,d_2) \\
    \deltattrip(l_1,l_2) &\gets \deltadetour(l_1,l_2) + \distpsg(d_1) - \distpsg(d_2)
\end{align*}
The difference in penalties for violating the wait time and trip time soft constraints are also bounded:
\begin{align*}
    \deltawaitvio(l_1,l_2) &\gets \gammawait \max \{ \deltatdepatpickup(l_1, l_2), 0 \} \\
    \deltatripvio(l_1,l_2) &\gets \gammatrip \max \{ \deltattrip(l_1,l_2), 0 \}
\end{align*}
Note that the differences in detours and trip times are allowed to be negative to express a cost advantage for $l_1$ but the differences in penalties are not.
Even if $\deltatdepatpickup(l_1,l_2) < 0$ or $\deltattrip(l_1,l_2) < 0$, we may find insertions in $\Gdown_v$ where no penalties apply for either label.
To cover these cases, the penalty difference has to be non-negative.
Let $\Delta_{\textit{walk}} (l_1,l_2)$ be the fix difference in walking costs.
Putting it all together, we get
\begin{align*}
  \deltacmax(l_1,l_2) \gets & \deltadetour(l_1,l_2) + \\
                            & \tripweight \deltattrip(l_1,l_2) + \\
                            & \walkweight \Delta_{\textit{walk}}(l_1,l_2) + \\
                            & \deltawaitvio(l_1,l_2) + \deltatripvio(l_1,l_2)
\end{align*}

During a collective BCH search, we can compute $\deltacmax(l_1,l_2)$ in constant time with information that is known at $v$ without looking ahead in the search tree.
Since we under-approximate domination, it is possible that $l_1$ actually dominates $l_2$ but our condition does not hold.
However, we find that our domination criterion still manages to prune the vast majority of labels early.

\parheader{Limitations}
\label{par:limitations_of_collective_pals}
We remark that the insertion $\iota$ found by the collective search is only guaranteed to be the best possible PALS insertion if $\iota$ holds the service time hard constraint.
Since our search ignores the service time constraint, it may return an insertion that breaks the constraint even if there are other PALS insertions that do not.

Therefore, if $\iota$ breaks the service time constraint, we fall back to computing the distances from every last stop to every pickup using individual last stop BCH searches. 
The fallback individual BCH searches can make use of the good cost upper bounds found during the collective search.
We find that this is only necessary in exceedingly rare cases.


\section{Dropoff After Last Stop Insertions}
\label{sec:dropoff_after_last_stop_insertions}
A \emph{dropoff after last stop (DALS)} insertion $\ins=(r,p,d,\nu,i,\numstopsnu)$ inserts the dropoff (but not also the pickup) after the last stop of the vehicle's current route (cf.~\cref{fig:insertion_types}).
To compute the cost of a DALS insertion we need to compute the distance $\distveh(\laststopnu, d)$ from the vehicle's last stop to the dropoff.
This is similar to the shortest path problem in the PALS case.
We can utilize the approaches of bundled searches, BCH queries with sorted last stop buckets, and collective BCH queries with some minor differences.

Firstly, cost pruning is less effective than in the PALS case since the lower bounds on costs cannot include the PD-distance. 
Secondly, we cannot update the global cost upper bound $\cmaxglobal$ during the searches in the DALS case as we lack information about the cost of inserting the pickup earlier in the route.
Finally, collective BCH searches have some more intricate differences between the PALS and DALS cases. 
We go into more detail about these differences in the rest of this section.

\parheader{Collective BCH Searches for DALS}
Collective BCH searches for the DALS case propagate labels for individual dropoffs through the search graph.
Labels can be pruned based on a lower bound cost for each label or based on domination pruning.
Domination between dropoff labels is a partial relation defined as follows.

\begin{Definition}
  \label{def:dals_domination_relation}
  Let $l_z=(d_z,\ddown(v,d_z))$ for $z=1,2$ be two dropoff labels at $v \in V$.
  Let 
  \begin{align*}
    \Delta(l_1,l_2) \gets& \ddown(v,d_1) - \ddown(v,d_2) \\
    \deltatwalk(l_1,l_2) \gets& \distpsg(d_1) - \distpsg(d_2) \\
    \deltattrip(l_1,l_2) \gets& \Delta(l_1,l_2) + \deltatwalk(l_1,l_2)
  \end{align*}
  Then $d_1$ \emph{dominates} $d_2$ if 
  \begin{enumerate}
    \item $\Delta(l_1,l_2) + \tripweight \deltattrip(l_1,l_2) + \walkweight \deltatwalk(l_1,l_2) < 0$ and
    \item $\Delta(l_1,l_2) + (\tripweight + \gammatrip) \deltattrip(l_1,l_2) + \walkweight \deltatwalk(l_1,l_2) < 0$.
  \end{enumerate}
\end{Definition}

Assume that a label $l_1$ dominates another label $l_2$ at $v \in V$.
The domination relation makes sure that the cost for any DALS insertion $\ins=(r,p,d_z,\nu,i,\numstopsnu)$ found in $\Gdown_v$ will have smaller costs with $d_1$ than with $d_2$.
In particular, the two conditions for domination cover the two possibilities that the combination of a pickup $p$ and stop index $i$ do or do not lead to violations of the trip time soft constraint.

\begin{theorem}
  If a dropoff label $l_1$ dominates another label $l_2$ at $v$, we do not need to settle $l_2$ at $v$.
\end{theorem}

\begin{proof}
  This can be shown by contradiction.
  Assume $l_1 = (d_1, \ddown(v,d_1))$ dominates $l_2 = (d_2, \ddown(v,d_2))$ at $v \in V$.
  Further, assume that $\ins_2=(r,p,d_2,\nu,i,\numstopsnu)$ is the best DALS insertion.

  Let $\pi$ be a shortest path from $\laststopnu$ to $d_2$.
  Wlog. $\pi$ is an up-down path in the CH consisting of an upwards prefix $\pi^\uparrow$ and a downwards suffix $\pi^\downarrow$.
  If $\pi^\downarrow$ does not contain $v$, then the collective search will not find $\pi$ in $\Gdown_v$, and we do not have to settle $l_2$ at v.

  Otherwise, $\pi^\downarrow=(w,\dots,v,\dots,d_2)$ with $w \in \Vdown_v$.
  Since $\pi$ is a shortest path, we know that
  \begin{align*}
    \distveh(\laststopnu, d_2) &= \dup(\laststopnu, w) + \ddown(w,v) + \ddown(v,d_2) \text{ and } \\
    \distveh(\laststopnu, d_1) &\le \dup(\laststopnu,w) + \ddown(w,v) + \ddown(v,d_1) \text{.}
  \end{align*}
  Consequently, 
  \[ \distveh(\laststopnu, d_1) - \distveh(\laststopnu, d_2) \le \ddown(v,d_1) - \ddown(v,d_2) \: (\ast) \text{.}\]
  We consider the insertion $\ins_1=(r,p,d_1,\nu,i,\numstopsnu)$ and the cost difference $c(\ins_1) - c(\ins_2)$.
  Since the pickup detour of $\ins_1$ and $\ins_2$ are equal, we have $\ttripplus(\ins_1) = \ttripplus(\ins_2)$ and $\cwaitvio(\ins_1) = \cwaitvio(\ins_2)$.
  Comparing the cost of both insertions then yields
  \begin{align*}
    c(\ins_1) - c(\ins_2) 
                          = \; & \tdetour(\ins_1) - \tdetour(\ins_2) + \\
                            & \tripweight (\ttrip(\ins_1) - \ttrip(\ins_2)) + \\
                            & \walkweight (\twalk(\ins_1) - \twalk(\ins_2)) + \\
                            & \ctripvio(\ins_1) - \ctripvio(\ins_2) \\
                        \overset{(\ast)}{\le} \; & \Delta(l_1, l_2) + \\
                            & \tripweight \deltattrip(l_1, l_2) + \\
                            & \walkweight \deltatwalk(l_1,l_2) + \\
                            & \ctripvio(\ins_1) - \ctripvio(\ins_2)
  \end{align*}

  In case $\ttrip(\ins_1) - \ttrip(\ins_2) \le 0$, we have the upper bound $\ctripvio(\ins_1) - \ctripvio(\ins_2) \le 0$.
  Otherwise, we get the upper bound $\ctripvio(\ins_1) - \ctripvio(\ins_2) \le \gammatrip \deltattrip(l_1, l_2)$.
  Since $l_1$ dominates $l_2$, both cases yield $c(\ins_1) - c(\ins_2) < 0$.
  This contradicts $\ins_2$ being the best DALS insertion.
  Hence, we do not have to settle $l_2$ at $v$.
\end{proof}

For each vehicle $\nu \in F$, the collective BCH search finds a set of dropoffs $D(\nu)$ and the distances $\distveh(\laststopnu, d)$ for $d \in D(\nu)$ s.t. $D(\nu)$ contains the best dropoff for every possible combination of pickup $p \in \Prho$ and stop index $0 \le i < \numstopsnu$.

There is not necessarily a single best dropoff for each vehicle as different combinations of $p$ and $i$ may or may not lead to a violation of the trip time soft constraint.
The possible penalty term in the cost function means that the trip time may be differently weighted for different $p$ and $i$.
In particular, the dropoff walking time $\distpsg(d)$ may have a larger impact on the insertion cost for some combinations of $p$ and $i$ and a smaller impact for others.
In effect, $D(\nu)$ is a set of dropoffs that are pareto-optimal for $\nu$ w.r.t. the costs of insertions with and without trip time penalties.
This is also the reason why the domination relation for dropoff labels is a partial relation.

We find that the sets $D(\nu)$ remain very small.
The average size of $D(\nu)$ for vehicles $\nu \in F$ with $D(\nu) \ne \emptyset$ is only around $1.15$ for our tested instances.

\section{Experimental Evaluation}
\label{sec:experimental_evaluation}
\begin{table*}[t]
  \centering
  \caption{Key figures of our benchmark instances.
  Shows number of vertices ($|V|$), edges ($|E|$), vehicles (\#veh.), and requests (\#req.).
  Additionally, shows average number (rounded down) of pickups ($\numpickups$) and dropoffs ($\numdropoffs$) for walking radius $\rho \in \{0\s, 150\s, 300\s, 450\s, 600\s \}$ on the \BerlinOne and \BerlinTen instances, and $\rho \in \{ 0\s, 300\s, 600\s \}$ on the \RuhrOne and \RuhrTen instances.}
  \begin{tabular}{cRRRRRRRRRRRRRR}
    \toprule
    \multicolumn{5}{c}{} & \multicolumn{2}{c}{$\rho=0\s$} & \multicolumn{2}{c}{$\rho=150\s$} & \multicolumn{2}{c}{$\rho=300\s$} & \multicolumn{2}{c}{$\rho=450\s$} & \multicolumn{2}{c}{$\rho=600\s$} \\
    \cmidrule(lr){6-7}\cmidrule(lr){8-9}\cmidrule(lr){10-11}\cmidrule(lr){12-13}\cmidrule(l){14-15}
    Instance & \makecell[c]{|V|} & \makecell[c]{|E|} & \makecell[c]{\#\text{veh.}} & \makecell[c]{\#\text{req.}} & \makecell[c]{$\numpickups$} & \makecell[c]{$\numdropoffs$}  & \makecell[c]{$\numpickups$} & \makecell[c]{$\numdropoffs$} & \makecell[c]{$\numpickups$} & \makecell[c]{$\numdropoffs$}  & \makecell[c]{$\numpickups$} & \makecell[c]{$\numdropoffs$} & \makecell[c]{$\numpickups$} & \makecell[c]{$\numdropoffs$} \\
    \midrule
    \ShortBerlinOne & 94422 & 193212 & 1000  & 16569 & 1 & 1 & 12 & 12 & 44 & 44 & 100 & 99 & 178 & 178 \\
    \midrule
    \ShortBerlinTen & 94422 & 193212 & 10000 & 149185 & 1 & 1 & 13 & 13 & 46 & 46 & 103 & 104 & 183 & 186 \\
    \midrule
    \ShortRuhrOne & 420700 & 887790 & 3000 & 49707 & 1 & 1 & & & 40 & 39 & & & 137 & 136 \\
    \cmidrule(r){1-7}\cmidrule(lr){10-11}\cmidrule(l){14-15}
    \ShortRuhrTen & 420700 & 887790 & 30000 & 447555 & 1 & 1 & & & 40 & 39 & & & 137 & 136 \\
    \bottomrule
  \end{tabular}
  \label{tab:key_figures_of_instances}
\end{table*}
Our source code\footnote{Available at \url{https://github.com/molaupi/karri}.} is written in C++17 and compiled with GCC 9.4 using \texttt{-O3}.
We run our experiments on a machine with Ubuntu 20.04, 512 GiB of memory and two 16-core Intel Xeon E5-2683 v4 processors at 2.1GHz.
We use 32-bit distance labels and the AVX2 SIMD instruction set with 256-bit registers to compute up to $8$ operations in one vector instruction.

We evaluate \myalg on the \BerlinOne (\ShortBerlinOne), \BerlinTen (\ShortBerlinTen), \RuhrOne (\ShortRuhrOne), and \RuhrTen (\ShortRuhrTen) request sets~\cite{Buchhold2021} that respectively represent 1\% and 10\% of taxi sharing demand in the Berlin and Rhein-Ruhr metropolitan areas on a weekday. 
The request sets for Berlin were artificially generated using the Open Berlin Scenario~\cite{Ziemke2019} for the MATSim transport simulation~\cite{Horni2016}%
\footnote{MATSim generates realistic demand data but considering more than 10\% of taxi sharing demand would take processing times in the order of multiple months. 
For details, see~\cite{Buchhold2021}.}%
.
For the Rhein-Ruhr request sets, the request times as well as the pickup and dropoff locations are randomly drawn according to distributions that lead to similar trip times and request density as the Berlin request sets (for details, see~\cite{Buchhold2021}). 
The underlying road networks are obtained from OpenStreetMap data\footnote{\url{https://download.geofabrik.de/}, accessed Oct 30th 2023.}.
We use the known speed limit of each road to determine the travel time of the according edge in the vehicle network.
For the passenger network, we assume a constant walking speed of $4.5$km/h. 
We use the open-source library RoutingKit\footnote{\url{https://github.com/RoutingKit}, accessed Oct 30th 2023.} to compute the contraction hierarchies of the road networks which takes less than a minute for our instances.

We consider walking as a mode of local transportation for riders.
We scale the number of pickups $\numpickups$ and dropoffs $\numdropoffs$ by using increasing walking radii $\rho \in \{0\s, 150\s, 300\s, 450\s, 600\s\}$.
We show the number of vertices, edges, vehicles, and requests as well as the numbers of pickups and dropoffs for different walking radii in~\cref{tab:key_figures_of_instances}. 
We run five iterations of every experiment, and report average running times. 

For our cost function (see~\cref{eq:cost_function}), we adopt a basic ``time is money'' approach.
We use $\tripweight = 1$ to weight the time of a driver and a rider equally.
By setting $\omega=0$ we do not penalize walking over driving.
This choice maximizes the effect of meeting points on vehicle detours. 
In accordance with the MATSim transport simulation, we choose $\alpha=1.7$ and $\beta=2\text{min}$ which means that each trip may take up to a maximum trip time of $1.7 \distveh(\orig, \dest) + 2\text{min}$.
For the remaining parameters, we choose $\twaitmax = 600\s$, $\gammawait = 1$, and $\gammatrip = 10$.

\subsection{Bundled Searches}
\label{subsec:bundled_searches_experiments}
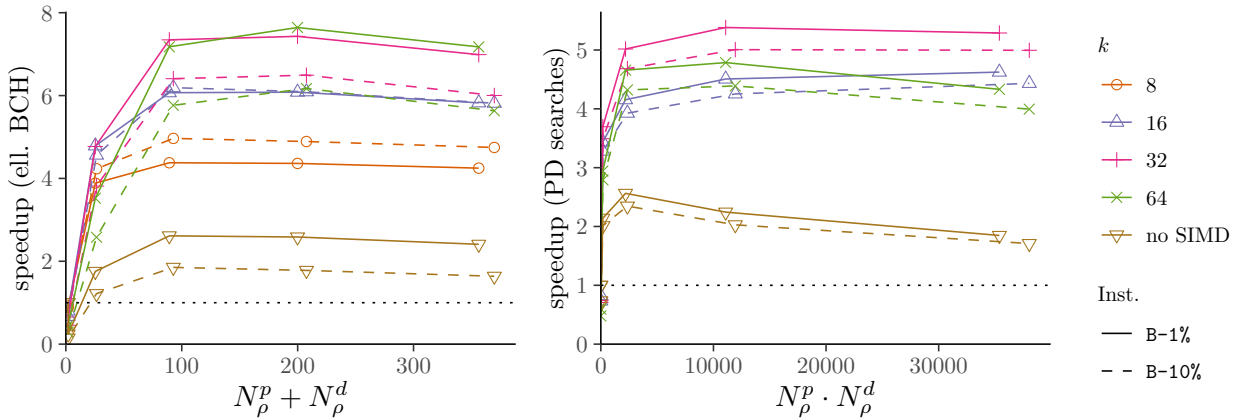
\begin{figure*}[t]
  \centering
    \input{ell_bch_and_pd_bundled_speedup_eval}%

  \vspace*{-0.8cm}
  \caption{Mean speedups for bundling with SIMD instructions for elliptic BCH searches (left) and PD-distance searches (right) on the \BerlinOne and \BerlinTen instances with $\rho \in \{0\s, 150\s, 300\s, 450\s, 600\s\}$.
  Considers $k \in \{8,16,32,64\}$ for elliptic BCH searches and $k \in \{16,32,64\}$ for PD-distance searches.
  Additionally shows running times without SIMD instructions with $k=32$.
  Note the different $y$-axes.}
  \label{fig:ell_bch_and_pd_bundled_eval_plots}
\end{figure*}
\begin{figure*}[t]
  \centering
    \input{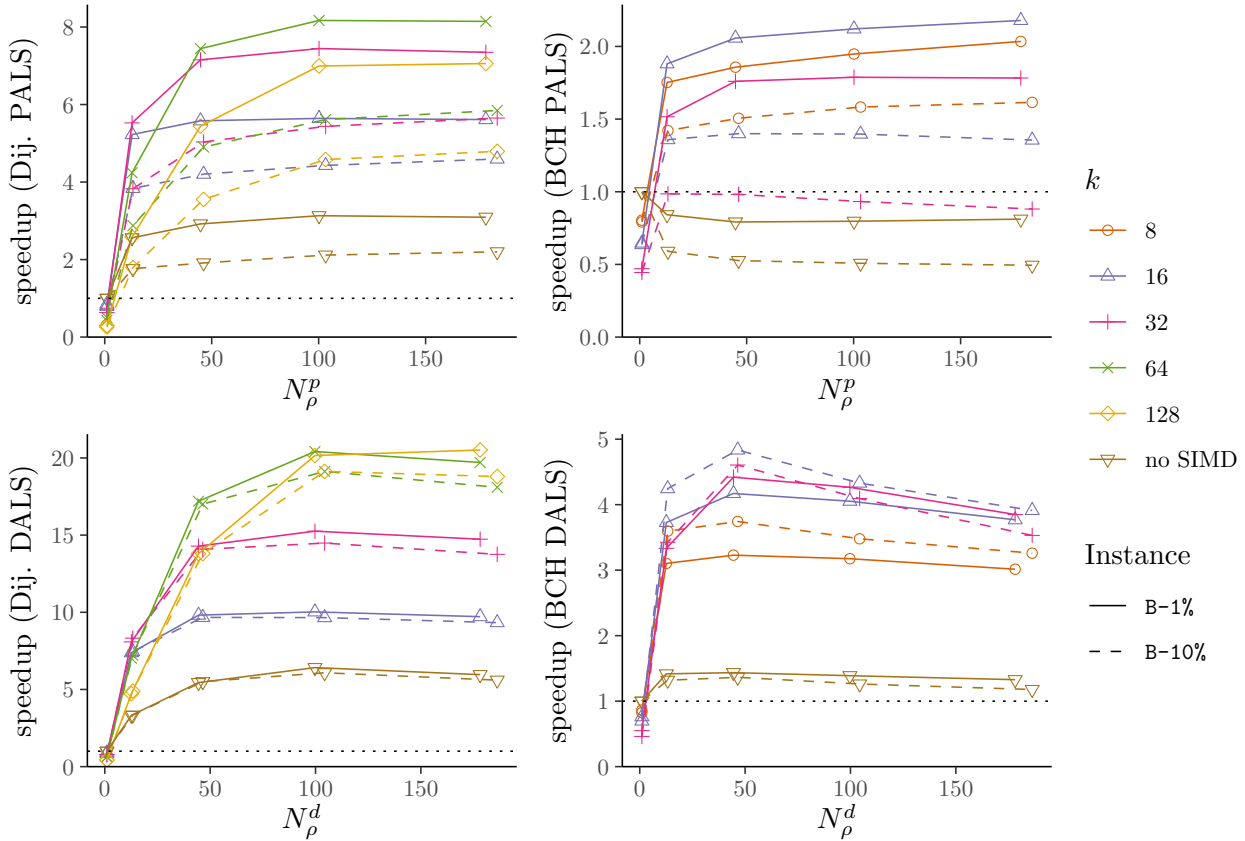}%

  \vspace*{-0.8cm}
  \caption{Mean speedups for bundling with SIMD instructions for Dijkstra searches (left) and individual BCH searches (right) in the PALS (top) and DALS (bottom) cases on the \BerlinOne and \BerlinTen instances with $\rho \in \{0\s, 150\s, 300\s, 450\s, 600\s\}$.
  Considers $k \in \{16,32,64,128\}$ for Dijkstra searches and $k \in \{8,16,32\}$ for BCH searches.
  Additionally shows speedups for bundling without SIMD instructions for Dijkstra searches with $k=64$ and BCH searches with $k=8$.
  Note the different $y$-axes.}
  \label{fig:last_stop_bundled_eval_plots}
\end{figure*}
In this section, we experimentally evaluate bundled searches in each of the described applications and find the optimal value of $k$ for each of them.
We conduct our experiments on the \BerlinOne and \BerlinTen instances with $\rho \in \{0\s, 150\s,300\s,450\s,600\s\}$.

\parheader{Bundled Elliptic BCH Searches and PD-Distance Searches}
We show experimental speedups for bundled elliptic BCH searches and bundled PD-distance searches in~\cref{fig:ell_bch_and_pd_bundled_eval_plots}.

For both search types, we find that bundled searches with $k=32$ lead to good speedups of between $5$ and $7$ with vector instructions and between $1.6$ and $2.4$ without vector instructions.
As both search types explore the entire CH search space of each source, a lot of work is performed in the periphery of sources.
Since the sources are close to one another, their search trees grow identical at larger distances which enables effective bundling.
However, larger values of $k$ lead to overheads for bundled edge relaxations and bucket entry scans closer to the sources that may not be bundled well.
The value $k=32$ strikes a balance between these two aspects.  

\parheader{Bundled Last Stop Searches}
We depict speedups for bundled Dijkstra searches and individual BCH searches for the PALS and DALS cases in~\cref{fig:last_stop_bundled_eval_plots}.

We find that Dijkstra searches are well suited for bundling as we observe the smallest search times with $k=64$ or in some cases even $k=128$.
Since Dijkstra searches do not use shortcut edges, the searches for each individual source meet much earlier than BCH searches.
Thus, the vast majority of the large number of edge relaxations of Dijkstra searches can be bundled well.
This is evidenced by the fact that we see good speedups for bundled Dijkstra searches even without SIMD instructions.
Larger $k > 64$ may be useful for larger numbers of sources but eventually we will run into cache limitations as hundreds of bytes of distance labels need to be handled per vertex. 

Contrarily, individual last stop BCH searches cannot be bundled as well due to two opposing properties:
Firstly, most work is performed close to the sources.
With sorted buckets, more bucket entries are scanned at vertices closer to the sources.
Additionally, the cost based stopping criterion of last stop BCH searches limits the search radius.
Secondly, due to the usage of shortcut edges, the search trees of individual searches only overlap at larger distances from the sources. 
Thus, edge relaxations and bucket entry scans cannot be bundled well in the proximity of the sources.
In effect, most work performed by individual last stop BCH searches is not well suited for bundling.

These factors have a stronger impact in the PALS case, as the stopping criterion is more effective (see~\cref{sec:dropoff_after_last_stop_insertions}).
Thus, in the PALS case, bundling only achieves speedups of $2.17$ for $k=16$ on \BerlinOne and $1.62$ for $k=8$ on \BerlinTen.
In fact, bundled searches without vector instructions are slower than non-bundled searches in the PALS case.
In the DALS case, a larger search radius is explored which allows better bundling with speedups of $3.77$ and $3.99$ for $k=16$ on \BerlinOne and \BerlinTen, respectively.

\subsection{Sorted Buckets}
\label{subsec:sorted_buckets_experiments}
\begin{figure}[t]
  \centering
    \input{elliptic_sorted_buckets_speedup_eval}%

  \vspace*{-0.4cm}
  \caption{Mean speedups of elliptic BCH searches ($k=1$) with sorted buckets over unsorted buckets on the \BerlinOne and \BerlinTen instances for $\rho \in \{ 0\s, 150\s, 300\s, 450\s, 600\s \}$.}
  \label{fig:elliptic_sorted_buckets_eval}
\end{figure}
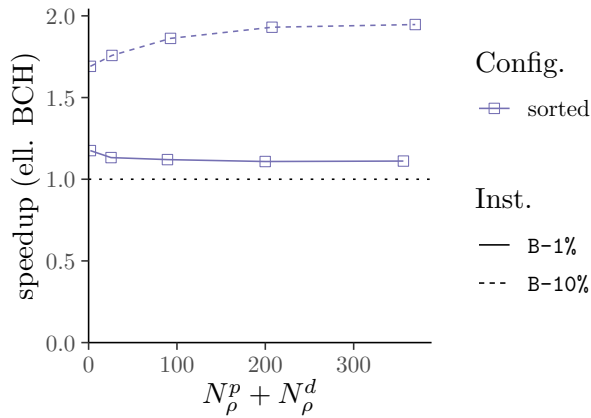
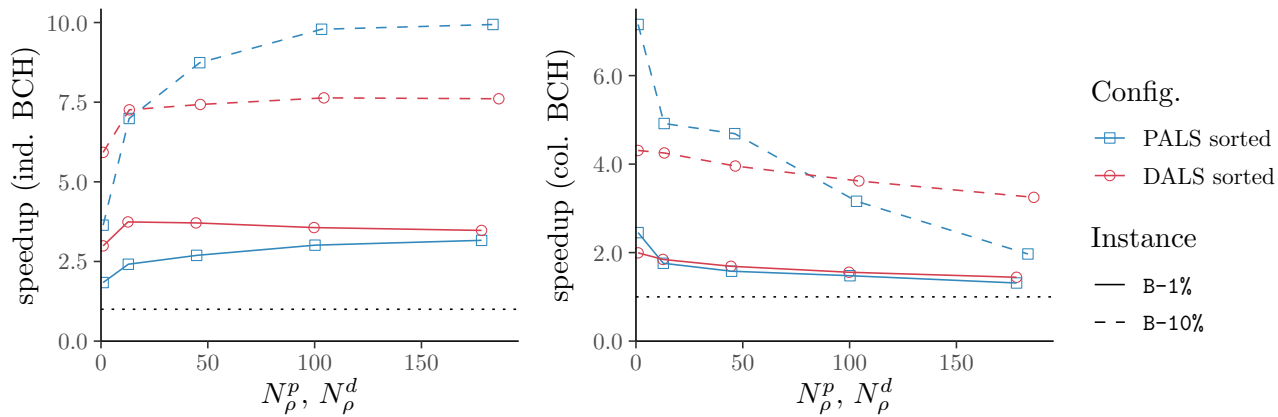
\begin{figure*}[t]
  \centering
    \input{last_stop_sorted_speedups}%

  \vspace*{-0.8cm}
  \caption{Mean speedups for individual (left, $k=1$) and collective (right) BCH queries with sorted buckets over unsorted buckets. 
  Considers the PALS and DALS cases on the \ShortBerlinOne and \ShortBerlinTen instances for $\rho \in \{ 0\s, 150\s, 300\s, 450\s, 600\s \}$.
  Note the different $y$-axes.}
  \label{fig:last_stop_sorted_buckets_eval}
\end{figure*}
In the following, we analyze the effect of sorted buckets on elliptic BCH searches as well as individual and collective last stop BCH searches.
We consider the reduction in the number of bucket entries scanned as well as the effects on the running time of the searches and the time for updating buckets.
We experimentally compare all searches with sorted and unsorted buckets on the \BerlinOne and \BerlinTen instances with $\rho \in \{ 0\s, 150\s, 300\s, 450\s, 600\s \}$ and $k=1$.

\parheader{Sorted Buckets for Elliptic BCH Searches}
The buckets for elliptic BCH searches are already strongly pruned using elliptic pruning.
Therefore, sorting these buckets only elicits a major effect with a sufficiently large number of vehicles. 
As we can see in~\cref{fig:elliptic_sorted_buckets_eval}, sorted buckets only have a limited impact for the \BerlinOne instance but a much larger one for the \BerlinTen instance as the latter considers ten times more vehicles.
On the larger input, sorted buckets reduce the number of entries scanned by about half, which leads to a decrease in the search time by up to $48\%$ ($37\ms$).
At the same time, maintaining the order of bucket entries increases the time for updating bucket entries by only about $32\mus$.
In conclusion, sorted buckets are a valuable improvement for elliptic BCH searches, particularly with respect to the scalability to larger numbers of vehicles. 

\parheader{Sorted Buckets for Last Stop BCH Searches}
In the following, we analyze the effect of sorted buckets on individual and collective last stop BCH searches.
We experimentally evaluate both searches with sorted and unsorted buckets on the \BerlinOne and \BerlinTen instances with $\rho \in \{ 0\s, 150\s, 300\s, 450\s, 600\s \}$ and $k=1$.
The speedups achieved with sorted buckets are shown in~\cref{fig:last_stop_sorted_buckets_eval}.

For last stop BCH searches, sorted buckets are vital to reduce the number of bucket entries scanned since we cannot use elliptic pruning.
For individual BCH searches, more than $97\%$ and $89\%$ fewer bucket entries are scanned with sorted buckets in the PALS and DALS cases, respectively.
This reduces search times by factors of up to $9.09$ and $7.14$.

For collective searches, the number of bucket entries scanned decreases by similar rates of $97\%$ and $87\%$.
However, the resulting speedups are less pronounced, particularly for larger numbers of meeting points in the PALS case. 
We attribute this to the fact that collective searches spend comparatively more time on pruning the searches.
This means that the searches need to spend less time scanning bucket entries, which limits the impact of sorted buckets.
Notably, collective PALS searches generate initial labels for every PD-pair but prune almost all of them immediately.
As the number of PD-pairs is proportional to $\rho^4$, this initialization can constitute up to $85\%$ of the search time for larger values of $\rho$ but sorted buckets have no effect on it.

Consequently, we observe speedups of only $1.96$ (PALS) and $3.22$ (DALS) for $\rho=600\s$ on \BerlinTen.
If we disregard the overhead for initial labels, these speedups increase to $7.51$ and $3.63$.

Maintaining sorted last stop buckets incurs an average overhead per request of about $10\mus$ for \BerlinOne and about $35\mus$ for \BerlinTen while the reduction in search time is one to three orders of magnitude larger.

\parheader{Collective BCH Searches}
\label{par:collective_bch_searches_experiments}
\begin{table*}[t!]
  \centering
  \caption{
  Comparison of the PALS and DALS running times (in \mus) of collective BCH searches (Coll.), individual BCH searches (BCH), and Dijkstra searches (Dij.) in their optimal configurations for three radii $\rho \in \{ 0\s, 300\s, 600\s \}$ on the \ShortBerlinOne and \ShortBerlinTen instances.
  Shows average number of edge relaxations ($\#_{\text{rel.}}$), number of bucket entries scanned ($\#_{\text{scans}}$), search time ($t_{\text{search}}$) and time for enumerating insertions ($t_{\text{enum}}$) per request.
  The smallest times per radius are marked in bold.
  }
  \begin{tabular}{cccRRRRRRRR}
  \toprule
  \multicolumn{3}{c}{} & \multicolumn{4}{c}{\BerlinOne} & \multicolumn{4}{c}{\BerlinTen} \\
  \cmidrule(lr){4-7}\cmidrule(l){8-11}
  Type & \makecell[c]{$\rho$} & Search & \makecell[c]{\#_{\text{rel.}}} & \makecell[c]{\#_{\text{scans}}} & \makecell[c]{t_{\text{search}}} & \makecell[c]{t_{\text{enum}}} & \makecell[c]{\#_{\text{rel.}}} & \makecell[c]{\#_{\text{scans}}} & \makecell[c]{t_{\text{search}}} & \makecell[c]{t_{\text{enum}}}\\
  \midrule
  \multirow{9}{*}{PALS} & \multirow{3}{*}{0} & Coll.  & 40 & 8 & 4.96 & \mathbf{0.04} & 19 & 11 & 3.89 & \mathbf{0.06}  \\
  & & BCH  & 37 & 8 & \mathbf{3.63} & 0.37 & 18 & 10 & \mathbf{3.11} & 0.56 \\
  & & Dij.  & 577 & 2 & 43.44 & 0.36 & 225 & 4 & 18.75 & 0.52 \\
  \cmidrule{2-11}
  & \multirow{3}{*}{300} & Coll.  & 412 & 57 & \mathbf{70.16} & \mathbf{0.08} & 168 & 58 & \mathbf{41.95} & \mathbf{1.49} \\
  & & BCH  & 967 & 274 & 73.95 & 44.12 & 797 & 1011 & 103.36 & 155.86  \\
  & & Dij. & 4302 & 17 & 497.94 & 40.4 & 3533 & 99 & 433.66 & 138.16 \\
  \cmidrule{2-11}
  & \multirow{3}{*}{600} & Coll. & 806 & 108 & \mathbf{286.11} & \mathbf{0.09} & 219 & 82 & \mathbf{213.9} & \mathbf{23.90} \\
  & & BCH & 5555 & 2514 & 424.65 & 812.64 & 4734 & 12214 & 823.72 & 3475.00 \\
  & & Dij. & 41092 & 137 & 4481.08 & 812.24 & 38102 & 960 & 4412.38 & 3098.09  \\
  \midrule
  \multirow{9}{*}{DALS} & \multirow{3}{*}{0} & Coll. & 210 & 676 & 36.56 & 0.76 & 191 & 5066 & 95.37 & \mathbf{2.74}  \\
  & & BCH  & 216 & 721 & \mathbf{23.78} & \mathbf{0.72} & 197 & 5419 & \mathbf{87.14} & 2.96 \\
  & & Dij.  & 19063 & - & 1665.24 & 15.60 & 14920 & - & 1344.48 & 58.98 \\
  \cmidrule{2-11}
  & \multirow{3}{*}{300} & Coll.  & 253 & 662 & \mathbf{58.04} & \mathbf{5.20} & 235 & 5049 & \mathbf{117.04} & \mathbf{20.97} \\
  & & BCH  & 2015 & 4116 & 182.81 & 145.01 & 1961 & 32487 & 623.84 & 596.46  \\
  & & Dij. & 26567 & - & 3561.61 & 232.62 & 22228 & - & 3014.30 & 984.84 \\
  \cmidrule{2-11}
  & \multirow{3}{*}{600} & Coll. & 296 & 656 & \mathbf{93.11} & \mathbf{13.84} & 277 & 5091 & \mathbf{157.07} & \mathbf{53.61} \\
  & & BCH & 8042 & 14602 & 685.23 & 1227.00 & 7683 & 115912 & 2214.15 & 4261.21 \\
  & & Dij. & 98143 & - & 12453.92 & 3063.71 & 88021 & - & 11307.36 & 12665.75 \\
  \bottomrule
  \end{tabular}
  \label{tab:last_stop_search_comparison}
\end{table*}
In~\cref{tab:last_stop_search_comparison}, we compare the search times and the times needed to enumerate candidate insertions for the three search approaches used for the PALS and DALS insertion types.
Additionally, we show the number of relaxed edges and scanned bucket entries.
We report the results for $\rho \in \{ 0\s, 300\s, 600\s \}$ on the \ShortBerlinOne and \ShortBerlinTen instances.
We use the optimal configuration for each combination of search type, insertion type, and radius.

At $\rho=0\s$, collective searches are slower than individual BCH searches as there is only a single pickup and dropoff so the overhead for explicitly maintaining labels instead of a single distance per vertex is unwarranted.
At $\rho=300\s$ and $\rho=600\s$, collective searches offer the best search times, though.
In the PALS case, collective searches are up to $4$ times faster than individual BCH searches.
In the DALS case, this relative speedup is even larger at up to $14$.
We attribute the better scalability of collective searches to two main advantages:

Firstly, collective searches can be pruned more precisely as we use lower bounds on the cost of specific PD-pairs or dropoffs instead of a general lower bound on the cost of every PD-pair or dropoff.
This applies to the stopping criteria for bucket scans and for the searches as a whole.

Secondly, collective searches consider all sources in one search, maximizing the amount of information available for domination pruning.
Bundled searches can only consider $k$ searches at once which means work may be repeated up to $\numpickups / k$ times.
Thus, the number of edge relaxations and bucket entry scans increases much faster with the number of pickups ($ \numpickups, \numdropoffs \sim \rho^2$) for individual BCH searches than for collective BCH searches.

In addition, the enumeration times remain small for collective searches while they increase massively with $\rho$ for individual BCH searches.
This is due to the fact that collective searches identify a single candidate insertion in the PALS case or a small set of candidate vehicles and dropoffs in the DALS case.
Contrarily, individual BCH searches first find all distances and then enumerate an insertion for each combination of candidate vehicle, pickup and dropoff.
As the number of PD-pairs is proportional to $\rho^4$, enumeration times of individual BCH searches quickly become very large with tens to hundreds of thousands of insertions tried. 

Collective searches scale worse in the PALS case compared to the DALS case because many more initial labels need to be created.
In the PALS case, one label is initialized for every PD-pair but almost all of these labels are pruned immediately based on cost or domination pruning.
For instance, an average of $99.6\%$ of initial labels are discarded right away for $\rho=600\s$ on the \ShortBerlinTen instance.
Checking and pruning these labels makes up $85\%$ of the search time of collective searches in the PALS case.
In the DALS case, we only need to initialize one label per dropoff, which vastly reduces this overhead for pruning initial labels.

\mysubsection{Comparison with Baseline Dispatcher}
\label{subsec:comparison_with_baseline_dispatcher}
\begin{table*}[t]
  \centering
  \caption{Running times (in \mus) of different phases of the baseline (\tabbase), na\"ively extended baseline (\tabbaseast), \myalg (\tabmyalg), and \myalgcch (\tabmyalgcch) with different radii ($\rho \in \{0\s,300\s,600\s\}$) on \ShortBerlinOne, \ShortBerlinTen, \ShortRuhrOne, and \ShortRuhrTen.
  Shows mean times for finding $\Prho$ and $\Drho$, PD-distance searches, elliptic BCH searches, enumerating ordinary and PBNS insertions, PALS and DALS searches, and updating routes and buckets as well as the mean total time per request.
  }
  \begin{tabular}{cRcRRRRRRRR}
    \toprule
    \makecell[c]{Inst.} & \makecell[c]{\rho} & \makecell[c]{Alg.}  & \makecell[c]{\text{find} \\ \Prho, \Drho} & \makecell[c]{\text{PD}} & \makecell[c]{\text{Ell.} \\ \text{BCH}} & \makecell[c]{\text{Ord.\&}\\\text{PBNS}} & \makecell[c]{\text{PALS}} & \makecell[c]{\text{DALS}} & \makecell[c]{\text{update}} & \makecell[c]{\text{total}} \\
    \midrule
    \multirowcell{6}{\ShortBerlinOne} & \multirowcell{3}{0} & \tabbase & 0 & 21 & 113 & 68 & 55 & 1682 & 97 & 2036 \\
    & & \tabmyalg & 2 & 74 & 115 & 41 & 4 & 24 & 151 & 413 \\
    & & \tabmyalgcch & 2 & 48 & 125 & 39 & 6 & 31 & 185 & 436\\
    \cmidrule{2-11}
    & \multirowcell{2}{300} & \tabmyalg & 173 & 300 & 617 & 151 & 72 & 63 & 154 & 1530 \\
    & & \tabmyalgcch & 171 & 236 & 596 & 141 & 74 & 103 & 188 & 1510 \\
    \cmidrule{2-11}
    & \multirowcell{2}{600} & \tabmyalg & 617 & 1536 & 2536 & 881 & 298 & 107 & 155 & 6129 \\
    & & \tabmyalgcch & 610 & 1325 & 2120 & 814 & 214 & 140 & 189 & 5411\\
    \midrule
    \multirowcell{9}{\ShortBerlinTen} & \multirowcell{3}{0} & \tabbase & 0 & 19 & 328 & 247 & 27 & 1361 & 94 & 2076 \\
    & & \tabmyalg & 3 & 73 & 351 & 346 & 4 & 88 & 245 & 1111 \\
    & & \tabmyalgcch & 3 & 48 & 474 & 338 & 6 & 145 & 638 & 1653 \\
    \cmidrule{2-11}
    & \multirowcell{3}{300} & \tabbaseast & 194 & 30600 & 20196 & 905 & 2275 & 54270 & 119 & 108559 \\
    & & \tabmyalg & 195 & 308 & 1770 & 783 & 51 & 138 & 246 & 3490 \\
    & & \tabmyalgcch & 202 & 247 & 2175 & 752 & 45 & 216 & 626 & 4262 \\
    \cmidrule{2-11}
    & \multirowcell{3}{600} & \tabbaseast & 716 & 513788 & 77813 & 3313 & 28901 & 220555 & 118 & 845204 \\
    & & \tabmyalg & 708 & 1662 & 6891 & 3227 & 312 & 211 & 249 & 13260 \\
    & & \tabmyalgcch & 722 & 1429 & 7470 & 3000 & 169 & 278 & 632 & 13701 \\
    \midrule
    \multirowcell{4}{\ShortRuhrOne} & \multirowcell{2}{0} & \tabbase &  0 & 18 & 156 & 224 & 103 & 7157 & 78 & 7735 \\
    & & \tabmyalg & 3 & 52 & 164 & 108 & 4 & 29 & 119 & 479  \\
    \cmidrule(l){2-11}
    & 300 & \tabmyalg & 139 & 225 & 668 & 178 & 69 & 57 & 126 & 1461 \\
    \cmidrule(l){2-11}
    & 600 & \tabmyalg & 449 & 945 & 2090 & 476 & 275 & 82 & 130 & 4446 \\
    \midrule
    \multirowcell{4}{\ShortRuhrTen} & \multirowcell{2}{0} & \tabbase & 0 & 17 & 617 & 929 & 61 & 5114 & 76 & 6814 \\
    & & \tabmyalg & 3 & 54 & 809 & 913 & 4 & 147 & 237 & 2167 \\
    \cmidrule(l){2-11}
    & 300 & \tabmyalg & 145 & 230 & 3381 & 1198 & 55 & 173 & 253 & 5433 \\
    \cmidrule(l){2-11}
    & 600 & \tabmyalg & 468 & 965 & 9586 & 2262 & 311 & 216 & 271 & 14078 \\
    \bottomrule
  \end{tabular}
  \label{tab:comparison_with_LOUD_run_times}
\end{table*}
\begin{table}[t]
  \centering
  \caption{Solution quality of \myalg with different radii ($\rho \in \{0\s,300\s,600\s\}$) on \ShortBerlinOne, \ShortBerlinTen, \ShortRuhrOne, and \ShortRuhrTen.
  For riders, we report the average wait and trip times (in mm:ss).
  For vehicles, we give the average occupancy while driving and average total operation time (in hh:mm).}
  \begin{tabular}{crrrrr}
    \toprule
    \makecell[c]{Inst.} & \makecell[c]{$\rho$} & \makecell[c]{\text{wait}} & \makecell[c]{\text{trip}} & \makecell[c]{\text{occ}} & \makecell[c]{\text{op}} \\
    \midrule
    \multirowcell{3}{\ShortBerlinOne} & 0 & 3:44 & 16:53 & 0.88 & 4:28 \\
    & 300 & 3:15 & 15:54 & 0.93 & 4:00 \\
    & 600 & 3:27 & 16:02 & 0.94 & 3:54 \\
    \midrule
    \multirowcell{3}{\ShortBerlinTen} & 0 & 2:40 & 15:34 & 1.06 & 3:14 \\
    & 300 & 2:36 & 15:21 & 1.20 & 2:44 \\
    & 600 & 2:53 & 15:40 & 1.24 & 2:38 \\
    \midrule
    \multirowcell{3}{\ShortRuhrOne} & 0 & 4:32 & 18:05 & 0.80 & 4:29 \\
    & 300 & 3:57 & 16:41 & 0.83 & 4:06 \\
    & 600 & 4:02 & 16:29 & 0.84 & 4:01 \\
    \midrule
    \multirowcell{3}{\ShortRuhrTen} & 0 & 2:52 & 15:17 & 0.92 & 3:24 \\
    & 300 & 2:29 & 14:24 & 0.98 & 3:03 \\
    & 600 & 2:38 & 14:29 & 1.00 & 2:58 \\
    \bottomrule
  \end{tabular}
  \label{tab:comparison_with_LOUD_solution_quality}
\end{table}
In this section, we compare \myalg with the baseline dispatcher by~\citet{Buchhold2021}.
For this comparison, we implemented the extended \myalg cost function in the baseline dispatcher.

\parheader{Running Times} 
We give the running times for the different phases of both our algorithm (\tabmyalg) and the baseline (\tabbase) on all instances in~\cref{tab:comparison_with_LOUD_run_times}.

First, we consider the scenario without meeting points ($\rho=0\s$) and compare \myalg with the baseline dispatcher.
Here, sorted buckets have no positive impact on the search times of elliptic BCH searches even though the number of bucket entries scanned is reduced.
We attribute this to the fact that our implementation is meant to deal with any number of meeting points while the baseline is specialized for the case of $\numpickups=\numdropoffs=1$.
Our last stop BCH searches are well suited for $\rho=0\s$, though.
They are up to one and two orders of magnitude faster than the baseline Dijkstra searches in the PALS and DALS cases, respectively.
In the baseline, last stop searches make up at least $66\%$ and up to $93\%$ of the total running time while our last stop searches make up at most $8\%$ of the total running time of \myalg.
Note that maintaining sorted buckets does lead to increased update times, though.
In total, we can reduce the average time per request by factors of about $5$, $2$, $16$, and $3$ for the \ShortBerlinOne, \ShortBerlinTen, \ShortRuhrOne, and \ShortRuhrTen instances, respectively, compared to the baseline.

Unfortunately, the source code of \myalg's closest existing competitor \starsplus~\cite{Mounesan2021} is currently not publicly available, making an experimental comparison of both approaches difficult. 
Instead, for now, we validate the effectiveness of our approach by comparing it with a na\"ive extension of the techniques used by our baseline algorithm.
For this, we configured \myalg to use no bundled searches or sorted buckets, to use Dijkstra searches for the PALS and DALS cases, and to use point-to-point CH queries to compute PD-distances.
We report the running times of this extension (\tabbaseast) for $\rho=300\s$ and $\rho=600\s$ on the \ShortBerlinTen instance and compare them to \myalg.

We find that bundling the elliptic BCH searches and using sorted buckets makes them about one order of magnitude faster than the na\"ive extension.
PD-distances can be computed around two orders of magnitude faster with our bucket based approach than with individual CH queries.
Our collective searches for the PALS and DALS cases beat the na\"ive approach by two and three orders of magnitude, respectively. 

\parheader{Running Times with CCHs}
We also equipped \myalg with the possibility to use customizable CHs (CCHs)~\cite{Dibbelt2016}, a technique that allows a CH to quickly be adapted to changing travel times in the road network at the cost of a slight increase in shortest path query times.
This extension to \myalg was straightforward and all of our many-to-many routing techniques can still be used.
We experimentally evaluated \myalgcch by running it on \BerlinOne and \BerlinTen with $\rho \in \{ 0\s, 300\s, 600\s \}$ in the same configurations as with standard CHs.
The resulting running times are shown in~\cref{tab:comparison_with_LOUD_run_times} (\tabmyalgcch).

As expected, the running times of \myalg increase slightly when using CCHs.
In particular, the time for updating buckets and the DALS search time always increase.
However, for most components of \myalg, we see similar or even smaller running times with CCHs compared to standard CHs.
We attribute this to the fact that we can employ elimination tree searches, a specialized type of query for CCHs that is often better suited for bundling than standard CH queries~\cite{Buchhold2019}.
This advantage of elimination tree searches is evidenced by the fact that when using CCHs our bundled elliptic BCH queries and PD-distance searches fare relatively better with increasing $\rho$ and number of meeting points $\numpickups$, $\numdropoffs$.

Overall, the running time of \myalg increases by less than $1\ms{}$ per request on average when using CCHs instead of standard CHs.
At the same time, CCHs allow us to adapt our CH to changed travel times in the road network in less than $100\ms{}$.
Thus, \myalgcch combines fast dispatching with the ability to react to changing traffic conditions in real time.

\parheader{Solution Quality}
In the following, we give a first idea of how trip times and vehicle operation times can be improved by extending taxi sharing with meeting points.

In~\cref{tab:comparison_with_LOUD_solution_quality}, we compare the solution quality of \myalg with $\rho \in \{ 0\s, 300\s, 600\s \}$.
Note that we also allow riders to walk to their destinations as an alternative to a taxi sharing trip (see~\cref{subsec:walking_time_and_walking_to_the_destination}).
At $\rho=300\s$, we observe improvements for both riders and vehicles.
Here, we expect that mostly existing waiting times are replaced with walking which leads to benefits for all agents.
If we allow longer walking distances at $\rho=600\s$, we can further improve the vehicle operation times.
However, since we equally weight vehicle and rider times ($\tripweight=1$), riders are often required to walk further to save time for vehicles, increasing the average wait and trip times.
Different values for the cost function parameters may be better suited to reflect the needs of riders, particularly in a future of autonomous taxis.
We defer an according analysis to future work.

\section{Conclusions and Future Work}
\label{sec:conclusion}
\myalg develops efficient many-to-many routing with bucket contraction hierarchies for dynamic taxi sharing.
This allows real-time dispatching systems to enjoy benefits like a reduction in operating costs and air pollution even with large vehicle fleets and many meeting points.
A flexible cost function allows configuration to many situations, e.g. using walking, bicycles or scooters. 
We expect that the new techniques like sorted buckets can also be applied for other problems that use many-to-many routing with correlated sources and targets.

\myalg's small running times open dynamic taxi sharing up to a variety of extensions that promise to improve the quality of service.
We are particularly interested in going away from greedy online scheduling, instead taking into account pre-booked trips and opportunities to transparently change existing trips for local search style optimizations.
Additionally, we expect that we can generalize \myalg to integrate it with public transportation s.t. meeting points can be stops of buses or trains and the cost function has to take into account the public transportation schedule.
A longer term perspective is to allow transfers between vehicles during a trip.
This may increase the number of shared rides, eventually leading to a highly adaptive software defined public transportation system. 
These extensions imply interesting algorithmic challenges as they lead to a combinatorial explosion of possible route options.

Future parallelization both over over different meeting points and over entire requests can improve scalability to even larger metropolitan regions.

%
%

\bibliography{karri}

\end{document}

%% file: macros.tex
\newcommand{\karri}[0]{KaRRi\xspace}
\newcommand{\karricch}[0]{KaRRi-CCH\xspace}

\newcommand{\starsplus}[0]{STaRS+\xspace}
\newcommand{\myalg}[0]{\karri}
\newcommand{\myalgcch}[0]{\karricch}
\newcommand{\parheader}[1]{\paragraph*{#1.}}
\newcommand{\mysubsection}[1]{\subsection{#1.}}

\newcommand{\tabmyalg}[0]{K}
\newcommand{\tabmyalgcch}[0]{K-CCH}
\newcommand{\tabbase}[0]{B}
\newcommand{\tabbaseast}[0]{B$^\ast$}

\NewEnviron{myAlign}{%
\vspace{-2mm}
\begin{align*}
  \BODY
\end{align*}
}

\newcommand{\setc}[2]{\left \{ #1 \mid #2 \right \}}
\renewcommand{\gets}[0]{:=}

\newcommand{\Vveh}[0]{\ensuremath{V}}

\newcommand{\Gpsg}[0]{\ensuremath{G_{\textit{psg}}}}
\newcommand{\Vpsg}[0]{\ensuremath{V_{\textit{psg}}}}
\newcommand{\Epsg}[0]{\ensuremath{E_{\textit{psg}}}}

\newcommand{\dist}[0]{\ensuremath{\delta}}
\newcommand{\distveh}[0]{\ensuremath{\dist}}
\newcommand{\distpsg}[0]{\ensuremath{\delta_{\textit{psg}}}}

\newcommand{\orig}[0]{\textit{orig}}
\newcommand{\dest}[0]{\textit{dest}}
\newcommand{\Prho}[0]{\ensuremath{P_{\rho}}}
\newcommand{\Drho}[0]{\ensuremath{D_{\rho}}}

\newcommand{\maxPDDist}[0]{\ensuremath{\dist^{\max}_{\textit{PD}}}}

\newcommand{\ins}[0]{\ensuremath{\iota}}

\newcommand{\numstops}[1]{\ensuremath{k(#1)}}
\newcommand{\numstopsnu}[0]{\numstops{\nu}}

\newcommand{\laststop}[1]{\ensuremath{s_{\numstops{#1}}}}
\newcommand{\laststopnu}[0]{\laststop{\nu}}

\newcommand{\stoploc}[0]{\ensuremath{\textit{loc}}}
\newcommand{\curloc}[0]{\ensuremath{\textit{loc}}}

\newcommand{\Ord}[0]{\emph{ordinary}\xspace}
\newcommand{\OP}[0]{\emph{ordinary paired}\xspace}
\newcommand{\PBNS}[0]{\emph{pickup before next stop}\xspace}
\newcommand{\PALS}[0]{\emph{pickup after last stop}\xspace}

\newcommand{\rank}[0]{\ensuremath{\text{rank}}}

\newcommand{\Gup}[0]{\ensuremath{G^{\uparrow}}}
\newcommand{\Gdown}[0]{\ensuremath{G^{\downarrow}}}
\newcommand{\Vup}[0]{\ensuremath{V^\uparrow}}
\newcommand{\Vdown}[0]{\ensuremath{V^\downarrow}}
\newcommand{\Eup}[0]{\ensuremath{E^\uparrow}}
\newcommand{\Edown}[0]{\ensuremath{E^\downarrow}}
\newcommand{\dup}[0]{\ensuremath{\delta^{\uparrow}}}
\newcommand{\ddown}[0]{\ensuremath{\delta^{\downarrow}}}

\newcommand{\tdist}[0]{\ensuremath{\tilde{\delta}}}

\newcommand{\tddown}[0]{\ddown}

\newcommand{\treq}[0]{\ensuremath{t_{\textit{req}}}}
\newcommand{\tarr}[0]{\ensuremath{t_{\textit{arr}}}}
\newcommand{\tdep}[0]{\ensuremath{t_{\textit{dep}}}}

\newcommand{\twaitmax}[0]{\ensuremath{t_{\textit{wait}}^{\textit{max}}}}
\newcommand{\ttripmax}[0]{\ensuremath{t_{\textit{trip}}^{\textit{max}}}}

\newcommand{\tdepmin}[0]{\ensuremath{t_{\textit{dep}}^{\textit{min}}}}

\newcommand{\tarrmin}[0]{\ensuremath{t_{\textit{arr}}^{\textit{min}}}}

\newcommand{\tvehwait}[0]{\ensuremath{w}}

\newcommand{\tarrminprime}[0]{\ensuremath{\tarrmin{}'}}
\newcommand{\tdepminprime}[0]{\ensuremath{\tdepmin{}'}}

\newcommand{\remleeway}[0]{\ensuremath{\lambda_{\textit{res}}}}

\newcommand{\tdetour}[0]{\ensuremath{t_{\textit{detour}}}}

\newcommand{\ttrip}[0]{\ensuremath{t_{\textit{trip}}}}

\newcommand{\ttripplus}[0]{\ensuremath{t_{\textit{trip}}^{+}}}

\newcommand{\twalk}[0]{\ensuremath{t_{\textit{walk}}}}

\newcommand{\ctripvio}[0]{\ensuremath{c_{\textit{trip}}^{vio}}}

\newcommand{\cwaitvio}[0]{\ensuremath{c_{\textit{wait}}^{vio}}}

\newcommand{\tripweight}[0]{\ensuremath{\tau}}
\newcommand{\walkweight}[0]{\ensuremath{\omega}}
\newcommand{\gammatrip}[0]{\ensuremath{\gamma_{\textit{trip}}}}
\newcommand{\gammawait}[0]{\ensuremath{\gamma_{\textit{wait}}}}

\newcommand{\initpdetour}[0]{\ensuremath{\blacktriangle_{p}}}
\newcommand{\initddetour}[0]{\ensuremath{\blacktriangle_{d}}}
\newcommand{\resdetour}[1]{\ensuremath{\triangle_{#1}}}

\newcommand{\cmaxglobal}[0]{\ensuremath{\hat{c}}} 
\newcommand{\cmax}[0]{\ensuremath{c_{\max}}}
\newcommand{\cmin}[0]{\ensuremath{c_{\min}}}

\newcommand{\Bs}[0]{\ensuremath{B^{\uparrow}}}
\newcommand{\Bt}[0]{\ensuremath{B^{\downarrow}}}

\newcommand{\Blast}[0]{\ensuremath{B^{\uparrow}}}

\newcommand{\open}[0]{\ensuremath{\textit{open}}}
\newcommand{\closed}[0]{\ensuremath{\textit{closed}}}

\newcommand{\bestp}[0]{\ensuremath{p^\ast}}
\newcommand{\bestd}[0]{\ensuremath{d^\ast}}

\newcommand{\bestPALSins}[0]{\ensuremath{\iota_{\text{pals}}^{\ast}}}

\newcommand{\deltacmax}[0]{\ensuremath{\Delta_{\textit{c}}^{\max}}}
\newcommand{\deltatdepatpickup}[0]{\ensuremath{\Delta_{\textit{dep}}}}
\newcommand{\deltadetour}[0]{\ensuremath{\Delta_{\textit{detour}}}}
\newcommand{\deltattrip}[0]{\ensuremath{\Delta_{\textit{trip}}}}
\newcommand{\deltatwalk}[0]{\ensuremath{\Delta_{\textit{walk}}}}
\newcommand{\deltawaitvio}[0]{\ensuremath{\Delta_{\textit{wait}}^{\textit{vio}}}}
\newcommand{\deltatripvio}[0]{\ensuremath{\Delta_{\textit{trip}}^{\textit{vio}}}}

%% file: shared_macros.tex
\newcommand{\BerlinOne}[0]{\texttt{Berlin-1pct}\xspace}
\newcommand{\BerlinTen}[0]{\texttt{Berlin-10pct}\xspace}
\newcommand{\RuhrOne}[0]{\texttt{Ruhr-1pct}\xspace}
\newcommand{\RuhrTen}[0]{\texttt{Ruhr-10pct}\xspace}
\newcommand{\ShortBerlinOne}[0]{\texttt{B-1\%}\xspace}
\newcommand{\ShortBerlinTen}[0]{\texttt{B-10\%}\xspace}
\newcommand{\ShortRuhrOne}[0]{\texttt{R-1\%}\xspace}
\newcommand{\ShortRuhrTen}[0]{\texttt{R-10\%}\xspace}

\newcommand{\s}[0]{\si{\second}}
\newcommand{\ms}[0]{\si{\milli\second}}
\newcommand{\mus}[0]{\si{\micro\second}}

\newcommand{\numpickups}[0]{\ensuremath{N_{\rho}^{p}}}
\newcommand{\numdropoffs}[0]{\ensuremath{N_{\rho}^{d}}}

%% file: ell_bch_and_pd_bundled_speedup_eval.tex
\begin{tikzpicture}[x=1pt,y=1pt]
\definecolor{fillColor}{RGB}{255,255,255}
\path[use as bounding box,fill=fillColor,fill opacity=0.00] (0,0) rectangle (505.89,162.61);
\begin{scope}
\path[clip] (  0.00,  0.00) rectangle (202.36,162.61);
\definecolor{drawColor}{RGB}{255,255,255}
\definecolor{fillColor}{RGB}{255,255,255}

\path[draw=drawColor,line width= 0.6pt,line join=round,line cap=round,fill=fillColor] (  0.00,  0.00) rectangle (202.36,162.61);
\end{scope}
\begin{scope}
\path[clip] ( 26.69, 31.25) rectangle (196.86,157.11);
\definecolor{fillColor}{RGB}{255,255,255}

\path[fill=fillColor] ( 26.69, 31.25) rectangle (196.86,157.11);
\definecolor{drawColor}{RGB}{166,118,29}

\path[draw=drawColor,line width= 0.4pt,line join=round,line cap=round] ( 27.57, 43.88) --
	( 30.21, 48.46) --
	( 24.93, 48.46) --
	cycle;
\definecolor{drawColor}{RGB}{217,95,2}

\path[draw=drawColor,line width= 0.4pt,line join=round,line cap=round] ( 27.57, 45.73) circle (  1.96);
\definecolor{drawColor}{RGB}{117,112,179}

\path[draw=drawColor,line width= 0.4pt,line join=round,line cap=round] ( 27.57, 46.87) --
	( 30.21, 42.29) --
	( 24.93, 42.29) --
	cycle;
\definecolor{drawColor}{RGB}{231,41,138}

\path[draw=drawColor,line width= 0.4pt,line join=round,line cap=round] ( 24.80, 41.54) -- ( 30.35, 41.54);

\path[draw=drawColor,line width= 0.4pt,line join=round,line cap=round] ( 27.57, 38.77) -- ( 27.57, 44.32);
\definecolor{drawColor}{RGB}{102,166,30}

\path[draw=drawColor,line width= 0.4pt,line join=round,line cap=round] ( 25.61, 36.53) -- ( 29.53, 40.45);

\path[draw=drawColor,line width= 0.4pt,line join=round,line cap=round] ( 25.61, 40.45) -- ( 29.53, 36.53);
\definecolor{drawColor}{RGB}{166,118,29}

\path[draw=drawColor,line width= 0.4pt,line join=round,line cap=round] ( 27.57, 31.99) --
	( 30.21, 36.57) --
	( 24.93, 36.57) --
	cycle;
\definecolor{drawColor}{RGB}{217,95,2}

\path[draw=drawColor,line width= 0.4pt,line join=round,line cap=round] ( 37.84, 92.19) circle (  1.96);
\definecolor{drawColor}{RGB}{117,112,179}

\path[draw=drawColor,line width= 0.4pt,line join=round,line cap=round] ( 37.84,109.36) --
	( 40.48,104.78) --
	( 35.19,104.78) --
	cycle;
\definecolor{drawColor}{RGB}{231,41,138}

\path[draw=drawColor,line width= 0.4pt,line join=round,line cap=round] ( 35.06,106.06) -- ( 40.61,106.06);

\path[draw=drawColor,line width= 0.4pt,line join=round,line cap=round] ( 37.84,103.28) -- ( 37.84,108.83);
\definecolor{drawColor}{RGB}{102,166,30}

\path[draw=drawColor,line width= 0.4pt,line join=round,line cap=round] ( 35.87, 84.52) -- ( 39.80, 88.44);

\path[draw=drawColor,line width= 0.4pt,line join=round,line cap=round] ( 35.87, 88.44) -- ( 39.80, 84.52);
\definecolor{drawColor}{RGB}{166,118,29}

\path[draw=drawColor,line width= 0.4pt,line join=round,line cap=round] ( 37.84, 55.72) --
	( 40.48, 60.30) --
	( 35.19, 60.30) --
	cycle;
\definecolor{drawColor}{RGB}{217,95,2}

\path[draw=drawColor,line width= 0.4pt,line join=round,line cap=round] ( 65.82, 99.91) circle (  1.96);
\definecolor{drawColor}{RGB}{117,112,179}

\path[draw=drawColor,line width= 0.4pt,line join=round,line cap=round] ( 65.82,129.47) --
	( 68.46,124.90) --
	( 63.18,124.90) --
	cycle;
\definecolor{drawColor}{RGB}{231,41,138}

\path[draw=drawColor,line width= 0.4pt,line join=round,line cap=round] ( 63.04,146.44) -- ( 68.59,146.44);

\path[draw=drawColor,line width= 0.4pt,line join=round,line cap=round] ( 65.82,143.67) -- ( 65.82,149.22);
\definecolor{drawColor}{RGB}{102,166,30}

\path[draw=drawColor,line width= 0.4pt,line join=round,line cap=round] ( 63.86,141.84) -- ( 67.78,145.77);

\path[draw=drawColor,line width= 0.4pt,line join=round,line cap=round] ( 63.86,145.77) -- ( 67.78,141.84);
\definecolor{drawColor}{RGB}{166,118,29}

\path[draw=drawColor,line width= 0.4pt,line join=round,line cap=round] ( 65.82, 69.19) --
	( 68.46, 73.76) --
	( 63.18, 73.76) --
	cycle;
\definecolor{drawColor}{RGB}{217,95,2}

\path[draw=drawColor,line width= 0.4pt,line join=round,line cap=round] (114.32, 99.66) circle (  1.96);
\definecolor{drawColor}{RGB}{117,112,179}

\path[draw=drawColor,line width= 0.4pt,line join=round,line cap=round] (114.32,129.69) --
	(116.96,125.11) --
	(111.68,125.11) --
	cycle;
\definecolor{drawColor}{RGB}{231,41,138}

\path[draw=drawColor,line width= 0.4pt,line join=round,line cap=round] (111.54,147.76) -- (117.09,147.76);

\path[draw=drawColor,line width= 0.4pt,line join=round,line cap=round] (114.32,144.99) -- (114.32,150.54);
\definecolor{drawColor}{RGB}{102,166,30}

\path[draw=drawColor,line width= 0.4pt,line join=round,line cap=round] (112.36,149.15) -- (116.28,153.08);

\path[draw=drawColor,line width= 0.4pt,line join=round,line cap=round] (112.36,153.08) -- (116.28,149.15);
\definecolor{drawColor}{RGB}{166,118,29}

\path[draw=drawColor,line width= 0.4pt,line join=round,line cap=round] (114.32, 68.75) --
	(116.96, 73.33) --
	(111.68, 73.33) --
	cycle;
\definecolor{drawColor}{RGB}{217,95,2}

\path[draw=drawColor,line width= 0.4pt,line join=round,line cap=round] (182.87, 97.87) circle (  1.96);
\definecolor{drawColor}{RGB}{117,112,179}

\path[draw=drawColor,line width= 0.4pt,line join=round,line cap=round] (182.87,125.65) --
	(185.52,121.07) --
	(180.23,121.07) --
	cycle;
\definecolor{drawColor}{RGB}{231,41,138}

\path[draw=drawColor,line width= 0.4pt,line join=round,line cap=round] (180.10,140.84) -- (185.65,140.84);

\path[draw=drawColor,line width= 0.4pt,line join=round,line cap=round] (182.87,138.06) -- (182.87,143.61);
\definecolor{drawColor}{RGB}{102,166,30}

\path[draw=drawColor,line width= 0.4pt,line join=round,line cap=round] (180.91,141.81) -- (184.84,145.74);

\path[draw=drawColor,line width= 0.4pt,line join=round,line cap=round] (180.91,145.74) -- (184.84,141.81);
\definecolor{drawColor}{RGB}{166,118,29}

\path[draw=drawColor,line width= 0.4pt,line join=round,line cap=round] (182.87, 65.99) --
	(185.52, 70.57) --
	(180.23, 70.57) --
	cycle;

\path[draw=drawColor,line width= 0.4pt,line join=round,line cap=round] ( 27.57, 43.88) --
	( 30.21, 48.46) --
	( 24.93, 48.46) --
	cycle;
\definecolor{drawColor}{RGB}{217,95,2}

\path[draw=drawColor,line width= 0.4pt,line join=round,line cap=round] ( 27.57, 45.11) circle (  1.96);
\definecolor{drawColor}{RGB}{117,112,179}

\path[draw=drawColor,line width= 0.4pt,line join=round,line cap=round] ( 27.57, 44.79) --
	( 30.21, 40.21) --
	( 24.93, 40.21) --
	cycle;
\definecolor{drawColor}{RGB}{231,41,138}

\path[draw=drawColor,line width= 0.4pt,line join=round,line cap=round] ( 24.80, 38.40) -- ( 30.35, 38.40);

\path[draw=drawColor,line width= 0.4pt,line join=round,line cap=round] ( 27.57, 35.62) -- ( 27.57, 41.17);
\definecolor{drawColor}{RGB}{102,166,30}

\path[draw=drawColor,line width= 0.4pt,line join=round,line cap=round] ( 25.61, 33.82) -- ( 29.53, 37.74);

\path[draw=drawColor,line width= 0.4pt,line join=round,line cap=round] ( 25.61, 37.74) -- ( 29.53, 33.82);
\definecolor{drawColor}{RGB}{166,118,29}

\path[draw=drawColor,line width= 0.4pt,line join=round,line cap=round] ( 27.57, 30.50) --
	( 30.21, 35.07) --
	( 24.93, 35.07) --
	cycle;
\definecolor{drawColor}{RGB}{217,95,2}

\path[draw=drawColor,line width= 0.4pt,line join=round,line cap=round] ( 38.29, 97.63) circle (  1.96);
\definecolor{drawColor}{RGB}{117,112,179}

\path[draw=drawColor,line width= 0.4pt,line join=round,line cap=round] ( 38.29,105.86) --
	( 40.93,101.29) --
	( 35.64,101.29) --
	cycle;
\definecolor{drawColor}{RGB}{231,41,138}

\path[draw=drawColor,line width= 0.4pt,line join=round,line cap=round] ( 35.51, 91.09) -- ( 41.06, 91.09);

\path[draw=drawColor,line width= 0.4pt,line join=round,line cap=round] ( 38.29, 88.32) -- ( 38.29, 93.87);
\definecolor{drawColor}{RGB}{102,166,30}

\path[draw=drawColor,line width= 0.4pt,line join=round,line cap=round] ( 36.32, 69.75) -- ( 40.25, 73.67);

\path[draw=drawColor,line width= 0.4pt,line join=round,line cap=round] ( 36.32, 73.67) -- ( 40.25, 69.75);
\definecolor{drawColor}{RGB}{166,118,29}

\path[draw=drawColor,line width= 0.4pt,line join=round,line cap=round] ( 38.29, 47.30) --
	( 40.93, 51.88) --
	( 35.64, 51.88) --
	cycle;
\definecolor{drawColor}{RGB}{217,95,2}

\path[draw=drawColor,line width= 0.4pt,line join=round,line cap=round] ( 67.34,109.13) circle (  1.96);
\definecolor{drawColor}{RGB}{117,112,179}

\path[draw=drawColor,line width= 0.4pt,line join=round,line cap=round] ( 67.34,131.39) --
	( 69.98,126.82) --
	( 64.70,126.82) --
	cycle;
\definecolor{drawColor}{RGB}{231,41,138}

\path[draw=drawColor,line width= 0.4pt,line join=round,line cap=round] ( 64.57,131.73) -- ( 70.12,131.73);

\path[draw=drawColor,line width= 0.4pt,line join=round,line cap=round] ( 67.34,128.95) -- ( 67.34,134.50);
\definecolor{drawColor}{RGB}{102,166,30}

\path[draw=drawColor,line width= 0.4pt,line join=round,line cap=round] ( 65.38,119.69) -- ( 69.30,123.62);

\path[draw=drawColor,line width= 0.4pt,line join=round,line cap=round] ( 65.38,123.62) -- ( 69.30,119.69);
\definecolor{drawColor}{RGB}{166,118,29}

\path[draw=drawColor,line width= 0.4pt,line join=round,line cap=round] ( 67.34, 57.25) --
	( 69.98, 61.83) --
	( 64.70, 61.83) --
	cycle;
\definecolor{drawColor}{RGB}{217,95,2}

\path[draw=drawColor,line width= 0.4pt,line join=round,line cap=round] (117.69,107.95) circle (  1.96);
\definecolor{drawColor}{RGB}{117,112,179}

\path[draw=drawColor,line width= 0.4pt,line join=round,line cap=round] (117.69,129.83) --
	(120.33,125.26) --
	(115.05,125.26) --
	cycle;
\definecolor{drawColor}{RGB}{231,41,138}

\path[draw=drawColor,line width= 0.4pt,line join=round,line cap=round] (114.92,133.07) -- (120.47,133.07);

\path[draw=drawColor,line width= 0.4pt,line join=round,line cap=round] (117.69,130.29) -- (117.69,135.84);
\definecolor{drawColor}{RGB}{102,166,30}

\path[draw=drawColor,line width= 0.4pt,line join=round,line cap=round] (115.73,126.06) -- (119.65,129.98);

\path[draw=drawColor,line width= 0.4pt,line join=round,line cap=round] (115.73,129.98) -- (119.65,126.06);
\definecolor{drawColor}{RGB}{166,118,29}

\path[draw=drawColor,line width= 0.4pt,line join=round,line cap=round] (117.69, 56.10) --
	(120.33, 60.68) --
	(115.05, 60.68) --
	cycle;
\definecolor{drawColor}{RGB}{217,95,2}

\path[draw=drawColor,line width= 0.4pt,line join=round,line cap=round] (188.75,105.70) circle (  1.96);
\definecolor{drawColor}{RGB}{117,112,179}

\path[draw=drawColor,line width= 0.4pt,line join=round,line cap=round] (188.75,125.50) --
	(191.40,120.92) --
	(186.11,120.92) --
	cycle;
\definecolor{drawColor}{RGB}{231,41,138}

\path[draw=drawColor,line width= 0.4pt,line join=round,line cap=round] (185.98,125.38) -- (191.53,125.38);

\path[draw=drawColor,line width= 0.4pt,line join=round,line cap=round] (188.75,122.60) -- (188.75,128.15);
\definecolor{drawColor}{RGB}{102,166,30}

\path[draw=drawColor,line width= 0.4pt,line join=round,line cap=round] (186.79,117.62) -- (190.72,121.54);

\path[draw=drawColor,line width= 0.4pt,line join=round,line cap=round] (186.79,121.54) -- (190.72,117.62);
\definecolor{drawColor}{RGB}{166,118,29}

\path[draw=drawColor,line width= 0.4pt,line join=round,line cap=round] (188.75, 53.89) --
	(191.40, 58.46) --
	(186.11, 58.46) --
	cycle;
\definecolor{drawColor}{RGB}{217,95,2}

\path[draw=drawColor,line width= 0.6pt,line join=round] ( 27.57, 45.73) --
	( 37.84, 92.19) --
	( 65.82, 99.91) --
	(114.32, 99.66) --
	(182.87, 97.87);

\path[draw=drawColor,line width= 0.6pt,dash pattern=on 4pt off 4pt ,line join=round] ( 27.57, 45.11) --
	( 38.29, 97.63) --
	( 67.34,109.13) --
	(117.69,107.95) --
	(188.75,105.70);
\definecolor{drawColor}{RGB}{117,112,179}

\path[draw=drawColor,line width= 0.6pt,line join=round] ( 27.57, 43.82) --
	( 37.84,106.31) --
	( 65.82,126.42) --
	(114.32,126.64) --
	(182.87,122.59);

\path[draw=drawColor,line width= 0.6pt,dash pattern=on 4pt off 4pt ,line join=round] ( 27.57, 41.73) --
	( 38.29,102.81) --
	( 67.34,128.34) --
	(117.69,126.78) --
	(188.75,122.44);
\definecolor{drawColor}{RGB}{231,41,138}

\path[draw=drawColor,line width= 0.6pt,line join=round] ( 27.57, 41.54) --
	( 37.84,106.06) --
	( 65.82,146.44) --
	(114.32,147.76) --
	(182.87,140.84);

\path[draw=drawColor,line width= 0.6pt,dash pattern=on 4pt off 4pt ,line join=round] ( 27.57, 38.40) --
	( 38.29, 91.09) --
	( 67.34,131.73) --
	(117.69,133.07) --
	(188.75,125.38);
\definecolor{drawColor}{RGB}{102,166,30}

\path[draw=drawColor,line width= 0.6pt,line join=round] ( 27.57, 38.49) --
	( 37.84, 86.48) --
	( 65.82,143.81) --
	(114.32,151.11) --
	(182.87,143.77);

\path[draw=drawColor,line width= 0.6pt,dash pattern=on 4pt off 4pt ,line join=round] ( 27.57, 35.78) --
	( 38.29, 71.71) --
	( 67.34,121.66) --
	(117.69,128.02) --
	(188.75,119.58);
\definecolor{drawColor}{RGB}{166,118,29}

\path[draw=drawColor,line width= 0.6pt,line join=round] ( 27.57, 46.93) --
	( 27.57, 35.04) --
	( 37.84, 58.77) --
	( 65.82, 72.24) --
	(114.32, 71.80) --
	(182.87, 69.04);

\path[draw=drawColor,line width= 0.6pt,dash pattern=on 4pt off 4pt ,line join=round] ( 27.57, 46.93) --
	( 27.57, 33.55) --
	( 38.29, 50.36) --
	( 67.34, 60.30) --
	(117.69, 59.16) --
	(188.75, 56.94);
\definecolor{drawColor}{RGB}{0,0,0}

\path[draw=drawColor,line width= 0.6pt,dash pattern=on 1pt off 3pt ,line join=round] ( 26.69, 46.93) -- (196.86, 46.93);
\end{scope}
\begin{scope}
\path[clip] (  0.00,  0.00) rectangle (505.89,162.61);
\definecolor{drawColor}{RGB}{0,0,0}

\path[draw=drawColor,line width= 0.6pt,line join=round] ( 26.69, 31.25) --
	( 26.69,157.11);
\end{scope}
\begin{scope}
\path[clip] (  0.00,  0.00) rectangle (505.89,162.61);
\definecolor{drawColor}{gray}{0.30}

\node[text=drawColor,anchor=base east,inner sep=0pt, outer sep=0pt, scale=  0.88] at ( 21.74, 28.22) {0};

\node[text=drawColor,anchor=base east,inner sep=0pt, outer sep=0pt, scale=  0.88] at ( 21.74, 59.59) {2};

\node[text=drawColor,anchor=base east,inner sep=0pt, outer sep=0pt, scale=  0.88] at ( 21.74, 90.95) {4};

\node[text=drawColor,anchor=base east,inner sep=0pt, outer sep=0pt, scale=  0.88] at ( 21.74,122.31) {6};

\node[text=drawColor,anchor=base east,inner sep=0pt, outer sep=0pt, scale=  0.88] at ( 21.74,153.68) {8};
\end{scope}
\begin{scope}
\path[clip] (  0.00,  0.00) rectangle (505.89,162.61);
\definecolor{drawColor}{gray}{0.20}

\path[draw=drawColor,line width= 0.6pt,line join=round] ( 23.94, 31.25) --
	( 26.69, 31.25);

\path[draw=drawColor,line width= 0.6pt,line join=round] ( 23.94, 62.62) --
	( 26.69, 62.62);

\path[draw=drawColor,line width= 0.6pt,line join=round] ( 23.94, 93.98) --
	( 26.69, 93.98);

\path[draw=drawColor,line width= 0.6pt,line join=round] ( 23.94,125.34) --
	( 26.69,125.34);

\path[draw=drawColor,line width= 0.6pt,line join=round] ( 23.94,156.71) --
	( 26.69,156.71);
\end{scope}
\begin{scope}
\path[clip] (  0.00,  0.00) rectangle (505.89,162.61);
\definecolor{drawColor}{RGB}{0,0,0}

\path[draw=drawColor,line width= 0.6pt,line join=round] ( 26.69, 31.25) --
	(196.86, 31.25);
\end{scope}
\begin{scope}
\path[clip] (  0.00,  0.00) rectangle (505.89,162.61);
\definecolor{drawColor}{gray}{0.20}

\path[draw=drawColor,line width= 0.6pt,line join=round] ( 26.69, 28.50) --
	( 26.69, 31.25);

\path[draw=drawColor,line width= 0.6pt,line join=round] ( 70.53, 28.50) --
	( 70.53, 31.25);

\path[draw=drawColor,line width= 0.6pt,line join=round] (114.37, 28.50) --
	(114.37, 31.25);

\path[draw=drawColor,line width= 0.6pt,line join=round] (158.20, 28.50) --
	(158.20, 31.25);
\end{scope}
\begin{scope}
\path[clip] (  0.00,  0.00) rectangle (505.89,162.61);
\definecolor{drawColor}{gray}{0.30}

\node[text=drawColor,anchor=base,inner sep=0pt, outer sep=0pt, scale=  0.88] at ( 26.69, 20.24) {0};

\node[text=drawColor,anchor=base,inner sep=0pt, outer sep=0pt, scale=  0.88] at ( 70.53, 20.24) {100};

\node[text=drawColor,anchor=base,inner sep=0pt, outer sep=0pt, scale=  0.88] at (114.37, 20.24) {200};

\node[text=drawColor,anchor=base,inner sep=0pt, outer sep=0pt, scale=  0.88] at (158.20, 20.24) {300};
\end{scope}
\begin{scope}
\path[clip] (  0.00,  0.00) rectangle (505.89,162.61);
\definecolor{drawColor}{RGB}{0,0,0}

\node[text=drawColor,anchor=base,inner sep=0pt, outer sep=0pt, scale=  1.10] at (111.78,  7.93) {$\numpickups + \numdropoffs$};
\end{scope}
\begin{scope}
\path[clip] (  0.00,  0.00) rectangle (505.89,162.61);
\definecolor{drawColor}{RGB}{0,0,0}

\node[text=drawColor,rotate= 90.00,anchor=base,inner sep=0pt, outer sep=0pt, scale=  1.00] at ( 12.39, 94.18) {speedup (ell. BCH)};
\end{scope}
\begin{scope}
\path[clip] (202.36,  0.00) rectangle (404.71,162.61);
\definecolor{drawColor}{RGB}{255,255,255}
\definecolor{fillColor}{RGB}{255,255,255}

\path[draw=drawColor,line width= 0.6pt,line join=round,line cap=round,fill=fillColor] (202.36,  0.00) rectangle (404.71,162.61);
\end{scope}
\begin{scope}
\path[clip] (229.05, 31.25) rectangle (399.21,157.11);
\definecolor{fillColor}{RGB}{255,255,255}

\path[fill=fillColor] (229.05, 31.25) rectangle (399.21,157.11);
\definecolor{drawColor}{RGB}{166,118,29}

\path[draw=drawColor,line width= 0.4pt,line join=round,line cap=round] (229.05, 50.47) --
	(231.70, 55.05) --
	(226.41, 55.05) --
	cycle;
\definecolor{drawColor}{RGB}{117,112,179}

\path[draw=drawColor,line width= 0.4pt,line join=round,line cap=round] (229.05, 50.64) --
	(231.70, 46.06) --
	(226.41, 46.06) --
	cycle;
\definecolor{drawColor}{RGB}{231,41,138}

\path[draw=drawColor,line width= 0.4pt,line join=round,line cap=round] (226.28, 47.02) -- (231.83, 47.02);

\path[draw=drawColor,line width= 0.4pt,line join=round,line cap=round] (229.05, 44.25) -- (229.05, 49.80);
\definecolor{drawColor}{RGB}{102,166,30}

\path[draw=drawColor,line width= 0.4pt,line join=round,line cap=round] (227.09, 39.94) -- (231.02, 43.87);

\path[draw=drawColor,line width= 0.4pt,line join=round,line cap=round] (227.09, 43.87) -- (231.02, 39.94);
\definecolor{drawColor}{RGB}{166,118,29}

\path[draw=drawColor,line width= 0.4pt,line join=round,line cap=round] (229.05, 43.40) --
	(231.70, 47.97) --
	(226.41, 47.97) --
	cycle;
\definecolor{drawColor}{RGB}{117,112,179}

\path[draw=drawColor,line width= 0.4pt,line join=round,line cap=round] (229.79,110.81) --
	(232.43,106.24) --
	(227.15,106.24) --
	cycle;
\definecolor{drawColor}{RGB}{231,41,138}

\path[draw=drawColor,line width= 0.4pt,line join=round,line cap=round] (227.02,113.55) -- (232.57,113.55);

\path[draw=drawColor,line width= 0.4pt,line join=round,line cap=round] (229.79,110.77) -- (229.79,116.32);
\definecolor{drawColor}{RGB}{102,166,30}

\path[draw=drawColor,line width= 0.4pt,line join=round,line cap=round] (227.83, 94.80) -- (231.75, 98.72);

\path[draw=drawColor,line width= 0.4pt,line join=round,line cap=round] (227.83, 98.72) -- (231.75, 94.80);
\definecolor{drawColor}{RGB}{166,118,29}

\path[draw=drawColor,line width= 0.4pt,line join=round,line cap=round] (229.79, 75.83) --
	(232.43, 80.41) --
	(227.15, 80.41) --
	cycle;
\definecolor{drawColor}{RGB}{117,112,179}

\path[draw=drawColor,line width= 0.4pt,line join=round,line cap=round] (238.40,126.83) --
	(241.04,122.25) --
	(235.75,122.25) --
	cycle;
\definecolor{drawColor}{RGB}{231,41,138}

\path[draw=drawColor,line width= 0.4pt,line join=round,line cap=round] (235.62,142.97) -- (241.17,142.97);

\path[draw=drawColor,line width= 0.4pt,line join=round,line cap=round] (238.40,140.19) -- (238.40,145.74);
\definecolor{drawColor}{RGB}{102,166,30}

\path[draw=drawColor,line width= 0.4pt,line join=round,line cap=round] (236.43,133.00) -- (240.36,136.92);

\path[draw=drawColor,line width= 0.4pt,line join=round,line cap=round] (236.43,136.92) -- (240.36,133.00);
\definecolor{drawColor}{RGB}{166,118,29}

\path[draw=drawColor,line width= 0.4pt,line join=round,line cap=round] (238.40, 85.28) --
	(241.04, 89.85) --
	(235.75, 89.85) --
	cycle;
\definecolor{drawColor}{RGB}{117,112,179}

\path[draw=drawColor,line width= 0.4pt,line join=round,line cap=round] (276.27,134.66) --
	(278.91,130.08) --
	(273.62,130.08) --
	cycle;
\definecolor{drawColor}{RGB}{231,41,138}

\path[draw=drawColor,line width= 0.4pt,line join=round,line cap=round] (273.49,151.11) -- (279.04,151.11);

\path[draw=drawColor,line width= 0.4pt,line join=round,line cap=round] (276.27,148.34) -- (276.27,153.89);
\definecolor{drawColor}{RGB}{102,166,30}

\path[draw=drawColor,line width= 0.4pt,line join=round,line cap=round] (274.30,135.85) -- (278.23,139.78);

\path[draw=drawColor,line width= 0.4pt,line join=round,line cap=round] (274.30,139.78) -- (278.23,135.85);
\definecolor{drawColor}{RGB}{166,118,29}

\path[draw=drawColor,line width= 0.4pt,line join=round,line cap=round] (276.27, 78.17) --
	(278.91, 82.74) --
	(273.62, 82.74) --
	cycle;
\definecolor{drawColor}{RGB}{117,112,179}

\path[draw=drawColor,line width= 0.4pt,line join=round,line cap=round] (379.80,137.28) --
	(382.44,132.70) --
	(377.16,132.70) --
	cycle;
\definecolor{drawColor}{RGB}{231,41,138}

\path[draw=drawColor,line width= 0.4pt,line join=round,line cap=round] (377.03,149.03) -- (382.58,149.03);

\path[draw=drawColor,line width= 0.4pt,line join=round,line cap=round] (379.80,146.25) -- (379.80,151.80);
\definecolor{drawColor}{RGB}{102,166,30}

\path[draw=drawColor,line width= 0.4pt,line join=round,line cap=round] (377.84,125.75) -- (381.76,129.67);

\path[draw=drawColor,line width= 0.4pt,line join=round,line cap=round] (377.84,129.67) -- (381.76,125.75);
\definecolor{drawColor}{RGB}{166,118,29}

\path[draw=drawColor,line width= 0.4pt,line join=round,line cap=round] (379.80, 69.33) --
	(382.44, 73.91) --
	(377.16, 73.91) --
	cycle;

\path[draw=drawColor,line width= 0.4pt,line join=round,line cap=round] (229.05, 50.47) --
	(231.70, 55.05) --
	(226.41, 55.05) --
	cycle;
\definecolor{drawColor}{RGB}{117,112,179}

\path[draw=drawColor,line width= 0.4pt,line join=round,line cap=round] (229.05, 52.61) --
	(231.70, 48.04) --
	(226.41, 48.04) --
	cycle;
\definecolor{drawColor}{RGB}{231,41,138}

\path[draw=drawColor,line width= 0.4pt,line join=round,line cap=round] (226.28, 47.80) -- (231.83, 47.80);

\path[draw=drawColor,line width= 0.4pt,line join=round,line cap=round] (229.05, 45.03) -- (229.05, 50.58);
\definecolor{drawColor}{RGB}{102,166,30}

\path[draw=drawColor,line width= 0.4pt,line join=round,line cap=round] (227.09, 42.56) -- (231.02, 46.48);

\path[draw=drawColor,line width= 0.4pt,line join=round,line cap=round] (227.09, 46.48) -- (231.02, 42.56);
\definecolor{drawColor}{RGB}{166,118,29}

\path[draw=drawColor,line width= 0.4pt,line join=round,line cap=round] (229.05, 44.57) --
	(231.70, 49.15) --
	(226.41, 49.15) --
	cycle;
\definecolor{drawColor}{RGB}{117,112,179}

\path[draw=drawColor,line width= 0.4pt,line join=round,line cap=round] (229.85,107.36) --
	(232.49,102.79) --
	(227.20,102.79) --
	cycle;
\definecolor{drawColor}{RGB}{231,41,138}

\path[draw=drawColor,line width= 0.4pt,line join=round,line cap=round] (227.07,109.02) -- (232.62,109.02);

\path[draw=drawColor,line width= 0.4pt,line join=round,line cap=round] (229.85,106.24) -- (229.85,111.79);
\definecolor{drawColor}{RGB}{102,166,30}

\path[draw=drawColor,line width= 0.4pt,line join=round,line cap=round] (227.89, 91.45) -- (231.81, 95.37);

\path[draw=drawColor,line width= 0.4pt,line join=round,line cap=round] (227.89, 95.37) -- (231.81, 91.45);
\definecolor{drawColor}{RGB}{166,118,29}

\path[draw=drawColor,line width= 0.4pt,line join=round,line cap=round] (229.85, 73.02) --
	(232.49, 77.60) --
	(227.20, 77.60) --
	cycle;
\definecolor{drawColor}{RGB}{117,112,179}

\path[draw=drawColor,line width= 0.4pt,line join=round,line cap=round] (239.11,121.73) --
	(241.75,117.15) --
	(236.46,117.15) --
	cycle;
\definecolor{drawColor}{RGB}{231,41,138}

\path[draw=drawColor,line width= 0.4pt,line join=round,line cap=round] (236.33,135.52) -- (241.88,135.52);

\path[draw=drawColor,line width= 0.4pt,line join=round,line cap=round] (239.11,132.74) -- (239.11,138.29);
\definecolor{drawColor}{RGB}{102,166,30}

\path[draw=drawColor,line width= 0.4pt,line join=round,line cap=round] (237.14,125.51) -- (241.07,129.44);

\path[draw=drawColor,line width= 0.4pt,line join=round,line cap=round] (237.14,129.44) -- (241.07,125.51);
\definecolor{drawColor}{RGB}{166,118,29}

\path[draw=drawColor,line width= 0.4pt,line join=round,line cap=round] (239.11, 80.45) --
	(241.75, 85.02) --
	(236.46, 85.02) --
	cycle;
\definecolor{drawColor}{RGB}{117,112,179}

\path[draw=drawColor,line width= 0.4pt,line join=round,line cap=round] (279.91,129.02) --
	(282.55,124.44) --
	(277.27,124.44) --
	cycle;
\definecolor{drawColor}{RGB}{231,41,138}

\path[draw=drawColor,line width= 0.4pt,line join=round,line cap=round] (277.14,142.73) -- (282.69,142.73);

\path[draw=drawColor,line width= 0.4pt,line join=round,line cap=round] (279.91,139.96) -- (279.91,145.51);
\definecolor{drawColor}{RGB}{102,166,30}

\path[draw=drawColor,line width= 0.4pt,line join=round,line cap=round] (277.95,127.03) -- (281.87,130.96);

\path[draw=drawColor,line width= 0.4pt,line join=round,line cap=round] (277.95,130.96) -- (281.87,127.03);
\definecolor{drawColor}{RGB}{166,118,29}

\path[draw=drawColor,line width= 0.4pt,line join=round,line cap=round] (279.91, 73.36) --
	(282.55, 77.94) --
	(277.27, 77.94) --
	cycle;
\definecolor{drawColor}{RGB}{117,112,179}

\path[draw=drawColor,line width= 0.4pt,line join=round,line cap=round] (391.11,133.01) --
	(393.75,128.43) --
	(388.47,128.43) --
	cycle;
\definecolor{drawColor}{RGB}{231,41,138}

\path[draw=drawColor,line width= 0.4pt,line join=round,line cap=round] (388.33,142.48) -- (393.88,142.48);

\path[draw=drawColor,line width= 0.4pt,line join=round,line cap=round] (391.11,139.71) -- (391.11,145.26);
\definecolor{drawColor}{RGB}{102,166,30}

\path[draw=drawColor,line width= 0.4pt,line join=round,line cap=round] (389.15,118.25) -- (393.07,122.17);

\path[draw=drawColor,line width= 0.4pt,line join=round,line cap=round] (389.15,122.17) -- (393.07,118.25);
\definecolor{drawColor}{RGB}{166,118,29}

\path[draw=drawColor,line width= 0.4pt,line join=round,line cap=round] (391.11, 66.26) --
	(393.75, 70.84) --
	(388.47, 70.84) --
	cycle;
\definecolor{drawColor}{RGB}{117,112,179}

\path[draw=drawColor,line width= 0.6pt,line join=round] (229.05, 47.58) --
	(229.79,107.76) --
	(238.40,123.78) --
	(276.27,131.61) --
	(379.80,134.23);

\path[draw=drawColor,line width= 0.6pt,dash pattern=on 4pt off 4pt ,line join=round] (229.05, 49.56) --
	(229.85,104.31) --
	(239.11,118.68) --
	(279.91,125.96) --
	(391.11,129.96);
\definecolor{drawColor}{RGB}{231,41,138}

\path[draw=drawColor,line width= 0.6pt,line join=round] (229.05, 47.02) --
	(229.79,113.55) --
	(238.40,142.97) --
	(276.27,151.11) --
	(379.80,149.03);

\path[draw=drawColor,line width= 0.6pt,dash pattern=on 4pt off 4pt ,line join=round] (229.05, 47.80) --
	(229.85,109.02) --
	(239.11,135.52) --
	(279.91,142.73) --
	(391.11,142.48);
\definecolor{drawColor}{RGB}{102,166,30}

\path[draw=drawColor,line width= 0.6pt,line join=round] (229.05, 41.91) --
	(229.79, 96.76) --
	(238.40,134.96) --
	(276.27,137.82) --
	(379.80,127.71);

\path[draw=drawColor,line width= 0.6pt,dash pattern=on 4pt off 4pt ,line join=round] (229.05, 44.52) --
	(229.85, 93.41) --
	(239.11,127.47) --
	(279.91,129.00) --
	(391.11,120.21);
\definecolor{drawColor}{RGB}{166,118,29}

\path[draw=drawColor,line width= 0.6pt,line join=round] (229.05, 53.52) --
	(229.05, 46.45) --
	(229.79, 78.88) --
	(238.40, 88.33) --
	(276.27, 81.22) --
	(379.80, 72.39);

\path[draw=drawColor,line width= 0.6pt,dash pattern=on 4pt off 4pt ,line join=round] (229.05, 53.52) --
	(229.05, 47.62) --
	(229.85, 76.07) --
	(239.11, 83.50) --
	(279.91, 76.41) --
	(391.11, 69.32);
\definecolor{drawColor}{RGB}{0,0,0}

\path[draw=drawColor,line width= 0.6pt,dash pattern=on 1pt off 3pt ,line join=round] (229.05, 53.52) -- (399.21, 53.52);
\end{scope}
\begin{scope}
\path[clip] (  0.00,  0.00) rectangle (505.89,162.61);
\definecolor{drawColor}{RGB}{0,0,0}

\path[draw=drawColor,line width= 0.6pt,line join=round] (229.05, 31.25) --
	(229.05,157.11);
\end{scope}
\begin{scope}
\path[clip] (  0.00,  0.00) rectangle (505.89,162.61);
\definecolor{drawColor}{gray}{0.30}

\node[text=drawColor,anchor=base east,inner sep=0pt, outer sep=0pt, scale=  0.88] at (224.10, 28.22) {0};

\node[text=drawColor,anchor=base east,inner sep=0pt, outer sep=0pt, scale=  0.88] at (224.10, 50.49) {1};

\node[text=drawColor,anchor=base east,inner sep=0pt, outer sep=0pt, scale=  0.88] at (224.10, 72.76) {2};

\node[text=drawColor,anchor=base east,inner sep=0pt, outer sep=0pt, scale=  0.88] at (224.10, 95.02) {3};

\node[text=drawColor,anchor=base east,inner sep=0pt, outer sep=0pt, scale=  0.88] at (224.10,117.29) {4};

\node[text=drawColor,anchor=base east,inner sep=0pt, outer sep=0pt, scale=  0.88] at (224.10,139.56) {5};
\end{scope}
\begin{scope}
\path[clip] (  0.00,  0.00) rectangle (505.89,162.61);
\definecolor{drawColor}{gray}{0.20}

\path[draw=drawColor,line width= 0.6pt,line join=round] (226.30, 31.25) --
	(229.05, 31.25);

\path[draw=drawColor,line width= 0.6pt,line join=round] (226.30, 53.52) --
	(229.05, 53.52);

\path[draw=drawColor,line width= 0.6pt,line join=round] (226.30, 75.79) --
	(229.05, 75.79);

\path[draw=drawColor,line width= 0.6pt,line join=round] (226.30, 98.05) --
	(229.05, 98.05);

\path[draw=drawColor,line width= 0.6pt,line join=round] (226.30,120.32) --
	(229.05,120.32);

\path[draw=drawColor,line width= 0.6pt,line join=round] (226.30,142.59) --
	(229.05,142.59);
\end{scope}
\begin{scope}
\path[clip] (  0.00,  0.00) rectangle (505.89,162.61);
\definecolor{drawColor}{RGB}{0,0,0}

\path[draw=drawColor,line width= 0.6pt,line join=round] (229.05, 31.25) --
	(399.21, 31.25);
\end{scope}
\begin{scope}
\path[clip] (  0.00,  0.00) rectangle (505.89,162.61);
\definecolor{drawColor}{gray}{0.20}

\path[draw=drawColor,line width= 0.6pt,line join=round] (229.05, 28.50) --
	(229.05, 31.25);

\path[draw=drawColor,line width= 0.6pt,line join=round] (271.65, 28.50) --
	(271.65, 31.25);

\path[draw=drawColor,line width= 0.6pt,line join=round] (314.24, 28.50) --
	(314.24, 31.25);

\path[draw=drawColor,line width= 0.6pt,line join=round] (356.84, 28.50) --
	(356.84, 31.25);
\end{scope}
\begin{scope}
\path[clip] (  0.00,  0.00) rectangle (505.89,162.61);
\definecolor{drawColor}{gray}{0.30}

\node[text=drawColor,anchor=base,inner sep=0pt, outer sep=0pt, scale=  0.88] at (229.05, 20.24) {0};

\node[text=drawColor,anchor=base,inner sep=0pt, outer sep=0pt, scale=  0.88] at (271.65, 20.24) {10000};

\node[text=drawColor,anchor=base,inner sep=0pt, outer sep=0pt, scale=  0.88] at (314.24, 20.24) {20000};

\node[text=drawColor,anchor=base,inner sep=0pt, outer sep=0pt, scale=  0.88] at (356.84, 20.24) {30000};
\end{scope}
\begin{scope}
\path[clip] (  0.00,  0.00) rectangle (505.89,162.61);
\definecolor{drawColor}{RGB}{0,0,0}

\node[text=drawColor,anchor=base,inner sep=0pt, outer sep=0pt, scale=  1.10] at (314.13,  7.93) {$\numpickups \cdot \numdropoffs$};
\end{scope}
\begin{scope}
\path[clip] (  0.00,  0.00) rectangle (505.89,162.61);
\definecolor{drawColor}{RGB}{0,0,0}

\node[text=drawColor,rotate= 90.00,anchor=base,inner sep=0pt, outer sep=0pt, scale=  1.00] at (214.74, 94.18) {speedup (PD searches)};
\end{scope}
\begin{scope}
\path[clip] (  0.00,  0.00) rectangle (505.89,162.61);
\definecolor{fillColor}{RGB}{255,255,255}

\path[fill=fillColor] (417.08, 65.46) rectangle (466.42,148.63);
\end{scope}
\begin{scope}
\path[clip] (  0.00,  0.00) rectangle (505.89,162.61);
\definecolor{drawColor}{RGB}{0,0,0}

\node[text=drawColor,anchor=base west,inner sep=0pt, outer sep=0pt, scale=  0.80] at (417.08,141.86) {$k$};
\end{scope}
\begin{scope}
\path[clip] (  0.00,  0.00) rectangle (505.89,162.61);
\definecolor{drawColor}{RGB}{217,95,2}

\path[draw=drawColor,line width= 0.4pt,line join=round,line cap=round] (424.19,129.48) circle (  1.96);
\end{scope}
\begin{scope}
\path[clip] (  0.00,  0.00) rectangle (505.89,162.61);
\definecolor{drawColor}{RGB}{217,95,2}

\path[draw=drawColor,line width= 0.6pt,line join=round] (418.50,129.48) -- (429.88,129.48);
\end{scope}
\begin{scope}
\path[clip] (  0.00,  0.00) rectangle (505.89,162.61);
\definecolor{drawColor}{RGB}{117,112,179}

\path[draw=drawColor,line width= 0.4pt,line join=round,line cap=round] (424.19,118.31) --
	(426.83,113.73) --
	(421.55,113.73) --
	cycle;
\end{scope}
\begin{scope}
\path[clip] (  0.00,  0.00) rectangle (505.89,162.61);
\definecolor{drawColor}{RGB}{117,112,179}

\path[draw=drawColor,line width= 0.6pt,line join=round] (418.50,115.26) -- (429.88,115.26);
\end{scope}
\begin{scope}
\path[clip] (  0.00,  0.00) rectangle (505.89,162.61);
\definecolor{drawColor}{RGB}{231,41,138}

\path[draw=drawColor,line width= 0.4pt,line join=round,line cap=round] (421.42,101.03) -- (426.97,101.03);

\path[draw=drawColor,line width= 0.4pt,line join=round,line cap=round] (424.19, 98.26) -- (424.19,103.80);
\end{scope}
\begin{scope}
\path[clip] (  0.00,  0.00) rectangle (505.89,162.61);
\definecolor{drawColor}{RGB}{231,41,138}

\path[draw=drawColor,line width= 0.6pt,line join=round] (418.50,101.03) -- (429.88,101.03);
\end{scope}
\begin{scope}
\path[clip] (  0.00,  0.00) rectangle (505.89,162.61);
\definecolor{drawColor}{RGB}{102,166,30}

\path[draw=drawColor,line width= 0.4pt,line join=round,line cap=round] (422.23, 84.84) -- (426.15, 88.77);

\path[draw=drawColor,line width= 0.4pt,line join=round,line cap=round] (422.23, 88.77) -- (426.15, 84.84);
\end{scope}
\begin{scope}
\path[clip] (  0.00,  0.00) rectangle (505.89,162.61);
\definecolor{drawColor}{RGB}{102,166,30}

\path[draw=drawColor,line width= 0.6pt,line join=round] (418.50, 86.80) -- (429.88, 86.80);
\end{scope}
\begin{scope}
\path[clip] (  0.00,  0.00) rectangle (505.89,162.61);
\definecolor{drawColor}{RGB}{166,118,29}

\path[draw=drawColor,line width= 0.4pt,line join=round,line cap=round] (424.19, 69.53) --
	(426.83, 74.10) --
	(421.55, 74.10) --
	cycle;
\end{scope}
\begin{scope}
\path[clip] (  0.00,  0.00) rectangle (505.89,162.61);
\definecolor{drawColor}{RGB}{166,118,29}

\path[draw=drawColor,line width= 0.6pt,line join=round] (418.50, 72.58) -- (429.88, 72.58);
\end{scope}
\begin{scope}
\path[clip] (  0.00,  0.00) rectangle (505.89,162.61);
\definecolor{drawColor}{RGB}{0,0,0}

\node[text=drawColor,anchor=base west,inner sep=0pt, outer sep=0pt, scale=  0.80] at (435.31,126.73) {8};
\end{scope}
\begin{scope}
\path[clip] (  0.00,  0.00) rectangle (505.89,162.61);
\definecolor{drawColor}{RGB}{0,0,0}

\node[text=drawColor,anchor=base west,inner sep=0pt, outer sep=0pt, scale=  0.80] at (435.31,112.50) {16};
\end{scope}
\begin{scope}
\path[clip] (  0.00,  0.00) rectangle (505.89,162.61);
\definecolor{drawColor}{RGB}{0,0,0}

\node[text=drawColor,anchor=base west,inner sep=0pt, outer sep=0pt, scale=  0.80] at (435.31, 98.28) {32};
\end{scope}
\begin{scope}
\path[clip] (  0.00,  0.00) rectangle (505.89,162.61);
\definecolor{drawColor}{RGB}{0,0,0}

\node[text=drawColor,anchor=base west,inner sep=0pt, outer sep=0pt, scale=  0.80] at (435.31, 84.05) {64};
\end{scope}
\begin{scope}
\path[clip] (  0.00,  0.00) rectangle (505.89,162.61);
\definecolor{drawColor}{RGB}{0,0,0}

\node[text=drawColor,anchor=base west,inner sep=0pt, outer sep=0pt, scale=  0.80] at (435.31, 69.82) {no SIMD};
\end{scope}
\begin{scope}
\path[clip] (  0.00,  0.00) rectangle (505.89,162.61);
\definecolor{fillColor}{RGB}{255,255,255}

\path[fill=fillColor] (417.08, 13.98) rectangle (493.52, 54.46);
\end{scope}
\begin{scope}
\path[clip] (  0.00,  0.00) rectangle (505.89,162.61);
\definecolor{drawColor}{RGB}{0,0,0}

\node[text=drawColor,anchor=base west,inner sep=0pt, outer sep=0pt, scale=  0.80] at (417.08, 47.69) {Inst.};
\end{scope}
\begin{scope}
\path[clip] (  0.00,  0.00) rectangle (505.89,162.61);
\definecolor{drawColor}{RGB}{0,0,0}

\path[draw=drawColor,line width= 0.6pt,line join=round] (418.50, 35.32) -- (429.88, 35.32);
\end{scope}
\begin{scope}
\path[clip] (  0.00,  0.00) rectangle (505.89,162.61);
\definecolor{drawColor}{RGB}{0,0,0}

\path[draw=drawColor,line width= 0.6pt,dash pattern=on 4pt off 4pt ,line join=round] (418.50, 21.09) -- (429.88, 21.09);
\end{scope}
\begin{scope}
\path[clip] (  0.00,  0.00) rectangle (505.89,162.61);
\definecolor{drawColor}{RGB}{0,0,0}

\node[text=drawColor,anchor=base west,inner sep=0pt, outer sep=0pt, scale=  0.80] at (435.31, 32.57) {\ShortBerlinOne};
\end{scope}
\begin{scope}
\path[clip] (  0.00,  0.00) rectangle (505.89,162.61);
\definecolor{drawColor}{RGB}{0,0,0}

\node[text=drawColor,anchor=base west,inner sep=0pt, outer sep=0pt, scale=  0.80] at (435.31, 18.34) {\ShortBerlinTen};
\end{scope}
\end{tikzpicture}

%% file: elliptic_sorted_buckets_speedup_eval.tex
\begin{tikzpicture}[x=1pt,y=1pt]
\definecolor{fillColor}{RGB}{255,255,255}
\path[use as bounding box,fill=fillColor,fill opacity=0.00] (0,0) rectangle (234.88,162.61);
\begin{scope}
\path[clip] (  0.00,  0.00) rectangle (234.88,162.61);
\definecolor{drawColor}{RGB}{255,255,255}
\definecolor{fillColor}{RGB}{255,255,255}

\path[draw=drawColor,line width= 0.6pt,line join=round,line cap=round,fill=fillColor] (  0.00,  0.00) rectangle (234.88,162.61);
\end{scope}
\begin{scope}
\path[clip] ( 34.16, 30.69) rectangle (163.89,157.11);
\definecolor{fillColor}{RGB}{255,255,255}

\path[fill=fillColor] ( 34.16, 30.69) rectangle (163.89,157.11);
\definecolor{drawColor}{RGB}{117,112,179}

\path[draw=drawColor,line width= 0.4pt,line join=round,line cap=round] ( 32.86,101.41) rectangle ( 36.79,105.33);

\path[draw=drawColor,line width= 0.4pt,line join=round,line cap=round] ( 40.69, 98.76) rectangle ( 44.61,102.68);

\path[draw=drawColor,line width= 0.4pt,line join=round,line cap=round] ( 62.02, 98.00) rectangle ( 65.95,101.92);

\path[draw=drawColor,line width= 0.4pt,line join=round,line cap=round] ( 99.00, 97.27) rectangle (102.92,101.20);

\path[draw=drawColor,line width= 0.4pt,line join=round,line cap=round] (151.27, 97.44) rectangle (155.19,101.37);

\path[draw=drawColor,line width= 0.4pt,line join=round,line cap=round] ( 32.86,133.31) rectangle ( 36.79,137.23);

\path[draw=drawColor,line width= 0.4pt,line join=round,line cap=round] ( 41.03,137.43) rectangle ( 44.96,141.35);

\path[draw=drawColor,line width= 0.4pt,line join=round,line cap=round] ( 63.18,143.89) rectangle ( 67.11,147.81);

\path[draw=drawColor,line width= 0.4pt,line join=round,line cap=round] (101.57,148.11) rectangle (105.50,152.03);

\path[draw=drawColor,line width= 0.4pt,line join=round,line cap=round] (155.75,149.13) rectangle (159.67,153.05);

\path[draw=drawColor,line width= 0.6pt,line join=round] ( 34.82,103.37) --
	( 42.65,100.72) --
	( 63.99, 99.96) --
	(100.96, 99.24) --
	(153.23, 99.41);

\path[draw=drawColor,line width= 0.6pt,dash pattern=on 2pt off 2pt ,line join=round] ( 34.83,135.27) --
	( 42.99,139.39) --
	( 65.14,145.85) --
	(103.53,150.07) --
	(157.71,151.09);
\definecolor{drawColor}{RGB}{0,0,0}

\path[draw=drawColor,line width= 0.6pt,dash pattern=on 1pt off 3pt ,line join=round] ( 34.16, 92.53) -- (163.89, 92.53);
\end{scope}
\begin{scope}
\path[clip] (  0.00,  0.00) rectangle (234.88,162.61);
\definecolor{drawColor}{RGB}{0,0,0}

\path[draw=drawColor,line width= 0.6pt,line join=round] ( 34.16, 30.69) --
	( 34.16,157.11);
\end{scope}
\begin{scope}
\path[clip] (  0.00,  0.00) rectangle (234.88,162.61);
\definecolor{drawColor}{gray}{0.30}

\node[text=drawColor,anchor=base east,inner sep=0pt, outer sep=0pt, scale=  0.88] at ( 29.21, 27.66) {0.0};

\node[text=drawColor,anchor=base east,inner sep=0pt, outer sep=0pt, scale=  0.88] at ( 29.21, 58.58) {0.5};

\node[text=drawColor,anchor=base east,inner sep=0pt, outer sep=0pt, scale=  0.88] at ( 29.21, 89.50) {1.0};

\node[text=drawColor,anchor=base east,inner sep=0pt, outer sep=0pt, scale=  0.88] at ( 29.21,120.42) {1.5};

\node[text=drawColor,anchor=base east,inner sep=0pt, outer sep=0pt, scale=  0.88] at ( 29.21,151.35) {2.0};
\end{scope}
\begin{scope}
\path[clip] (  0.00,  0.00) rectangle (234.88,162.61);
\definecolor{drawColor}{gray}{0.20}

\path[draw=drawColor,line width= 0.6pt,line join=round] ( 31.41, 30.69) --
	( 34.16, 30.69);

\path[draw=drawColor,line width= 0.6pt,line join=round] ( 31.41, 61.61) --
	( 34.16, 61.61);

\path[draw=drawColor,line width= 0.6pt,line join=round] ( 31.41, 92.53) --
	( 34.16, 92.53);

\path[draw=drawColor,line width= 0.6pt,line join=round] ( 31.41,123.45) --
	( 34.16,123.45);

\path[draw=drawColor,line width= 0.6pt,line join=round] ( 31.41,154.38) --
	( 34.16,154.38);
\end{scope}
\begin{scope}
\path[clip] (  0.00,  0.00) rectangle (234.88,162.61);
\definecolor{drawColor}{RGB}{0,0,0}

\path[draw=drawColor,line width= 0.6pt,line join=round] ( 34.16, 30.69) --
	(163.89, 30.69);
\end{scope}
\begin{scope}
\path[clip] (  0.00,  0.00) rectangle (234.88,162.61);
\definecolor{drawColor}{gray}{0.20}

\path[draw=drawColor,line width= 0.6pt,line join=round] ( 34.16, 27.94) --
	( 34.16, 30.69);

\path[draw=drawColor,line width= 0.6pt,line join=round] ( 67.58, 27.94) --
	( 67.58, 30.69);

\path[draw=drawColor,line width= 0.6pt,line join=round] (101.00, 27.94) --
	(101.00, 30.69);

\path[draw=drawColor,line width= 0.6pt,line join=round] (134.42, 27.94) --
	(134.42, 30.69);
\end{scope}
\begin{scope}
\path[clip] (  0.00,  0.00) rectangle (234.88,162.61);
\definecolor{drawColor}{gray}{0.30}

\node[text=drawColor,anchor=base,inner sep=0pt, outer sep=0pt, scale=  0.88] at ( 34.16, 19.68) {0};

\node[text=drawColor,anchor=base,inner sep=0pt, outer sep=0pt, scale=  0.88] at ( 67.58, 19.68) {100};

\node[text=drawColor,anchor=base,inner sep=0pt, outer sep=0pt, scale=  0.88] at (101.00, 19.68) {200};

\node[text=drawColor,anchor=base,inner sep=0pt, outer sep=0pt, scale=  0.88] at (134.42, 19.68) {300};
\end{scope}
\begin{scope}
\path[clip] (  0.00,  0.00) rectangle (234.88,162.61);
\definecolor{drawColor}{RGB}{0,0,0}

\node[text=drawColor,anchor=base,inner sep=0pt, outer sep=0pt, scale=  1.10] at ( 99.02,  7.64) {$\numpickups + \numdropoffs$};
\end{scope}
\begin{scope}
\path[clip] (  0.00,  0.00) rectangle (234.88,162.61);
\definecolor{drawColor}{RGB}{0,0,0}

\node[text=drawColor,rotate= 90.00,anchor=base,inner sep=0pt, outer sep=0pt, scale=  1.10] at ( 13.08, 93.90) {speedup (ell. BCH)};
\end{scope}
\begin{scope}
\path[clip] (  0.00,  0.00) rectangle (234.88,162.61);
\definecolor{fillColor}{RGB}{255,255,255}

\path[fill=fillColor] (174.89,106.62) rectangle (229.38,147.29);
\end{scope}
\begin{scope}
\path[clip] (  0.00,  0.00) rectangle (234.88,162.61);
\definecolor{drawColor}{RGB}{0,0,0}

\node[text=drawColor,anchor=base west,inner sep=0pt, outer sep=0pt, scale=  1.10] at (180.39,133.15) {Config.};
\end{scope}
\begin{scope}
\path[clip] (  0.00,  0.00) rectangle (234.88,162.61);
\definecolor{drawColor}{RGB}{117,112,179}

\path[draw=drawColor,line width= 0.4pt,line join=round,line cap=round] (185.65,117.39) rectangle (189.58,121.31);
\end{scope}
\begin{scope}
\path[clip] (  0.00,  0.00) rectangle (234.88,162.61);
\definecolor{drawColor}{RGB}{117,112,179}

\path[draw=drawColor,line width= 0.6pt,line join=round] (181.83,119.35) -- (193.40,119.35);
\end{scope}
\begin{scope}
\path[clip] (  0.00,  0.00) rectangle (234.88,162.61);
\definecolor{drawColor}{RGB}{0,0,0}

\node[text=drawColor,anchor=base west,inner sep=0pt, outer sep=0pt, scale=  0.88] at (200.34,116.32) {sorted};
\end{scope}
\begin{scope}
\path[clip] (  0.00,  0.00) rectangle (234.88,162.61);
\definecolor{fillColor}{RGB}{255,255,255}

\path[fill=fillColor] (174.89, 40.50) rectangle (228.94, 95.62);
\end{scope}
\begin{scope}
\path[clip] (  0.00,  0.00) rectangle (234.88,162.61);
\definecolor{drawColor}{RGB}{0,0,0}

\node[text=drawColor,anchor=base west,inner sep=0pt, outer sep=0pt, scale=  1.10] at (180.39, 81.48) {Inst.};
\end{scope}
\begin{scope}
\path[clip] (  0.00,  0.00) rectangle (234.88,162.61);
\definecolor{drawColor}{RGB}{0,0,0}

\path[draw=drawColor,line width= 0.6pt,line join=round] (181.83, 67.68) -- (193.40, 67.68);
\end{scope}
\begin{scope}
\path[clip] (  0.00,  0.00) rectangle (234.88,162.61);
\definecolor{drawColor}{RGB}{0,0,0}

\path[draw=drawColor,line width= 0.6pt,dash pattern=on 2pt off 2pt ,line join=round] (181.83, 53.23) -- (193.40, 53.23);
\end{scope}
\begin{scope}
\path[clip] (  0.00,  0.00) rectangle (234.88,162.61);
\definecolor{drawColor}{RGB}{0,0,0}

\node[text=drawColor,anchor=base west,inner sep=0pt, outer sep=0pt, scale=  0.88] at (200.34, 64.65) {\ShortBerlinOne};
\end{scope}
\begin{scope}
\path[clip] (  0.00,  0.00) rectangle (234.88,162.61);
\definecolor{drawColor}{RGB}{0,0,0}

\node[text=drawColor,anchor=base west,inner sep=0pt, outer sep=0pt, scale=  0.88] at (200.34, 50.20) {\ShortBerlinTen};
\end{scope}
\end{tikzpicture}

%% file: last_stop_sorted_speedups.tex
\begin{tikzpicture}[x=1pt,y=1pt]
\definecolor{fillColor}{RGB}{255,255,255}
\path[use as bounding box,fill=fillColor,fill opacity=0.00] (0,0) rectangle (505.89,162.61);
\begin{scope}
\path[clip] (  0.00,  0.00) rectangle (202.36,162.61);
\definecolor{drawColor}{RGB}{255,255,255}
\definecolor{fillColor}{RGB}{255,255,255}

\path[draw=drawColor,line width= 0.6pt,line join=round,line cap=round,fill=fillColor] (  0.00,  0.00) rectangle (202.36,162.61);
\end{scope}
\begin{scope}
\path[clip] ( 38.85, 31.25) rectangle (196.86,157.11);
\definecolor{fillColor}{RGB}{255,255,255}

\path[fill=fillColor] ( 38.85, 31.25) rectangle (196.86,157.11);
\definecolor{drawColor}{RGB}{50,136,189}

\path[draw=drawColor,line width= 0.4pt,line join=round,line cap=round] ( 37.69, 51.46) rectangle ( 41.62, 55.38);
\definecolor{drawColor}{RGB}{213,62,79}

\path[draw=drawColor,line width= 0.4pt,line join=round,line cap=round] ( 39.65, 67.30) circle (  1.96);
\definecolor{drawColor}{RGB}{50,136,189}

\path[draw=drawColor,line width= 0.4pt,line join=round,line cap=round] ( 47.22, 58.38) rectangle ( 51.14, 62.31);
\definecolor{drawColor}{RGB}{213,62,79}

\path[draw=drawColor,line width= 0.4pt,line join=round,line cap=round] ( 49.05, 76.37) circle (  1.96);
\definecolor{drawColor}{RGB}{50,136,189}

\path[draw=drawColor,line width= 0.4pt,line join=round,line cap=round] ( 73.07, 61.75) rectangle ( 76.99, 65.68);
\definecolor{drawColor}{RGB}{213,62,79}

\path[draw=drawColor,line width= 0.4pt,line join=round,line cap=round] ( 74.79, 75.99) circle (  1.96);
\definecolor{drawColor}{RGB}{50,136,189}

\path[draw=drawColor,line width= 0.4pt,line join=round,line cap=round] (117.84, 65.61) rectangle (121.76, 69.53);
\definecolor{drawColor}{RGB}{213,62,79}

\path[draw=drawColor,line width= 0.4pt,line join=round,line cap=round] (119.43, 74.23) circle (  1.96);
\definecolor{drawColor}{RGB}{50,136,189}

\path[draw=drawColor,line width= 0.4pt,line join=round,line cap=round] (180.85, 67.43) rectangle (184.78, 71.36);
\definecolor{drawColor}{RGB}{213,62,79}

\path[draw=drawColor,line width= 0.4pt,line join=round,line cap=round] (182.79, 73.15) circle (  1.96);
\definecolor{drawColor}{RGB}{50,136,189}

\path[draw=drawColor,line width= 0.4pt,line join=round,line cap=round] ( 37.69, 73.18) rectangle ( 41.62, 77.10);
\definecolor{drawColor}{RGB}{213,62,79}

\path[draw=drawColor,line width= 0.4pt,line join=round,line cap=round] ( 39.65,102.70) circle (  1.96);
\definecolor{drawColor}{RGB}{50,136,189}

\path[draw=drawColor,line width= 0.4pt,line join=round,line cap=round] ( 47.54,113.57) rectangle ( 51.46,117.49);
\definecolor{drawColor}{RGB}{213,62,79}

\path[draw=drawColor,line width= 0.4pt,line join=round,line cap=round] ( 49.56,118.81) circle (  1.96);
\definecolor{drawColor}{RGB}{50,136,189}

\path[draw=drawColor,line width= 0.4pt,line join=round,line cap=round] ( 74.24,134.69) rectangle ( 78.16,138.61);
\definecolor{drawColor}{RGB}{213,62,79}

\path[draw=drawColor,line width= 0.4pt,line join=round,line cap=round] ( 76.42,120.83) circle (  1.96);
\definecolor{drawColor}{RGB}{50,136,189}

\path[draw=drawColor,line width= 0.4pt,line join=round,line cap=round] (120.32,147.34) rectangle (124.25,151.27);
\definecolor{drawColor}{RGB}{213,62,79}

\path[draw=drawColor,line width= 0.4pt,line join=round,line cap=round] (123.16,123.33) circle (  1.96);
\definecolor{drawColor}{RGB}{50,136,189}

\path[draw=drawColor,line width= 0.4pt,line join=round,line cap=round] (185.15,149.15) rectangle (189.07,153.08);
\definecolor{drawColor}{RGB}{213,62,79}

\path[draw=drawColor,line width= 0.4pt,line join=round,line cap=round] (189.33,122.99) circle (  1.96);
\definecolor{drawColor}{RGB}{50,136,189}

\path[draw=drawColor,line width= 0.6pt,line join=round] ( 39.65, 53.42) --
	( 49.18, 60.35) --
	( 75.03, 63.71) --
	(119.80, 67.57) --
	(182.81, 69.40);

\path[draw=drawColor,line width= 0.6pt,dash pattern=on 4pt off 4pt ,line join=round] ( 39.66, 75.14) --
	( 49.50,115.53) --
	( 76.20,136.65) --
	(122.29,149.30) --
	(187.11,151.11);
\definecolor{drawColor}{RGB}{213,62,79}

\path[draw=drawColor,line width= 0.6pt,line join=round] ( 39.65, 67.30) --
	( 49.05, 76.37) --
	( 74.79, 75.99) --
	(119.43, 74.23) --
	(182.79, 73.15);

\path[draw=drawColor,line width= 0.6pt,dash pattern=on 4pt off 4pt ,line join=round] ( 39.65,102.70) --
	( 49.56,118.81) --
	( 76.42,120.83) --
	(123.16,123.33) --
	(189.33,122.99);
\definecolor{drawColor}{RGB}{0,0,0}

\path[draw=drawColor,line width= 0.6pt,dash pattern=on 1pt off 3pt ,line join=round] ( 38.85, 43.31) -- (196.86, 43.31);
\end{scope}
\begin{scope}
\path[clip] (  0.00,  0.00) rectangle (505.89,162.61);
\definecolor{drawColor}{RGB}{0,0,0}

\path[draw=drawColor,line width= 0.6pt,line join=round] ( 38.85, 31.25) --
	( 38.85,157.11);
\end{scope}
\begin{scope}
\path[clip] (  0.00,  0.00) rectangle (505.89,162.61);
\definecolor{drawColor}{gray}{0.30}

\node[text=drawColor,anchor=base east,inner sep=0pt, outer sep=0pt, scale=  0.88] at ( 33.90, 28.22) { 0.0};

\node[text=drawColor,anchor=base east,inner sep=0pt, outer sep=0pt, scale=  0.88] at ( 33.90, 58.37) { 2.5};

\node[text=drawColor,anchor=base east,inner sep=0pt, outer sep=0pt, scale=  0.88] at ( 33.90, 88.52) { 5.0};

\node[text=drawColor,anchor=base east,inner sep=0pt, outer sep=0pt, scale=  0.88] at ( 33.90,118.67) { 7.5};

\node[text=drawColor,anchor=base east,inner sep=0pt, outer sep=0pt, scale=  0.88] at ( 33.90,148.82) {10.0};
\end{scope}
\begin{scope}
\path[clip] (  0.00,  0.00) rectangle (505.89,162.61);
\definecolor{drawColor}{gray}{0.20}

\path[draw=drawColor,line width= 0.6pt,line join=round] ( 36.10, 31.25) --
	( 38.85, 31.25);

\path[draw=drawColor,line width= 0.6pt,line join=round] ( 36.10, 61.40) --
	( 38.85, 61.40);

\path[draw=drawColor,line width= 0.6pt,line join=round] ( 36.10, 91.55) --
	( 38.85, 91.55);

\path[draw=drawColor,line width= 0.6pt,line join=round] ( 36.10,121.70) --
	( 38.85,121.70);

\path[draw=drawColor,line width= 0.6pt,line join=round] ( 36.10,151.85) --
	( 38.85,151.85);
\end{scope}
\begin{scope}
\path[clip] (  0.00,  0.00) rectangle (505.89,162.61);
\definecolor{drawColor}{RGB}{0,0,0}

\path[draw=drawColor,line width= 0.6pt,line join=round] ( 38.85, 31.25) --
	(196.86, 31.25);
\end{scope}
\begin{scope}
\path[clip] (  0.00,  0.00) rectangle (505.89,162.61);
\definecolor{drawColor}{gray}{0.20}

\path[draw=drawColor,line width= 0.6pt,line join=round] ( 38.85, 28.50) --
	( 38.85, 31.25);

\path[draw=drawColor,line width= 0.6pt,line join=round] ( 79.25, 28.50) --
	( 79.25, 31.25);

\path[draw=drawColor,line width= 0.6pt,line join=round] (119.66, 28.50) --
	(119.66, 31.25);

\path[draw=drawColor,line width= 0.6pt,line join=round] (160.06, 28.50) --
	(160.06, 31.25);
\end{scope}
\begin{scope}
\path[clip] (  0.00,  0.00) rectangle (505.89,162.61);
\definecolor{drawColor}{gray}{0.30}

\node[text=drawColor,anchor=base,inner sep=0pt, outer sep=0pt, scale=  0.88] at ( 38.85, 20.24) {0};

\node[text=drawColor,anchor=base,inner sep=0pt, outer sep=0pt, scale=  0.88] at ( 79.25, 20.24) {50};

\node[text=drawColor,anchor=base,inner sep=0pt, outer sep=0pt, scale=  0.88] at (119.66, 20.24) {100};

\node[text=drawColor,anchor=base,inner sep=0pt, outer sep=0pt, scale=  0.88] at (160.06, 20.24) {150};
\end{scope}
\begin{scope}
\path[clip] (  0.00,  0.00) rectangle (505.89,162.61);
\definecolor{drawColor}{RGB}{0,0,0}

\node[text=drawColor,anchor=base,inner sep=0pt, outer sep=0pt, scale=  1.10] at (117.85,  7.93) {$\numpickups$, $\numdropoffs$};
\end{scope}
\begin{scope}
\path[clip] (  0.00,  0.00) rectangle (505.89,162.61);
\definecolor{drawColor}{RGB}{0,0,0}

\node[text=drawColor,rotate= 90.00,anchor=base,inner sep=0pt, outer sep=0pt, scale=  1.10] at ( 13.08, 94.18) {speedup (ind. BCH)};
\end{scope}
\begin{scope}
\path[clip] (202.36,  0.00) rectangle (404.71,162.61);
\definecolor{drawColor}{RGB}{255,255,255}
\definecolor{fillColor}{RGB}{255,255,255}

\path[draw=drawColor,line width= 0.6pt,line join=round,line cap=round,fill=fillColor] (202.36,  0.00) rectangle (404.71,162.61);
\end{scope}
\begin{scope}
\path[clip] (241.20, 31.25) rectangle (399.21,157.11);
\definecolor{fillColor}{RGB}{255,255,255}

\path[fill=fillColor] (241.20, 31.25) rectangle (399.21,157.11);
\definecolor{drawColor}{RGB}{50,136,189}

\path[draw=drawColor,line width= 0.4pt,line join=round,line cap=round] (240.05, 70.36) rectangle (243.97, 74.28);
\definecolor{drawColor}{RGB}{213,62,79}

\path[draw=drawColor,line width= 0.4pt,line join=round,line cap=round] (242.01, 64.67) circle (  1.96);
\definecolor{drawColor}{RGB}{50,136,189}

\path[draw=drawColor,line width= 0.4pt,line join=round,line cap=round] (249.58, 58.75) rectangle (253.50, 62.68);
\definecolor{drawColor}{RGB}{213,62,79}

\path[draw=drawColor,line width= 0.4pt,line join=round,line cap=round] (251.41, 62.17) circle (  1.96);
\definecolor{drawColor}{RGB}{50,136,189}

\path[draw=drawColor,line width= 0.4pt,line join=round,line cap=round] (275.42, 55.74) rectangle (279.35, 59.67);
\definecolor{drawColor}{RGB}{213,62,79}

\path[draw=drawColor,line width= 0.4pt,line join=round,line cap=round] (277.14, 59.58) circle (  1.96);
\definecolor{drawColor}{RGB}{50,136,189}

\path[draw=drawColor,line width= 0.4pt,line join=round,line cap=round] (320.19, 54.04) rectangle (324.12, 57.97);
\definecolor{drawColor}{RGB}{213,62,79}

\path[draw=drawColor,line width= 0.4pt,line join=round,line cap=round] (321.78, 57.32) circle (  1.96);
\definecolor{drawColor}{RGB}{50,136,189}

\path[draw=drawColor,line width= 0.4pt,line join=round,line cap=round] (383.21, 51.29) rectangle (387.13, 55.21);
\definecolor{drawColor}{RGB}{213,62,79}

\path[draw=drawColor,line width= 0.4pt,line join=round,line cap=round] (385.15, 55.42) circle (  1.96);
\definecolor{drawColor}{RGB}{50,136,189}

\path[draw=drawColor,line width= 0.4pt,line join=round,line cap=round] (240.05,149.15) rectangle (243.97,153.08);
\definecolor{drawColor}{RGB}{213,62,79}

\path[draw=drawColor,line width= 0.4pt,line join=round,line cap=round] (242.01,103.42) circle (  1.96);
\definecolor{drawColor}{RGB}{50,136,189}

\path[draw=drawColor,line width= 0.4pt,line join=round,line cap=round] (249.90,111.66) rectangle (253.82,115.58);
\definecolor{drawColor}{RGB}{213,62,79}

\path[draw=drawColor,line width= 0.4pt,line join=round,line cap=round] (251.92,102.47) circle (  1.96);
\definecolor{drawColor}{RGB}{50,136,189}

\path[draw=drawColor,line width= 0.4pt,line join=round,line cap=round] (276.59,107.82) rectangle (280.52,111.74);
\definecolor{drawColor}{RGB}{213,62,79}

\path[draw=drawColor,line width= 0.4pt,line join=round,line cap=round] (278.78, 97.53) circle (  1.96);
\definecolor{drawColor}{RGB}{50,136,189}

\path[draw=drawColor,line width= 0.4pt,line join=round,line cap=round] (322.68, 82.23) rectangle (326.60, 86.15);
\definecolor{drawColor}{RGB}{213,62,79}

\path[draw=drawColor,line width= 0.4pt,line join=round,line cap=round] (325.51, 91.87) circle (  1.96);
\definecolor{drawColor}{RGB}{50,136,189}

\path[draw=drawColor,line width= 0.4pt,line join=round,line cap=round] (387.50, 62.27) rectangle (391.43, 66.19);
\definecolor{drawColor}{RGB}{213,62,79}

\path[draw=drawColor,line width= 0.4pt,line join=round,line cap=round] (391.69, 85.67) circle (  1.96);
\definecolor{drawColor}{RGB}{50,136,189}

\path[draw=drawColor,line width= 0.6pt,line join=round] (242.01, 72.32) --
	(251.54, 60.72) --
	(277.38, 57.71) --
	(322.16, 56.01) --
	(385.17, 53.25);

\path[draw=drawColor,line width= 0.6pt,dash pattern=on 4pt off 4pt ,line join=round] (242.01,151.11) --
	(251.86,113.62) --
	(278.56,109.78) --
	(324.64, 84.19) --
	(389.47, 64.23);
\definecolor{drawColor}{RGB}{213,62,79}

\path[draw=drawColor,line width= 0.6pt,line join=round] (242.01, 64.67) --
	(251.41, 62.17) --
	(277.14, 59.58) --
	(321.78, 57.32) --
	(385.15, 55.42);

\path[draw=drawColor,line width= 0.6pt,dash pattern=on 4pt off 4pt ,line join=round] (242.01,103.42) --
	(251.92,102.47) --
	(278.78, 97.53) --
	(325.51, 91.87) --
	(391.69, 85.67);
\definecolor{drawColor}{RGB}{0,0,0}

\path[draw=drawColor,line width= 0.6pt,dash pattern=on 1pt off 3pt ,line join=round] (241.20, 48.00) -- (399.21, 48.00);
\end{scope}
\begin{scope}
\path[clip] (  0.00,  0.00) rectangle (505.89,162.61);
\definecolor{drawColor}{RGB}{0,0,0}

\path[draw=drawColor,line width= 0.6pt,line join=round] (241.20, 31.25) --
	(241.20,157.11);
\end{scope}
\begin{scope}
\path[clip] (  0.00,  0.00) rectangle (505.89,162.61);
\definecolor{drawColor}{gray}{0.30}

\node[text=drawColor,anchor=base east,inner sep=0pt, outer sep=0pt, scale=  0.88] at (236.25, 28.22) {0.0};

\node[text=drawColor,anchor=base east,inner sep=0pt, outer sep=0pt, scale=  0.88] at (236.25, 61.73) {2.0};

\node[text=drawColor,anchor=base east,inner sep=0pt, outer sep=0pt, scale=  0.88] at (236.25, 95.23) {4.0};

\node[text=drawColor,anchor=base east,inner sep=0pt, outer sep=0pt, scale=  0.88] at (236.25,128.73) {6.0};
\end{scope}
\begin{scope}
\path[clip] (  0.00,  0.00) rectangle (505.89,162.61);
\definecolor{drawColor}{gray}{0.20}

\path[draw=drawColor,line width= 0.6pt,line join=round] (238.45, 31.25) --
	(241.20, 31.25);

\path[draw=drawColor,line width= 0.6pt,line join=round] (238.45, 64.76) --
	(241.20, 64.76);

\path[draw=drawColor,line width= 0.6pt,line join=round] (238.45, 98.26) --
	(241.20, 98.26);

\path[draw=drawColor,line width= 0.6pt,line join=round] (238.45,131.76) --
	(241.20,131.76);
\end{scope}
\begin{scope}
\path[clip] (  0.00,  0.00) rectangle (505.89,162.61);
\definecolor{drawColor}{RGB}{0,0,0}

\path[draw=drawColor,line width= 0.6pt,line join=round] (241.20, 31.25) --
	(399.21, 31.25);
\end{scope}
\begin{scope}
\path[clip] (  0.00,  0.00) rectangle (505.89,162.61);
\definecolor{drawColor}{gray}{0.20}

\path[draw=drawColor,line width= 0.6pt,line join=round] (241.20, 28.50) --
	(241.20, 31.25);

\path[draw=drawColor,line width= 0.6pt,line join=round] (281.61, 28.50) --
	(281.61, 31.25);

\path[draw=drawColor,line width= 0.6pt,line join=round] (322.01, 28.50) --
	(322.01, 31.25);

\path[draw=drawColor,line width= 0.6pt,line join=round] (362.42, 28.50) --
	(362.42, 31.25);
\end{scope}
\begin{scope}
\path[clip] (  0.00,  0.00) rectangle (505.89,162.61);
\definecolor{drawColor}{gray}{0.30}

\node[text=drawColor,anchor=base,inner sep=0pt, outer sep=0pt, scale=  0.88] at (241.20, 20.24) {0};

\node[text=drawColor,anchor=base,inner sep=0pt, outer sep=0pt, scale=  0.88] at (281.61, 20.24) {50};

\node[text=drawColor,anchor=base,inner sep=0pt, outer sep=0pt, scale=  0.88] at (322.01, 20.24) {100};

\node[text=drawColor,anchor=base,inner sep=0pt, outer sep=0pt, scale=  0.88] at (362.42, 20.24) {150};
\end{scope}
\begin{scope}
\path[clip] (  0.00,  0.00) rectangle (505.89,162.61);
\definecolor{drawColor}{RGB}{0,0,0}

\node[text=drawColor,anchor=base,inner sep=0pt, outer sep=0pt, scale=  1.10] at (320.21,  7.93) {$\numpickups$, $\numdropoffs$};
\end{scope}
\begin{scope}
\path[clip] (  0.00,  0.00) rectangle (505.89,162.61);
\definecolor{drawColor}{RGB}{0,0,0}

\node[text=drawColor,rotate= 90.00,anchor=base,inner sep=0pt, outer sep=0pt, scale=  1.10] at (215.43, 94.18) {speedup (col. BCH)};
\end{scope}
\begin{scope}
\path[clip] (  0.00,  0.00) rectangle (505.89,162.61);
\definecolor{fillColor}{RGB}{255,255,255}

\path[fill=fillColor] (413.32, 86.80) rectangle (483.23,131.30);
\end{scope}
\begin{scope}
\path[clip] (  0.00,  0.00) rectangle (505.89,162.61);
\definecolor{drawColor}{RGB}{0,0,0}

\node[text=drawColor,anchor=base west,inner sep=0pt, outer sep=0pt, scale=  1.10] at (413.32,122.24) {Config.};
\end{scope}
\begin{scope}
\path[clip] (  0.00,  0.00) rectangle (505.89,162.61);
\definecolor{drawColor}{RGB}{50,136,189}

\path[draw=drawColor,line width= 0.4pt,line join=round,line cap=round] (418.47,106.18) rectangle (422.39,110.11);
\end{scope}
\begin{scope}
\path[clip] (  0.00,  0.00) rectangle (505.89,162.61);
\definecolor{drawColor}{RGB}{50,136,189}

\path[draw=drawColor,line width= 0.6pt,line join=round] (414.74,108.14) -- (426.12,108.14);
\end{scope}
\begin{scope}
\path[clip] (  0.00,  0.00) rectangle (505.89,162.61);
\definecolor{drawColor}{RGB}{213,62,79}

\path[draw=drawColor,line width= 0.4pt,line join=round,line cap=round] (420.43, 93.92) circle (  1.96);
\end{scope}
\begin{scope}
\path[clip] (  0.00,  0.00) rectangle (505.89,162.61);
\definecolor{drawColor}{RGB}{213,62,79}

\path[draw=drawColor,line width= 0.6pt,line join=round] (414.74, 93.92) -- (426.12, 93.92);
\end{scope}
\begin{scope}
\path[clip] (  0.00,  0.00) rectangle (505.89,162.61);
\definecolor{drawColor}{RGB}{0,0,0}

\node[text=drawColor,anchor=base west,inner sep=0pt, outer sep=0pt, scale=  0.88] at (433.04,105.11) {PALS sorted};
\end{scope}
\begin{scope}
\path[clip] (  0.00,  0.00) rectangle (505.89,162.61);
\definecolor{drawColor}{RGB}{0,0,0}

\node[text=drawColor,anchor=base west,inner sep=0pt, outer sep=0pt, scale=  0.88] at (433.04, 90.89) {DALS sorted};
\end{scope}
\begin{scope}
\path[clip] (  0.00,  0.00) rectangle (505.89,162.61);
\definecolor{fillColor}{RGB}{255,255,255}

\path[fill=fillColor] (413.32, 31.31) rectangle (497.28, 75.80);
\end{scope}
\begin{scope}
\path[clip] (  0.00,  0.00) rectangle (505.89,162.61);
\definecolor{drawColor}{RGB}{0,0,0}

\node[text=drawColor,anchor=base west,inner sep=0pt, outer sep=0pt, scale=  1.10] at (413.32, 66.75) {Instance};
\end{scope}
\begin{scope}
\path[clip] (  0.00,  0.00) rectangle (505.89,162.61);
\definecolor{drawColor}{RGB}{0,0,0}

\path[draw=drawColor,line width= 0.6pt,line join=round] (414.74, 52.65) -- (426.12, 52.65);
\end{scope}
\begin{scope}
\path[clip] (  0.00,  0.00) rectangle (505.89,162.61);
\definecolor{drawColor}{RGB}{0,0,0}

\path[draw=drawColor,line width= 0.6pt,dash pattern=on 4pt off 4pt ,line join=round] (414.74, 38.42) -- (426.12, 38.42);
\end{scope}
\begin{scope}
\path[clip] (  0.00,  0.00) rectangle (505.89,162.61);
\definecolor{drawColor}{RGB}{0,0,0}

\node[text=drawColor,anchor=base west,inner sep=0pt, outer sep=0pt, scale=  0.88] at (433.04, 49.62) {\ShortBerlinOne};
\end{scope}
\begin{scope}
\path[clip] (  0.00,  0.00) rectangle (505.89,162.61);
\definecolor{drawColor}{RGB}{0,0,0}

\node[text=drawColor,anchor=base west,inner sep=0pt, outer sep=0pt, scale=  0.88] at (433.04, 35.39) {\ShortBerlinTen};
\end{scope}
\end{tikzpicture}

%% file: karri.bbl
\begin{thebibliography}{74}
\providecommand{\natexlab}[1]{#1}
\providecommand{\url}[1]{\texttt{#1}}
\expandafter\ifx\csname urlstyle\endcsname\relax
  \providecommand{\doi}[1]{doi: #1}\else
  \providecommand{\doi}{doi: \begingroup \urlstyle{rm}\Url}\fi

\bibitem[Abraham et~al.(2011)Abraham, Delling, Goldberg, and
  Werneck]{Abraham2011}
Ittai Abraham, Daniel Delling, Andrew~V. Goldberg, and Renato~F. Werneck.
\newblock A hub-based labeling algorithm for shortest paths in road networks.
\newblock In \emph{Lecture Notes in Computer Science (including subseries
  Lecture Notes in Artificial Intelligence and Lecture Notes in
  Bioinformatics)}, volume 6630 LNCS, pages 230--241. Springer, Berlin,
  Heidelberg, 2011.
\newblock \doi{10.1007/978-3-642-20662-7_20}.

\bibitem[Agatz et~al.(2011)Agatz, Erera, Savelsbergh, and Wang]{Agatz2011}
Niels Agatz, Alan Erera, Martin Savelsbergh, and Xing Wang.
\newblock Dynamic ride-sharing: A simulation study in metro {Atlanta}.
\newblock \emph{Transportation Research Part B: Methodological}, 45:\penalty0
  1450--1464, 2011.
\newblock ISSN 01912615.
\newblock \doi{10.1016/j.trb.2011.05.017}.

\bibitem[Agatz et~al.(2012)Agatz, Erera, Savelsbergh, and Wang]{Agatz2012}
Niels Agatz, Alan Erera, Martin Savelsbergh, and Xing Wang.
\newblock Optimization for dynamic ride-sharing: A review.
\newblock \emph{European Journal of Operational Research}, 223:\penalty0
  295--303, 2012.
\newblock ISSN 03772217.
\newblock \doi{10.1016/j.ejor.2012.05.028}.

\bibitem[Aissat and Oulamara(2014)]{Aissat2014}
Kamel Aissat and Ammar Oulamara.
\newblock A priori approach of real-time ridesharing problem with intermediate
  meeting locations.
\newblock \emph{Journal of Artificial Intelligence and Soft Computing
  Research}, 4:\penalty0 287--299, 2014.
\newblock ISSN 2083-2567.
\newblock \doi{10.1515/jaiscr-2015-0015}.

\bibitem[Alonso-Mora et~al.(2017)Alonso-Mora, Samaranayake, Wallar, Frazzoli,
  and Rus]{AlonsoMora2017}
Javier Alonso-Mora, Samitha Samaranayake, Alex Wallar, Emilio Frazzoli, and
  Daniela Rus.
\newblock On-demand high-capacity ride-sharing via dynamic trip-vehicle
  assignment.
\newblock \emph{Proceedings of the National Academy of Sciences of the United
  States of America}, 114:\penalty0 462--467, 2017.
\newblock ISSN 10916490.
\newblock \doi{10.1073/pnas.1611675114}.

\bibitem[Badue et~al.(2021)Badue, Guidolini, Carneiro, Azevedo, Cardoso,
  Forechi, Jesus, Berriel, Paixão, Mutz, de~Paula~Veronese, Oliveira-Santos,
  and Souza]{Badue2021}
Claudine Badue, Rânik Guidolini, Raphael~Vivacqua Carneiro, Pedro Azevedo,
  Vinicius~B. Cardoso, Avelino Forechi, Luan Jesus, Rodrigo Berriel, Thiago~M.
  Paixão, Filipe Mutz, Lucas de~Paula~Veronese, Thiago Oliveira-Santos, and
  Alberto F.~De Souza.
\newblock Self-driving cars: A survey.
\newblock \emph{Expert Systems with Applications}, 165, 2021.
\newblock ISSN 09574174.
\newblock \doi{10.1016/j.eswa.2020.113816}.

\bibitem[Bast et~al.(2007)Bast, Funke, Sanders, and Schultes]{Bast2007}
Holger Bast, Stefan Funke, Peter Sanders, and Dominik Schultes.
\newblock Fast routing in road networks with transit nodes.
\newblock \emph{Science}, 316:\penalty0 566--566, 2007.
\newblock ISSN 0036-8075.
\newblock \doi{10.1126/science.1137521}.

\bibitem[Bauer and Delling(2010)]{Bauer2010}
Reinhard Bauer and Daniel Delling.
\newblock {SHARC}: Fast and robust unidirectional routing.
\newblock \emph{ACM Journal of Experimental Algorithms (JEA)}, 14, 2010.
\newblock ISSN 1084-6654.
\newblock \doi{10.1145/1498698.1537599}.

\bibitem[Bischoff et~al.(2017)Bischoff, Maciejewski, and Nagel]{Bischoff2017}
Joschka Bischoff, Michal Maciejewski, and Kai Nagel.
\newblock City-wide shared taxis: A simulation study in {Berlin}.
\newblock In \emph{IEEE 20th International Conference on Intelligent
  Transportation Systems (ITSC)}, 2017.
\newblock \doi{10.1109/ITSC.2017.8317926}.

\bibitem[Bistaffa et~al.(2015)Bistaffa, Farinelli, and Ramchurn]{Bistaffa2015}
Filippo Bistaffa, Alessandro Farinelli, and Sarvapali Ramchurn.
\newblock Sharing rides with friends: A coalition formation algorithm for
  ridesharing.
\newblock \emph{Proceedings of the AAAI Conference on Artificial Intelligence},
  29, 2015.
\newblock ISSN 2374-3468.
\newblock \doi{10.1609/aaai.v29i1.9242}.

\bibitem[Buchhold et~al.(2019)Buchhold, Sanders, and Wagner]{Buchhold2019}
Valentin Buchhold, Peter Sanders, and Dorothea Wagner.
\newblock Real-time traffic assignment using engineered customizable
  contraction hierarchies.
\newblock \emph{ACM Journal of Experimental Algorithms (JEA)}, 24, 2019.
\newblock ISSN 1084-6654.
\newblock \doi{10.1145/3362693}.

\bibitem[Buchhold et~al.(2021)Buchhold, Sanders, and Wagner]{Buchhold2021}
Valentin Buchhold, Peter Sanders, and Dorothea Wagner.
\newblock Fast, exact and scalable dynamic ridesharing.
\newblock In \emph{2021 Proceedings of the Workshop on Algorithm Engineering
  and Experiments (ALENEX)}, pages 98--112. Society for Industrial and Applied
  Mathematics, 2021.
\newblock ISBN 9781611976472.
\newblock \doi{10.1137/1.9781611976472.8}.

\bibitem[Cordeau(2006)]{Cordeau2006}
Jean~François Cordeau.
\newblock A branch-and-cut algorithm for the dial-a-ride problem.
\newblock \emph{Operations Research}, 54:\penalty0 573--586, 2006.
\newblock ISSN 0030364X.
\newblock \doi{10.1287/opre.1060.0283}.

\bibitem[Cordeau and Laporte(2007)]{Cordeau2007}
Jean~François Cordeau and Gilbert Laporte.
\newblock The dial-a-ride problem: Models and algorithms.
\newblock \emph{Annals of Operations Research}, 153:\penalty0 29--46, 2007.
\newblock ISSN 02545330.
\newblock \doi{10.1007/s10479-007-0170-8}.

\bibitem[Delling et~al.(2011)Delling, Goldberg, and Werneck]{Delling2011}
Daniel Delling, Andrew~V. Goldberg, and Renato~F. Werneck.
\newblock Faster batched shortest paths in road networks.
\newblock In \emph{11th Workshop on Algorithmic Approaches for Transportation
  Modelling, Optimization, and Systems (ATMOS)}. LIPIcs, 2011.
\newblock ISBN 978-3-939897-33-0.
\newblock \doi{10.4230/OASIcs.ATMOS.2011.52}.

\bibitem[Delling et~al.(2013)Delling, Goldberg, Nowatzyk, and
  Werneck]{Delling2013}
Daniel Delling, Andrew~V. Goldberg, Andreas Nowatzyk, and Renato~F. Werneck.
\newblock {PHAST}: Hardware-accelerated shortest path trees.
\newblock \emph{Journal of Parallel and Distributed Computing}, 73:\penalty0
  940--952, 2013.
\newblock ISSN 07437315.
\newblock \doi{10.1016/j.jpdc.2012.02.007}.

\bibitem[Delling et~al.(2017)Delling, Goldberg, Pajor, and
  Werneck]{Delling2017}
Daniel Delling, Andrew~V. Goldberg, Thomas Pajor, and Renato~F. Werneck.
\newblock Customizable route planning in road networks.
\newblock \emph{INFORMS Transportation Science}, 51, 2017.
\newblock \doi{10.1287/trsc.2014.0579}.

\bibitem[Dibbelt et~al.(2016)Dibbelt, Strasser, and Wagner]{Dibbelt2016}
Julian Dibbelt, Ben Strasser, and Dorothea Wagner.
\newblock Customizable contraction hierarchies.
\newblock \emph{ACM Journal of Experimental Algorithmics}, 21:\penalty0 1--49,
  2016.
\newblock ISSN 1084-6654.
\newblock \doi{10.1145/2886843}.

\bibitem[Dijkstra(1959)]{Dijkstra1959}
Edsger~W Dijkstra.
\newblock A note on two problems in connexion with graphs.
\newblock \emph{Numerische Mathematik}, 1, 1959.

\bibitem[Duarte and Ratti(2018)]{Duarte2018}
Fábio Duarte and Carlo Ratti.
\newblock The impact of autonomous vehicles on cities: A review.
\newblock \emph{Journal of Urban Technology}, 25:\penalty0 3--18, 2018.
\newblock ISSN 1063-0732.
\newblock \doi{10.1080/10630732.2018.1493883}.

\bibitem[Fagnant and Kockelman(2014)]{Fagnant2014}
Daniel~J. Fagnant and Kara~M. Kockelman.
\newblock The travel and environmental implications of shared autonomous
  vehicles, using agent-based model scenarios.
\newblock \emph{Transportation Research Part C: Emerging Technologies},
  40:\penalty0 1--13, 2014.
\newblock ISSN 0968090X.
\newblock \doi{10.1016/j.trc.2013.12.001}.

\bibitem[Fagnant and Kockelman(2018)]{Fagnant2018}
Daniel~J. Fagnant and Kara~M. Kockelman.
\newblock Dynamic ride-sharing and fleet sizing for a system of shared
  autonomous vehicles in {Austin}, {Texas}.
\newblock \emph{Transportation}, 45:\penalty0 143--158, 2018.
\newblock ISSN 0049-4488.
\newblock \doi{10.1007/s11116-016-9729-z}.

\bibitem[Fielbaum et~al.(2021)Fielbaum, Bai, and Alonso-Mora]{Fielbaum2021}
Andres Fielbaum, Xiaoshan Bai, and Javier Alonso-Mora.
\newblock On-demand ridesharing with optimized pick-up and drop-off walking
  locations.
\newblock \emph{Transportation Research Part C: Emerging Technologies},
  126:\penalty0 103061, 2021.
\newblock ISSN 0968090X.
\newblock \doi{10.1016/j.trc.2021.103061}.

\bibitem[Furuhata et~al.(2013)Furuhata, Dessouky, Ordóñez, Brunet, Wang, and
  Koenig]{Furuhata2013}
Masabumi Furuhata, Maged Dessouky, Fernando Ordóñez, Marc~Etienne Brunet,
  Xiaoqing Wang, and Sven Koenig.
\newblock Ridesharing: The state-of-the-art and future directions.
\newblock \emph{Transportation Research Part B: Methodological}, 57:\penalty0
  28--46, 2013.
\newblock ISSN 01912615.
\newblock \doi{10.1016/j.trb.2013.08.012}.

\bibitem[Gargiulo et~al.(2015)Gargiulo, Giannantonio, Guercio, Borean, and
  Zenezini]{Gargiulo2015}
Eleonora Gargiulo, Roberta Giannantonio, Elena Guercio, Claudio Borean, and
  Giovanni Zenezini.
\newblock Dynamic ride sharing service: Are users ready to adopt it?
\newblock \emph{Procedia Manufacturing}, 3:\penalty0 777--784, 2015.
\newblock ISSN 23519789.
\newblock \doi{10.1016/j.promfg.2015.07.329}.

\bibitem[Geisberger et~al.(2010)Geisberger, Luxen, Neubauer, Sanders, and
  Volker]{Geisberger2010}
Robert Geisberger, Dennis Luxen, Sabine Neubauer, Peter Sanders, and Lars
  Volker.
\newblock Fast detour computation for ride sharing.
\newblock In \emph{OpenAccess Series in Informatics}, volume~14, pages 88--99.
  Schloss Dagstuhl - Leibniz-Zentrum fuer Informatik, 2010.
\newblock ISBN 9783939897200.
\newblock \doi{10.4230/OASIcs.ATMOS.2010.88}.

\bibitem[Geisberger et~al.(2012)Geisberger, Sanders, Schultes, and
  Vetter]{Geisberger2012}
Robert Geisberger, Peter Sanders, Dominik Schultes, and Christian Vetter.
\newblock Exact routing in large road networks using contraction hierarchies.
\newblock \emph{INFORMS Transportation Science}, 46, 2012.
\newblock \doi{10.1287/trsc.1110.0401}.

\bibitem[Gilibert et~al.(2020)Gilibert, Ribas, Rosen, and
  Siebeneich]{Gilibert2020}
Mireia Gilibert, Imma Ribas, Christian Rosen, and Alexander Siebeneich.
\newblock On-demand shared ride-hailing for commuting purposes: Comparison of
  {Barcelona} and {Hannover} case studies.
\newblock In \emph{Transportation Research Procedia}, volume~47, pages
  323--330. Elsevier B.V., 2020.
\newblock \doi{10.1016/j.trpro.2020.03.105}.

\bibitem[Goel et~al.(2016{\natexlab{a}})Goel, Kulik, and
  Ramamohanarao]{Goel2016Privacy}
Preeti Goel, Lars Kulik, and Kotagiri Ramamohanarao.
\newblock Privacy-aware dynamic ride sharing.
\newblock \emph{ACM Transactions on Spatial Algorithms and Systems},
  2:\penalty0 1--41, 2016{\natexlab{a}}.
\newblock ISSN 2374-0353.
\newblock \doi{10.1145/2845080}.

\bibitem[Goel et~al.(2016{\natexlab{b}})Goel, Kulik, and
  Ramamohanarao]{Goel2016PuPs}
Preeti Goel, Lars Kulik, and Kotagiri Ramamohanarao.
\newblock Optimal pick up point selection for effective ride sharing.
\newblock \emph{IEEE Transactions on Big Data}, 3:\penalty0 154--168,
  2016{\natexlab{b}}.
\newblock \doi{10.1109/tbdata.2016.2599936}.

\bibitem[Herbawi and Weber(2012)]{Herbawi2012}
Wesam Herbawi and Michael Weber.
\newblock A genetic and insertion heuristic algorithm for solving the dynamic
  ridematching problem with time windows.
\newblock In \emph{GECCO'12 - Proceedings of the 14th International Conference
  on Genetic and Evolutionary Computation}, pages 385--392, 2012.
\newblock ISBN 9781450311779.
\newblock \doi{10.1145/2330163.2330219}.

\bibitem[Hilger et~al.(2009)Hilger, Köhler, Möhring, and
  Schilling]{Hilger2009}
Moritz Hilger, Ekkehard Köhler, Rolf~H. Möhring, and Heiko Schilling.
\newblock Fast point-to-point shortest path computations with arc-flags.
\newblock \emph{The Shortest Path Problem: 9th DIMACS Implementation
  Challenge}, 2009.

\bibitem[Ho et~al.(2018)Ho, Szeto, Kuo, Leung, Petering, and Tou]{Ho2018}
Sin~C. Ho, W.Y. Szeto, Yong-Hong Kuo, Janny~M.Y. Leung, Matthew Petering, and
  Terence~W.H. Tou.
\newblock A survey of dial-a-ride problems: Literature review and recent
  developments.
\newblock \emph{Transportation Research Part B: Methodological}, 111:\penalty0
  395--421, 2018.
\newblock ISSN 01912615.
\newblock \doi{10.1016/j.trb.2018.02.001}.

\bibitem[Horn(2002)]{Horn2002}
Mark~E.T. Horn.
\newblock Fleet scheduling and dispatching for demand-responsive passenger
  services.
\newblock \emph{Transportation Research Part C: Emerging Technologies},
  10:\penalty0 35--63, 2002.
\newblock ISSN 0968090X.
\newblock \doi{10.1016/S0968-090X(01)00003-1}.

\bibitem[Horni et~al.(2016)Horni, Nagel, and Axhausen]{Horni2016}
Andreas Horni, Kai Nagel, and Kay~W. Axhausen, editors.
\newblock \emph{The Multi-Agent Transport Simulation {MATSim}}.
\newblock Ubiquity Press, 2016.
\newblock ISBN 9781909188754.
\newblock \doi{10.5334/baw}.

\bibitem[Hosni et~al.(2014)Hosni, Naoum-Sawaya, and Artail]{Hosni2014}
Hadi Hosni, Joe Naoum-Sawaya, and Hassan Artail.
\newblock The shared-taxi problem: Formulation and solution methods.
\newblock \emph{Transportation Research Part B: Methodological}, 70:\penalty0
  303--318, 2014.
\newblock ISSN 01912615.
\newblock \doi{10.1016/j.trb.2014.09.011}.

\bibitem[Huang et~al.(2014)Huang, Bastani, Jin, and Wang]{Huang2014}
Yan Huang, Favyen Bastani, Ruoming Jin, and Xiaoyang~Sean Wang.
\newblock Large scale realtime ridesharing with service guarantee on road
  networks.
\newblock In \emph{Proceedings of the VLDB Endowment}, volume~7, pages
  2017--2028. Association for Computing Machinery, 2014.
\newblock \doi{10.14778/2733085.2733106}.

\bibitem[Hunsaker and Savelsbergh(2002)]{Hunsaker2002}
Brady Hunsaker and Martin Savelsbergh.
\newblock Efficient feasibility testing for dial-a-ride problems.
\newblock \emph{Operations Research Letters}, 30:\penalty0 169--173, 2002.
\newblock ISSN 01676377.
\newblock \doi{10.1016/S0167-6377(02)00120-7}.

\bibitem[Häll et~al.(2012)Häll, Högberg, and Lundgren]{Haell2012}
Carl~H. Häll, Magdalena Högberg, and Jan~T. Lundgren.
\newblock A modeling system for simulation of dial-a-ride services.
\newblock \emph{Public Transport}, 4:\penalty0 17--37, 2012.
\newblock ISSN 1866-749X.
\newblock \doi{10.1007/s12469-012-0052-6}.

\bibitem[Jaw et~al.(1986)Jaw, Odoni, Psaraftis, and Wilson]{Jaw1986}
Jang-Jei Jaw, Amedeo~R. Odoni, Harilaos~N. Psaraftis, and Nigel~H.M. Wilson.
\newblock A heuristic algorithm for the multi-vehicle advance request
  dial-a-ride problem with time windows.
\newblock \emph{Transportation Research Part B: Methodological}, 20:\penalty0
  243--257, 1986.
\newblock ISSN 01912615.
\newblock \doi{10.1016/0191-2615(86)90020-2}.

\bibitem[Jokinen et~al.(2019)Jokinen, Sihvola, and Mladenovic]{Jokinen2019}
Jani-Pekka Jokinen, Teemu Sihvola, and Milos~N. Mladenovic.
\newblock Policy lessons from the flexible transport service pilot {Kutsuplus}
  in the {Helsinki} capital region.
\newblock \emph{Transport Policy}, 76:\penalty0 123--133, 2019.
\newblock ISSN 0967070X.
\newblock \doi{10.1016/j.tranpol.2017.12.004}.

\bibitem[Jung et~al.(2016)Jung, Jayakrishnan, and Park]{Jung2016}
Jaeyoung Jung, R.~Jayakrishnan, and Ji~Young Park.
\newblock Dynamic shared-taxi dispatch algorithm with hybrid-simulated
  annealing.
\newblock \emph{Computer-Aided Civil and Infrastructure Engineering},
  31:\penalty0 275--291, 2016.
\newblock ISSN 10939687.
\newblock \doi{10.1111/mice.12157}.

\bibitem[Kaan and Olinick(2013)]{Kaan2013}
Levent Kaan and Eli~V. Olinick.
\newblock The vanpool assignment problem: Optimization models and solution
  algorithms.
\newblock \emph{Computers and Industrial Engineering}, 66:\penalty0 24--40,
  2013.
\newblock ISSN 03608352.
\newblock \doi{10.1016/j.cie.2013.05.020}.

\bibitem[Knopp et~al.(2007)Knopp, Sanders, Schultes, Schulz, and
  Wagner]{Knopp2007}
Sebastian Knopp, Peter Sanders, Dominik Schultes, Frank Schulz, and Dorothea
  Wagner.
\newblock Computing many-to-many shortest paths using highway hierarchies.
\newblock In \emph{SIAM Workshop on Algorithm Engineering and Experiments
  (ALENEX)}, 2007.
\newblock \doi{10.1137/1.9781611972870.4}.

\bibitem[Kostorz et~al.(2021)Kostorz, Fraedrich, and Kagerbauer]{Kostorz2021}
Nadine Kostorz, Eva Fraedrich, and Martin Kagerbauer.
\newblock Usage and user characteristics—insights from {MOIA}, europe’s
  largest ridepooling service.
\newblock \emph{Sustainability}, 13:\penalty0 958, 2021.
\newblock ISSN 2071-1050.
\newblock \doi{10.3390/su13020958}.

\bibitem[Kuehnel et~al.(2023)Kuehnel, Rewald, Axer, Zwick, and
  Findeisen]{Kuehnel2023}
Nico Kuehnel, Hannes Rewald, Steffen Axer, Felix Zwick, and Rolf Findeisen.
\newblock Flow-inflated selective sampling for efficient agent-based dynamic
  ride-pooling simulations.
\newblock \emph{Transportation Research Record: Journal of the Transportation
  Research Board}, page 036119812311706, 2023.
\newblock ISSN 0361-1981.
\newblock \doi{10.1177/03611981231170624}.

\bibitem[Laupichler and Sanders(2023)]{Laupichler2023}
Moritz Laupichler and Peter Sanders.
\newblock Fast many-to-many routing for ridesharing with multiple pickup and
  dropoff locations.
\newblock 2023.
\newblock \doi{10.48550/arXiv.2305.05417}.

\bibitem[Li et~al.(2016)Li, Qin, Yu, and Mao]{Li2016}
Rong-Hua Li, Lu~Qin, Jeffrey~Xu Yu, and Rui Mao.
\newblock Optimal multi-meeting-point route search.
\newblock \emph{IEEE Transactions on Knowledge and Data Engineering},
  28:\penalty0 770--784, 2016.
\newblock ISSN 1041-4347.
\newblock \doi{10.1109/TKDE.2015.2492554}.

\bibitem[Lin et~al.(2012)Lin, Li, Qiu, and Xu]{Lin2012}
Yeqian Lin, Wenquan Li, Feng Qiu, and He~Xu.
\newblock Research on optimization of vehicle routing problem for ride-sharing
  taxi.
\newblock \emph{Procedia - Social and Behavioral Sciences}, 43:\penalty0
  494--502, 2012.
\newblock ISSN 18770428.
\newblock \doi{10.1016/j.sbspro.2012.04.122}.

\bibitem[Lotze et~al.(2022)Lotze, Marszal, Schröder, and Timme]{Lotze2022}
Charlotte Lotze, Philip Marszal, Malte Schröder, and Marc Timme.
\newblock Dynamic stop pooling for flexible and sustainable ride sharing.
\newblock \emph{New Journal of Physics}, 24:\penalty0 023034, 2022.
\newblock ISSN 1367-2630.
\newblock \doi{10.1088/1367-2630/ac47c9}.

\bibitem[Ma et~al.(2013)Ma, Zheng, and Wolfson]{Shuo2013}
Shuo Ma, Yu~Zheng, and Ouri Wolfson.
\newblock {T-Share}: A large-scale dynamic taxi ridesharing service.
\newblock In \emph{2013 IEEE 29th International Conference on Data Engineering
  (ICDE)}, pages 410--421. IEEE, 2013.
\newblock ISBN 978-1-4673-4910-9.
\newblock \doi{10.1109/ICDE.2013.6544843}.

\bibitem[Ma et~al.(2015)Ma, Zheng, and Wolfson]{Ma2015}
Shuo Ma, Yu~Zheng, and Ouri Wolfson.
\newblock Real-time city-scale taxi ridesharing.
\newblock \emph{IEEE Transactions on Knowledge and Data Engineering},
  27:\penalty0 1782--1795, 2015.
\newblock ISSN 1041-4347.
\newblock \doi{10.1109/TKDE.2014.2334313}.

\bibitem[Ma et~al.(2019)Ma, Rasulkhani, Chow, and Klein]{Ma2019}
Tai~Yu Ma, Saeid Rasulkhani, Joseph~Y.J. Chow, and Sylvain Klein.
\newblock A dynamic ridesharing dispatch and idle vehicle repositioning
  strategy with integrated transit transfers.
\newblock \emph{Transportation Research Part E: Logistics and Transportation
  Review}, 128:\penalty0 417--442, 2019.
\newblock ISSN 13665545.
\newblock \doi{10.1016/j.tre.2019.07.002}.

\bibitem[Madsen et~al.(1995)Madsen, Ravn, and Rygaard]{Madsen1995}
Oli B.~G. Madsen, Hans~F. Ravn, and Jens~Moberg Rygaard.
\newblock A heuristic algorithm for a dial-a-ride problem with time windows,
  multiple capacities, and multiple objectives.
\newblock \emph{Annals of Operations Research}, 60:\penalty0 193--208, 1995.
\newblock ISSN 0254-5330.
\newblock \doi{10.1007/BF02031946}.

\bibitem[Manna and Prestwich(2014)]{Manna2014}
Carlo Manna and Steve Prestwich.
\newblock Online stochastic planning for taxi and ridesharing.
\newblock In \emph{Proceedings - International Conference on Tools with
  Artificial Intelligence, ICTAI}, volume 2014-December, pages 906--913. IEEE
  Computer Society, 2014.
\newblock ISBN 9781479965724.
\newblock \doi{10.1109/ICTAI.2014.138}.

\bibitem[Milakis et~al.(2017)Milakis, van Arem, and van Wee]{Milakis2017}
Dimitris Milakis, Bart van Arem, and Bert van Wee.
\newblock Policy and society related implications of automated driving: A
  review of literature and directions for future research.
\newblock \emph{Journal of Intelligent Transportation Systems}, 21:\penalty0
  324--348, 2017.
\newblock ISSN 1547-2450.
\newblock \doi{10.1080/15472450.2017.1291351}.

\bibitem[Mounesan et~al.(2021)Mounesan, Jayawardana, Wu, Samaranayake, and
  Vo]{Mounesan2021}
Motahare Mounesan, Vindula Jayawardana, Yaocheng Wu, Samitha Samaranayake, and
  Huy~T. Vo.
\newblock Fleet management for ride-pooling with meeting points at scale: a
  case study in the five boroughs of {New York City}.
\newblock 2021.
\newblock \doi{10.48550/arXiv.2105.00994}.

\bibitem[Ota et~al.(2017)Ota, Vo, Silva, and Freire]{Ota2017}
Masayo Ota, Huy Vo, Claudio Silva, and Juliana Freire.
\newblock {STaRS}: Simulating taxi ride sharing at scale.
\newblock \emph{IEEE Transactions on Big Data}, 3:\penalty0 349--361, 2017.
\newblock ISSN 2332-7790.
\newblock \doi{10.1109/TBDATA.2016.2627223}.

\bibitem[Pelzer et~al.(2015)Pelzer, Xiao, Zehe, Lees, Knoll, and
  Aydt]{Pelzer2015}
Dominik Pelzer, Jiajian Xiao, Daniel Zehe, Michael~H. Lees, Alois~C. Knoll, and
  Heiko Aydt.
\newblock A partition-based match making algorithm for dynamic ridesharing.
\newblock \emph{IEEE Transactions on Intelligent Transportation Systems},
  16:\penalty0 2587--2598, 2015.
\newblock ISSN 1524-9050.
\newblock \doi{10.1109/TITS.2015.2413453}.

\bibitem[Psaraftis(1980)]{Psaraftis1980}
Harilaos~N. Psaraftis.
\newblock A dynamic programming solution to the single vehicle many-to-many
  immediate request dial-a-ride problem.
\newblock \emph{Transportation Science}, 14:\penalty0 130--154, 1980.
\newblock ISSN 0041-1655.
\newblock \doi{10.1287/trsc.14.2.130}.

\bibitem[Santos and Xavier(2015)]{Santos2015}
Douglas~O. Santos and Eduardo~C. Xavier.
\newblock Taxi and ride sharing: A dynamic dial-a-ride problem with money as an
  incentive.
\newblock \emph{Expert Systems with Applications}, 42:\penalty0 6728--6737,
  2015.
\newblock ISSN 09574174.
\newblock \doi{10.1016/j.eswa.2015.04.060}.

\bibitem[Savelsbergh(1985)]{Savelsbergh1985}
Martin Savelsbergh.
\newblock Local search in routing problems with time windows.
\newblock \emph{Annals of Operations Research}, 4:\penalty0 285--305, 1985.
\newblock ISSN 0254-5330.
\newblock \doi{10.1007/BF02022044}.

\bibitem[Schilde et~al.(2011)Schilde, Doerner, and Hartl]{Schilde2011}
Michael Schilde, Karl~F. Doerner, and Richard~F. Hartl.
\newblock Metaheuristics for the dynamic stochastic dial-a-ride problem with
  expected return transports.
\newblock \emph{Computers and Operations Research}, 38:\penalty0 1719--1730,
  2011.
\newblock ISSN 03050548.
\newblock \doi{10.1016/j.cor.2011.02.006}.

\bibitem[Song et~al.(2021)Song, Monteil, Ygnace, and Rey]{Song2021}
Changle Song, Julien Monteil, Jean-Luc Ygnace, and David Rey.
\newblock Incentives for ridesharing: A case study of welfare and traffic
  congestion.
\newblock \emph{Journal of Advanced Transportation}, 2021.
\newblock ISSN 0197-6729.
\newblock \doi{10.1155/2021/6627660}.

\bibitem[Stiglic et~al.(2015)Stiglic, Agatz, Savelsbergh, and
  Gradisar]{Stiglic2015}
Mitja Stiglic, Niels Agatz, Martin Savelsbergh, and Mirko Gradisar.
\newblock The benefits of meeting points in ride-sharing systems.
\newblock \emph{Transportation Research Part B: Methodological}, 82:\penalty0
  36--53, 2015.
\newblock ISSN 01912615.
\newblock \doi{10.1016/j.trb.2015.07.025}.

\bibitem[Tao and Wu(2008)]{Chichung2008}
Chichung Tao and Chungjung Wu.
\newblock Behavioral responses to dynamic ridesharing services - the case of
  taxi-sharing project in {Taipei}.
\newblock In \emph{2008 IEEE International Conference on Service Operations and
  Logistics, and Informatics}, pages 1576--1581. IEEE, 2008.
\newblock ISBN 978-1-4244-2012-4.
\newblock \doi{10.1109/SOLI.2008.4682777}.

\bibitem[Weckström et~al.(2018)Weckström, Mladenović, Ullah, Nelson, Givoni,
  and Bussman]{Weckstroem2018}
Christoffer Weckström, Miloš~N. Mladenović, Waqar Ullah, John~D. Nelson,
  Moshe Givoni, and Sebastian Bussman.
\newblock User perspectives on emerging mobility services: Ex post analysis of
  {Kutsuplus} pilot.
\newblock \emph{Research in Transportation Business \& Management},
  27:\penalty0 84--97, 2018.
\newblock ISSN 22105395.
\newblock \doi{10.1016/j.rtbm.2018.06.003}.

\bibitem[Wilkes et~al.(2021)Wilkes, Engelhardt, Briem, Dandl, Vortisch,
  Bogenberger, and Kagerbauer]{Wilkes2021}
Gabriel Wilkes, Roman Engelhardt, Lars Briem, Florian Dandl, Peter Vortisch,
  Klaus Bogenberger, and Martin Kagerbauer.
\newblock Self-regulating demand and supply equilibrium in joint simulation of
  travel demand and a ride-pooling service.
\newblock \emph{Transportation Research Record: Journal of the Transportation
  Research Board}, 2675:\penalty0 226--239, 2021.
\newblock ISSN 0361-1981.
\newblock \doi{10.1177/0361198121997140}.

\bibitem[Yanagisawa(2010)]{Yanagisawa2010}
Hiroki Yanagisawa.
\newblock A multi-source label-correcting algorithm for the all-pairs shortest
  paths problem.
\newblock In \emph{2010 IEEE International Symposium on Parallel \& Distributed
  Processing (IPDPS)}, pages 1--10. IEEE, 2010.
\newblock ISBN 978-1-4244-6442-5.
\newblock \doi{10.1109/IPDPS.2010.5470362}.

\bibitem[Yu et~al.(2017)Yu, Ma, Xue, Tang, Wang, Yan, and Wei]{Yu2017}
Biying Yu, Ye~Ma, Meimei Xue, Baojun Tang, Bin Wang, Jinyue Yan, and Yi-Ming
  Wei.
\newblock Environmental benefits from ridesharing: A case of {Beijing}.
\newblock \emph{Applied Energy}, 191:\penalty0 141--152, 2017.
\newblock ISSN 03062619.
\newblock \doi{10.1016/j.apenergy.2017.01.052}.

\bibitem[Zhan et~al.(2022)Zhan, Szeto, and Chen]{Zhan2022}
Xingbin Zhan, W.Y. Szeto, and Xiqun~Michael Chen.
\newblock The dynamic ride-hailing sharing problem with multiple vehicle types
  and user classes.
\newblock \emph{Transportation Research Part E: Logistics and Transportation
  Review}, 168, 2022.
\newblock ISSN 13665545.
\newblock \doi{10.1016/j.tre.2022.102891}.

\bibitem[Zhu(2021)]{Zhu2021}
Dianzhuo Zhu.
\newblock The limits of money in daily ridesharing: Evidence from a field
  experiment in rural {France}.
\newblock \emph{Revue d'économie industrielle}, pages 161--202, 2021.
\newblock ISSN 0154-3229.
\newblock \doi{10.4000/rei.9984}.

\bibitem[Ziemke et~al.(2019)Ziemke, Kaddoura, and Nagel]{Ziemke2019}
Dominik Ziemke, Ihab Kaddoura, and Kai Nagel.
\newblock The {MATSim} {Open} {Berlin} scenario: A multimodal agent-based
  transport simulation scenario based on synthetic demand modeling and open
  data.
\newblock \emph{Procedia Computer Science}, 151, 2019.
\newblock ISSN 1877-0509.
\newblock \doi{doi.org/10.1016/j.procs.2019.04.120}.

\bibitem[Zwick et~al.(2022)Zwick, Wilkes, Engelhardt, Axer, Dandl, Rewald,
  Kostorz, Fraedrich, Kagerbauer, and Axhausen]{Zwick2022}
Felix Zwick, Gabriel Wilkes, Roman Engelhardt, Steffen Axer, Florian Dandl,
  Hannes Rewald, Nadine Kostorz, Eva Fraedrich, Martin Kagerbauer, and Kay~W.
  Axhausen.
\newblock Mode choice and ride-pooling simulation: A comparison of {mobiTopp},
  {Fleetpy}, and {MATSim}.
\newblock \emph{Procedia Computer Science}, 201:\penalty0 608--613, 2022.
\newblock ISSN 18770509.
\newblock \doi{10.1016/j.procs.2022.03.079}.

\end{thebibliography}
